\documentclass[10pt,letterpaper,twoside,reqno,final]{amsart}
\pdfoutput=1

\makeatletter{}
\usepackage{fixltx2e}                       
\usepackage[usenames,dvipsnames]{xcolor}
\usepackage{fancyhdr}
\usepackage{amsmath,amsfonts,amsbsy,amsgen,amscd,mathrsfs,amssymb,amsthm}

\usepackage{subfig}
\usepackage{url}

\usepackage[font=small,margin=0.25in,labelfont={sc},labelsep={colon}]{caption}

\usepackage{tikz}
\usepackage{microtype}
\usepackage{enumitem}

\definecolor{dark-gray}{gray}{0.3}
\definecolor{dkgray}{rgb}{.4,.4,.4}
\definecolor{dkblue}{rgb}{0,0,.5}
\definecolor{medblue}{rgb}{0,0,.75}
\definecolor{rust}{rgb}{0.5,0.1,0.1}

\usepackage[colorlinks=true]{hyperref}
\hypersetup{pdftitle={Living on the edge: Phase transitions in convex programs with random data}}    \hypersetup{pdfauthor={Amelunxen, Lotz, McCoy, and Tropp}}     \hypersetup{pdfkeywords={convex optimization} {integral geometry} {conic geometry} {probability} {intrinsic volumes}}  \hypersetup{linkcolor=dkblue}          \hypersetup{citecolor=rust}        \hypersetup{urlcolor=rust}

\usepackage{graphicx}
\usepackage{booktabs,longtable,tabu} 
\usepackage{multirow} 
\usepackage{float}

\usepackage[T1]{fontenc}

\usepackage{fourier}
\usepackage{charter}

\usepackage{bm} 
\graphicspath{{figures/}}

\usepackage{pdfsync}

\newtheorem{bigthm}{Theorem}

\newtheorem{theorem}{Theorem}[section]
\newtheorem{lemma}[theorem]{Lemma}
\newtheorem{sublemma}[theorem]{Sublemma}
\newtheorem{proposition}[theorem]{Proposition}
\newtheorem{fact}[theorem]{Fact}

\newtheorem{conjecture}[theorem]{Conjecture}

\theoremstyle{definition}

\newtheorem{definition}[theorem]{Definition}
\newtheorem{example}[theorem]{Example}
\newtheorem{remark}[theorem]{Remark}

\newcommand{\term}{\emph}

\numberwithin{equation}{section} 
\numberwithin{figure}{section}
\numberwithin{table}{section}

\floatstyle{plaintop}
\newfloat{recipe}{thp}{lor}
\floatname{recipe}{Recipe}
\numberwithin{recipe}{section}

\providecommand{\mathbold}[1]{\bm{#1}}

\renewcommand{\phi}{\varphi}

\newcommand{\eps}{\varepsilon}

\newcommand{\half}{\tfrac{1}{2}}

\newcommand{\defeq}{\ensuremath{\mathrel{\mathop{:}}=}}  
 
\newcommand{\econst}{\mathrm{e}}

 \newcommand{\zerovct}{\vct{0}} 
\newcommand{\Id}{\mathbf{I}}

\newcommand{\coll}[1]{\mathscr{#1}}
\newcommand{\sphere}[1]{\mathsf{S}^{#1}}
\newcommand{\ball}[1]{\mathsf{B}^{#1}}

\providecommand{\mathbbm}{\mathbb} 
\newcommand{\R}{\mathbbm{R}}

\newcommand{\polar}{\circ}

\newcommand{\abs}[1]{\left\vert {#1} \right\vert}

\newcommand{\pos}{\operatorname{Pos}}

\newcommand{\diff}[1]{\mathrm{d}{#1}}
\newcommand{\idiff}[1]{\, \diff{#1}}

\newcommand{\argmin}{\operatorname*{arg\; min}}

\newcommand{\Prob}{\mathbbm{P}}

\newcommand{\Expect}{\operatorname{\mathbb{E}}}

\newcommand{\normal}{\textsc{normal}}

\DeclareMathOperator{\Var}{Var}

\newcommand{\vct}[1]{\mathbold{#1}}
\newcommand{\mtx}[1]{\mathbold{#1}}

\newcommand{\transp}{T}

\newcommand{\nullity}{\operatorname{null}}

\renewcommand{\vec}{\operatorname{vec}}

\newcommand{\Proj}{\ensuremath{\mtx{\Pi}}} 

\newcommand{\psdge}{\succcurlyeq}

\newcommand{\ip}[2]{\left\langle {#1},\ {#2} \right\rangle}
\newcommand{\absip}[2]{\abs{\ip{#1}{#2}}}

\newcommand{\norm}[1]{\left\Vert {#1} \right\Vert}
\newcommand{\normsq}[1]{\norm{#1}^2}

\newcommand{\sone}[1]{\norm{#1}_{S_1}}
\newcommand{\snorm}[1]{\sone{#1}}
\DeclareMathOperator{\dist}{dist}

\newcommand{\enorm}[1]{\norm{#1}}
\newcommand{\enormsm}[1]{\enorm{\smash{#1}}}

\newcommand{\enormsq}[1]{\enorm{#1}^2}

\newcommand{\fnorm}[1]{\norm{#1}_{\mathrm{F}}}
\newcommand{\fnormsq}[1]{\fnorm{#1}^2}

\newcommand{\pnorm}[2]{\norm{#2}_{#1}}

\newcommand{\Desc}{\mathcal{D}}
\newcommand{\NormC}{\mathcal{N}} \newcommand{\sdim}{\delta}

\DeclareMathOperator{\Circ}{Circ}

\newcommand{\uniform}{\textsc{uniform}}

\newcommand{\cone}{\operatorname{cone}}

\newcommand{\conv}{\operatorname{conv}}

\newcommand{\minimize}{\text{minimize}\quad}
\newcommand{\subjto}{\quad\text{subject to}\quad}

\newcommand{\Rmm}{R_{\mathrm{mm}}}

\newcommand{\distsubdiff}{J}

\title[Phase transitions in random convex programs]{Living on the edge: Phase transitions \\
in convex programs with random data}

\author[D.~Amelunxen]{Dennis Amelunxen}

\author[M.~Lotz]{Martin Lotz}

\author[M.~B.~McCoy]{Michael B. Mccoy}
 \author[J.~A.~Tropp]{Joel~A.~Tropp}

\date{26 March 2013. Revised 26 January 2014 and 24 April 2014.}

\subjclass[2010]{Primary:
90C25, 52A22, 60D05. Secondary:
52A20, 62C20. }

\begin{document}

\begin{abstract} Recent research indicates that many convex optimization problems with random constraints exhibit a phase transition as the number of constraints increases.  For example, this phenomenon emerges in the $\ell_1$ minimization method for identifying a sparse vector from random linear measurements.  Indeed, the $\ell_1$ approach succeeds with high probability when the number of measurements exceeds a threshold that depends on the sparsity level; otherwise, it fails with high probability.

This paper provides the first rigorous analysis that explains why phase transitions are ubiquitous in random convex optimization problems.  It also describes tools for making reliable predictions about the quantitative aspects of the transition, including the location and the width of the transition region.  These techniques apply to regularized linear inverse problems with random measurements, to demixing problems under a random incoherence model, and also to cone programs with random affine constraints.

The applied results depend on foundational research in conic geometry.  This paper introduces a summary parameter, called the statistical dimension, that canonically extends the dimension of a linear subspace to the class of convex cones.  The main technical result demonstrates that the sequence of intrinsic volumes of a convex cone concentrates sharply around the statistical dimension.  This fact leads to accurate bounds on the probability that a randomly rotated cone shares a ray
with a fixed cone.
\end{abstract}

\maketitle

\thispagestyle{empty}

\section{Motivation}
\label{sec:introduction}

A \term{phase transition} is a sharp change in the character of a computational problem as its parameters vary.  Recent research suggests that phase transitions emerge in many
random convex optimization problems from mathematical signal processing and computational statistics; for example, see~\cite{DonTan:09a,stojnic10, OH:10, CSPW2011, DonGavMon:13, mctr:12}.
This paper proves that the locations of these phase transitions are determined by geometric invariants associated with the mathematical programs.  Our analysis provides the first complete account of transition phenomena in random linear inverse problems, random demixing problems, and random cone programs.

\subsection{Vignette: Compressed sensing}
\label{sec:cs}

To illustrate our goals, we discuss the \term{compressed sensing problem}, a familiar example where a phase transition is plainly visible in numerical experiments~\cite{DonTan:09a}.  Let $\vct{x}_0 \in \R^d$ be an unknown vector with $s$ nonzero entries.  Let $\mtx{A}$ be an $m \times d$ random matrix whose entries are independent standard normal variables, and suppose we have access to the vector 
\begin{equation}\label{eq:l1-obs}
  \vct{z}_0 = \mtx{A} \vct{x}_0 \in \R^m.
\end{equation}
This serves as a model for data acquisition: we interpret $\vct{z}_0$ as a collection of $m$ independent linear measurements of the unknown $\vct{x}_0$.  The compressed sensing problem requires us to identify $\vct{x}_0$ given only the measurement vector $\vct{z}_0$ and the realization of the measurement matrix $\mtx{A}$.  When the number $m$ of measurements is smaller than the ambient dimension $d$, we cannot solve this inverse problem unless we take advantage of the prior knowledge that $\vct{x}_0$ is sparse.

The method of $\ell_1$ minimization~\cite{CDS:01,CT:06,dono:06} is a well-established approach to the compressed sensing problem.  This technique searches for the sparse unknown $\vct{x}_0$ by solving the convex program
\begin{equation} \label{eqn:l1-min}
\minimize \pnorm{1}{ \vct{x} }
\subjto \vct{z}_0 = \mtx{A}\vct{x},
\end{equation}
where $\pnorm{1}{\vct x}:= \sum_{i=1}^d\abs{x_i}$.  This approach is sensible because the $\ell_1$ norm of a vector can serve as a proxy for the sparsity.  We say that~\eqref{eqn:l1-min} \term{succeeds} at solving the compressed sensing problem when it has a unique optimal point $\widehat{\vct{x}}$ and $\widehat{\vct{x}}$ equals the true unknown $\vct{x}_0$; otherwise, it \term{fails}.

Figure~\ref{fig:l1-min} depicts the results of a computer experiment designed to estimate the probability that~\eqref{eqn:l1-min} succeeds as we vary the sparsity $s$ of the unknown $\vct{x}_0$ and the number $m$ of random measurements.  We consider two choices for the ambient dimension, $d = 100$ and $d = 600$.  For each choice of $s$ and $m$, we construct a vector $\vct{x}_0$ with $s$ nonzero entries, we draw $m$ random measurements according to the model~\eqref{eq:l1-obs}, and we solve the problem~\eqref{eqn:l1-min}.  The brightness of the point $(s, m)$ indicates the probability of success, estimated from $50$ independent trials.  White represents certain success, while black represents certain failure.  The plot evinces that, for a given sparsity level $s$, the $\ell_1$ minimization technique~\eqref{eqn:l1-min} almost always succeeds when we have an adequate number $m$ of measurements, while it almost always fails when we have fewer measurements.  Appendix~\ref{app:experiments} contains more details about this experiment.

\begin{figure}[t!]
  \centering
  \includegraphics[width=0.9\textwidth]{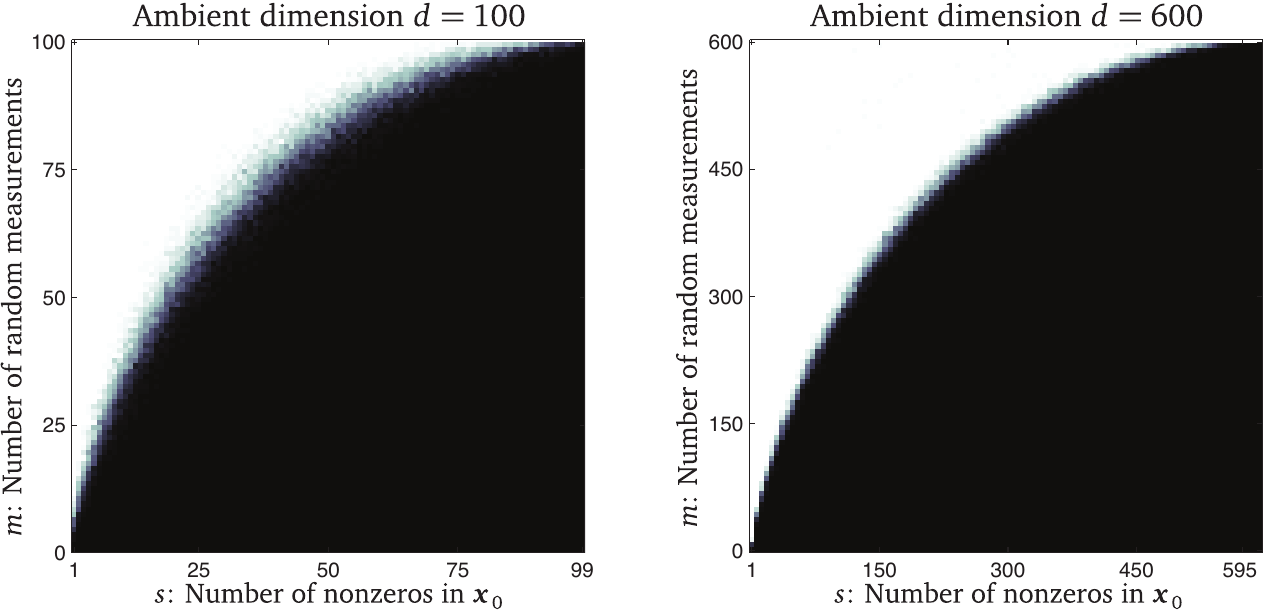}
  \caption{\textbf{The phase transition phenomenon in compressed sensing.}
  This diagram shows the empirical probability that the $\ell_1$ minimization
  method~\eqref{eqn:l1-min} successfully identifies a vector
  $\vct{x}_0 \in \R^d$ with $s$ nonzero entries given
  a vector $\vct{z}_0$ consisting of $m$ random measurements of the form
  $\vct{z}_0 = \mtx{A} \vct{x}_0$ where $\mtx{A}$ is an $m \times d$ matrix
  with independent standard normal entries.
  The brightness of each point reflects the observed probability of success,
  ranging from certain failure (black) to certain success (white).
  \textbf{[left]} The ambient dimension $d = 100$.
  \textbf{[right]} The ambient dimension $d = 600$.}
    \label{fig:l1-min}
\end{figure}

Figure~\ref{fig:l1-min} raises several interesting questions about the
performance of the $\ell_1$ minimization method~\eqref{eqn:l1-min} for solving the compressed sensing problem:

\begin{itemize} \setlength{\itemsep}{1mm}
\item	\textbf{What is the probability of success?}  For a given pair $(s, m)$ of parameters, can we estimate the probability that~\eqref{eqn:l1-min} succeeds or fails?

\item	\textbf{Does a phase transition exist?}  Is there a simple curve $m = \psi(s)$ that separates the parameter space into regions where~\eqref{eqn:l1-min} is very likely to succeed or to fail?

\item	\textbf{Where is the edge of the phase transition?}  Can we find a formula for the location of this threshold between success and failure?

\item	\textbf{How wide is the transition region?}  For a given sparsity level $s$ and ambient dimension $d$, how big is the range of $m$ where the probability of success and failure are comparable?

\item	\textbf{Why does the transition exist?}  Is there a geometric explanation for the phase transition in compressed sensing?  Can we export this reasoning to understand other problems?
\end{itemize}

\noindent
There is an extensive body of work dedicated to these questions and their relatives.
See the books~\cite{EldKut:12,FouRau:13}
for background on compressed sensing in general.
Section~\ref{sec:conclusion} outlines the current state of knowledge about phase
transitions in convex optimization methods for signal processing.
In spite of all this research, a complete explanation of these phenomena is lacking.
The goal of this paper is to answer the questions we have posed.

\subsection{Notation}
\label{sec:notation}
Before moving forward, let us introduce some notation.  We use standard conventions from convex analysis,
as set out in Rockafellar~\cite{Rock}.
For vectors $\vct{x}, \vct{y} \in \R^d$, we define the Euclidean inner product $\ip{\vct{x}}{\smash{\vct{y}}} := \sum_{i=1}^d x_i y_i$ and the squared Euclidean norm $\enormsq{\vct{x}} := \ip{\vct{x}}{\vct{x}}$.  The Euclidean distance to a set $S \subset \R^d$ is the function
$$
\dist( \cdot, S) : \R^d \to \R_+
\quad\text{where}\quad
\dist(\vct{x}, S) := \inf\big\{ \enorm{\smash{\vct{x} - \vct{y}}} : \vct{y} \in S \big\}
$$
and $\R_+$ denotes the set of nonnegative real numbers.
The unit ball $\ball{d}$ and unit sphere $\sphere{d-1}$ in $\R^d$ are the sets
$$
\ball{d} := \big\{ \vct{x} \in \R^d : \enorm{\vct{x}} \leq 1 \big\}
\quad\text{and}\quad
\sphere{d-1} := \big\{ \vct{x} \in \R^d : \enorm{\vct{x}} = 1 \big\}.
$$
An \term{orthogonal basis} $\mtx U \in \R^{d\times d}$ for $\R^d$ is a matrix that satisfies $\mtx U^\transp \mtx U = \Id$, where $\Id$ is the identity.

A \term{convex cone} $C \subset \R^d$ is a convex set that is
positively homogeneous: $C = \tau C$ for all $\tau > 0$.
Let us emphasize that the vertex of a convex cone is always located at the origin.
For a closed convex cone $C$, the Euclidean projection $\Proj_C(\vct{x})$ of a point $\vct{x}$ onto the cone returns the point in $C$ nearest to $\vct x$:
\begin{equation}\label{eq:cone-proj}
\mtx{\Pi}_C(\vct{x}) := \argmin\big\{ \enormsm{ \vct{x} - \vct{y} } : \vct{y} \in C \big\}.
\end{equation}
For a general cone $C\subset \R^d$,
the \term{polar cone} $C^\polar$ is the set of outward normals of $C$:
\begin{equation} \label{eq:polar-cone}
C^\polar := \big\{ \vct{u} \in \R^d : \ip{\vct{u}}{\vct{x}} \leq 0
\quad\text{for all $\vct{x} \in C$} \big\}.
\end{equation}
The polar cone $C^\polar$ is always closed and convex.

We make heavy use of probability in this work.  The symbol $\Prob\{\cdot\}$ denotes the probability of an event, and
$\Expect[ \cdot ]$ returns the expectation of a random variable.
We reserve the letter $\vct{g}$ for a standard normal random vector, i.e.,
a vector whose entries are independent Gaussian random variables with mean
zero and variance one.  We reserve the letter $\vct{\theta}$ for
a random vector uniformly distributed on the Euclidean unit
sphere.  The set of orthogonal matrices forms a compact Lie group,
so it admits an invariant Haar (i.e., uniform) probability measure.
We reserve the letter $\mtx{Q}$ for a uniformly random orthogonal matrix,
and we refer to $\mtx{Q}$ as a \term{random orthogonal basis}
or a \term{random rotation}.

\section{Conic geometry and phase transitions}
\label{sec:conic-phase}

In the theory of convex analysis, convex cones take over the central role that subspaces perform in linear algebra~\cite[p.~90]{HUL:93a}.  In particular, we can use convex cones to express the optimality conditions for a convex program~\cite[Part VII]{HUL:93a}.  When a convex optimization problem includes random data, the optimality conditions may involve \term{random} convex cones.  Therefore, the study of random convex optimization problems leads directly to questions about the stochastic geometry of cones.

This perspective is firmly established in the literature on convex optimization for signal processing applications.
Rudelson \& Vershynin~\cite[Sec.~4]{RV:08} analyze the $\ell_1$ minimization method~\eqref{eqn:l1-min} for the
compressed sensing problem by examining the conic formulation of the optimality conditions.  They apply deep
results~\cite{gord:85,gord:88} for Gaussian processes to bound the probability that $\ell_1$ minimization
succeeds.  Many subsequent papers, including~\cite{stojnic10,OH:10,CRPW:12}, rely on the same argument.

In sympathy with these prior works, we study random convex optimization problems
by considering the conic formulation of the optimality conditions.
In contrast, we have developed a new technical argument to study
the probability that the conic optimality conditions hold.
Our approach depends on exact formulas from the field of
\term{conic integral geometry}~\cite[Chap.~6.5]{scwe:08}.
In this context, the general idea of using integral geometry
is due to Donoho~\cite{dono:06b} and Donoho \& Tanner~\cite{dota:09a}.
The specific method in this paper was proposed in~\cite{mctr:12},
but we need to install additional machinery to
prove that phase transitions occur.

Sections~\ref{sec:kinem-intro}--\ref{sec:calc-sdim-intro} outline the results we need from conic integral geometry, along with our contributions to this subject.  We apply this theory in Sections~\ref{sec:line-inverse-probl} and~\ref{sec:phase-trans-demix} to study some random optimization problems.  We conclude with a summary of our main results in Section~\ref{sec:contributions}.

\subsection{The kinematic formula for cones}
\label{sec:kinem-intro}

Let us begin with a beautiful and classical problem from the field of conic integral geometry:

\begin{quotation} \it
What is the probability that a randomly rotated convex cone shares a ray with a fixed convex cone?
\end{quotation}

\noindent
See Figure~\ref{fig:two-cones} for an illustration of the geometry.  Formally, we consider convex cones $C$ and $K$ in $\R^d$, and we draw a random orthogonal basis $\mtx{Q} \in \R^{d \times d}$.  The goal is to find a useful expression for the probability
$$
\Prob\big\{ C \cap \mtx{Q} K \neq \{ \vct{0} \} \big\}.
$$
As we will discuss, this is the key question we must answer to understand phase transition phenomena in convex optimization problems with random data.

In two dimensions, we quickly determine the solution to the problem.
Consider two convex cones $C$ and $K$ in $\R^2$.  If neither cone is a linear subspace, then
$$
\Prob\big\{ C \cap \mtx{Q} K \neq \{ \vct{0} \} \big\}
	= \min\big\{ v_2(C) + v_2(K), \ 1 \big\},
$$
where $v_2(\cdot)$ returns the proportion of the unit circle subtended by (the closure of) a convex cone in $\R^2$.  A similar formula holds when one of the cones is a subspace.
In higher dimensions, however, convex cones can be complicated objects.  In three dimensions, the question already starts to look difficult, and we might despair that a reasonable solution exists in general.

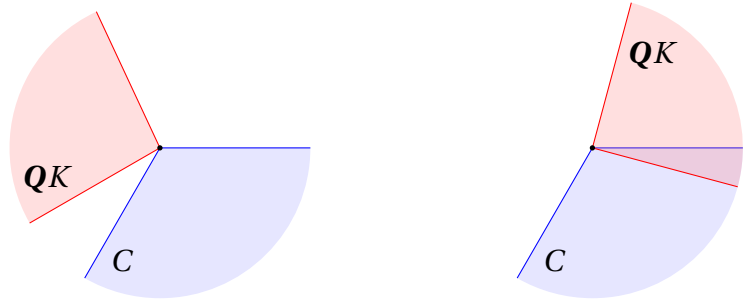
\begin{figure}
\begin{center}
\newcommand{\conelen}{2}
\newcommand{\conelend}{2}

\begin{tikzpicture}[scale=1]
\fill[white!90!blue] (0,0) -- ++(0:\conelend) arc(0:-120:\conelend) -- cycle;
\draw[thin,blue] (0,0) -- ++(0:\conelend) (0,0) -- ++(-120:\conelend);
\path (-0.5, -1.8) node[above, scale=1.25]{$C$};

\fill[white!50!red, nearly transparent] (0,0) -- ++(210:\conelen) arc(210:115:\conelen) -- cycle;
\draw[thin,red] (0,0) -- ++(210:\conelen) (0,0) -- ++(115:\conelen);
\path (-1.5, -0.75) node[above, scale=1.25]{$\mtx{Q} K$};

\fill[black] (0,0) circle (1pt);
\end{tikzpicture}
\hspace{1in}
\begin{tikzpicture}[scale=1]
\fill[white!90!blue] (0,0) -- ++(0:\conelend) arc(0:-120:\conelend) -- cycle;
\draw[thin,blue] (0,0) -- ++(0:\conelend) (0,0) -- ++(-120:\conelend);
\path (-0.5, -1.8) node[above, scale=1.25]{$C$};

\fill[white!50!red,nearly transparent] (0,0) -- ++(75:\conelen) arc(75:-15:\conelen) -- cycle;
\draw[thin,red] (0,0) -- ++(75:\conelen) (0,0) -- ++(-15:\conelen);
\path (0.8, 0.9) node[above, scale=1.25]{$\mtx{Q} K$};

\fill[black] (0,0) circle (1pt);
\end{tikzpicture}
\end{center}
\caption{\textbf{Randomly rotated convex cones.}  Let $C$ and $K$ be nontrivial convex cones, and let $\mtx{Q}$ be a random orthogonal basis.  The cone $C$ is fixed, and $\mtx{Q} K$ is a randomly rotated copy of $K$.  \textbf{[left]} The two cones have a trivial intersection.  \textbf{[right]} The two cones share a ray.}
\label{fig:two-cones}
\end{figure}

It turns out that there is an \emph{exact} formula for the probability that a randomly rotated convex cone shares a ray with a fixed convex cone.
Moreover, in $d$ dimensions, we only need $d + 1$ numbers to summarize each cone.  This wonderful result is called the \term{conic kinematic formula}~\cite[Thm.~6.5.6]{scwe:08}.  We record the statement here, but you should not focus on the details at this stage; Section~\ref{sec:intr-conic-integr} contains a more thorough presentation.

\begin{fact}[The kinematic formula for cones] \label{fact:kinem-intro}
Let $C$ and $K$ be closed convex cones in $\R^d$, one of which is not a subspace.
Draw a random orthogonal basis $\mtx{Q} \in \R^{d\times d}$.
Then
$$
\Prob\big\{ C \cap \mtx{Q} K \neq \{ \vct{0} \} \big\}
	= \sum_{i=0}^d \big( 1 + (-1)^{i+1} \big) \sum_{j=i}^d v_i(C) \cdot v_{d+i-j}(K).
$$
For each $k =0, 1, 2, \dots, d$, the geometric functional $v_k$ maps a closed convex cone to a nonnegative number, called the $k$th \term{intrinsic volume} of the cone. \end{fact}

The papers~\cite{ambu:11c,mctr:12} have recognized that the conic kinematic formula is tailor-made for studying random instances of convex optimization problems.  Unfortunately, this approach suffers a serious weakness: We do not have workable expressions for the intrinsic volumes of a cone, except in the simplest cases.
This paper provides a way to make the kinematic formula effective.  To explain, we need
to have a closer look at the conic intrinsic volumes.

\subsection{Concentration of intrinsic volumes and the statistical dimension}

The conic intrinsic volumes, introduced in Fact~\ref{fact:kinem-intro}, are the fundamental geometric invariants of a closed convex cone.  They do not depend on the dimension of the space in which the cone is embedded, nor on the orientation of the cone within that space.  For an analogy in Euclidean geometry, you may consider similar quantities defined for compact convex sets, such as the usual volume, the surface area, the mean width, and the Euler characteristic~\cite{Sch:93}.

We will provide a more rigorous treatment of the conic intrinsic volumes in Section~\ref{sec:intr-conic-integr}.
For now, the only formal property we need is that the intrinsic volumes of a closed convex cone $C$ in $\R^d$ compose a probability distribution on $\{0, 1, 2, \dots, d \}$.  That is,
$$
\sum_{k=0}^d v_k(C) = 1
\quad\text{and}\quad
v_k(C) \geq 0
\quad\text{for $k = 0, 1, 2, \dots, d$.}
$$
Figure~\ref{fig:int-vol-conc} displays the distribution of intrinsic volumes for a particular cone; you can see that the sequence has a sharp peak at its mean value.
Our work establishes a remarkable new fact about conic geometry:

\begin{quotation} \it
For every closed convex cone, the distribution of conic intrinsic volumes concentrates sharply around its mean value.
\end{quotation}

\noindent
This result is our main technical achievement; Theorem~\ref{thm:main-conc} contains a precise statement.

\begin{figure}
\includegraphics[height=2.5in]{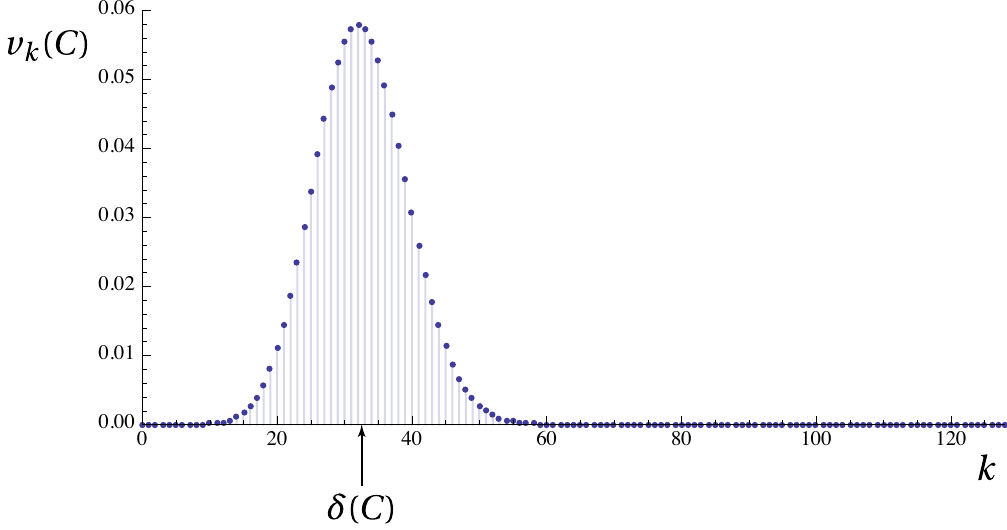}
\caption{\textbf{Concentration of conic intrinsic volumes.}
This plot displays the conic intrinsic volumes $v_k(C)$ of a circular cone $C \subset \R^{128}$ with angle $\pi/6$.  The distribution concentrates sharply around the statistical dimension $\delta(C) \approx 32.5$.  See Section~\ref{sec:circular-cones} for further discussion of this example.}
\label{fig:int-vol-conc}
\end{figure}

Because of the concentration phenomenon, the mean value of the distribution of conic intrinsic volumes serves as a summary for the entire distribution.  This insight leads to the central definition of the paper.

\begin{definition}[Statistical dimension: Intrinsic characterization] \label{def:sdim-int}
Let $C$ be a closed convex cone in $\R^d$.  The \term{statistical dimension} $\delta(C)$ of the cone is defined as
$$
\delta(C) := \sum_{k=0}^d k \, v_k(C).
$$
The statistical dimension of a general convex cone is the statistical dimension of its closure.
\end{definition}

\noindent
As the name suggests, the statistical dimension reflects the dimensionality of a convex cone.  Here are some properties that support this interpretation.  First, the statistical dimension increases with the size of a cone.  Indeed, for nested convex cones $C \subset K \subset \R^d$, we have the inequalities $0 \leq \delta(C) \leq \delta(K) \leq d$.  Second, the statistical dimension of a linear subspace $L$ always satisfies
\begin{equation} \label{eq:sdim-dim-intro}
\delta(L) = \dim(L).
\end{equation}
In fact, the statistical dimension is a \emph{canonical extension} of the dimension of a linear subspace to the class of convex cones!  Section~\ref{sec:sdim-canon} provides technical justification for the latter point, while Sections~\ref{sec:calc-stat-dimens},~\ref{sec:descent-cones}, and~\ref{sec:intr-conic-integr} establish various properties of the statistical dimension.

\subsection{The approximate kinematic formula}

We can simplify the conic kinematic formula, Fact~\ref{fact:kinem-intro},
by exploiting the concentration of intrinsic volumes.

\begin{bigthm}[Approximate kinematic formula] \label{thm:kinematic}
Fix a tolerance $\eta \in (0,1)$.  Let $C$ and $K$ be convex cones in $\R^d$, and draw a random orthogonal basis $\mtx{Q} \in \R^{d\times d}$.
Then
\begin{align*}
\delta(C) + \delta(K) \leq d - a_{\eta} \sqrt{d}
&\quad\Longrightarrow\quad
\Prob\big\{ C \cap \mtx{Q} K \neq \{\vct{0}\} \big\} \leq \eta; \\
\delta(C) + \delta(K) \geq d + a_{\eta} \sqrt{d}
&\quad\Longrightarrow\quad
\Prob\big\{ C \cap \mtx{Q} K \neq \{\vct{0}\} \big\} \geq 1 - \eta.
\end{align*}
The quantity $a_{\eta} := \sqrt{8 \log(4/\eta)}$.  For example, $a_{0.01} < 7$ and $a_{0.001} < 9$.
\end{bigthm}

\noindent
In Section~\ref{sec:kinem-cons-conc}, we derive Theorem~\ref{thm:kinematic} along with some more precise results.

Theorem~\ref{thm:kinematic} says that two randomly rotated cones are likely to share a ray if and only if the total statistical dimension of the two cones exceeds the ambient dimension.  This statement is in perfect sympathy with the analogous result for random subspaces.  We extract the following lesson:

\begin{quotation} \it
We can assign a dimension $\delta(C)$ to each convex cone $C$.  For problems in conic integral geometry, the cone behaves much like a subspace with approximate dimension $\delta(C)$.
\end{quotation}

\noindent
In Sections~\ref{sec:line-inverse-probl} and~\ref{sec:phase-trans-demix}, we use Theorem~\ref{thm:kinematic} to prove that a large class of random convex optimization problems always exhibits a phase transition, and we demonstrate that the statistical dimension describes the location of the phase transition.

\begin{remark}[Gaussian process theory]
If we replace the random orthogonal basis in Theorem~\ref{thm:kinematic} with a standard normal matrix, we obtain a different problem in stochastic geometry.  For the Gaussian model, we can establish a partial analog of Theorem~\ref{thm:kinematic} using a comparison inequality for Gaussian processes~\cite[Thm.~1.4]{gord:85}.  Rudelson \& Vershynin~\cite[Sec.~4]{RV:08} have used a corollary~\cite{gord:88} of this result to study the $\ell_1$ minimization method~\eqref{eqn:l1-min} for compressed sensing.  Many subsequent papers, including~\cite{stojnic10,OH:10,CRPW:12}, depend on the same argument.
In contrast to Theorem~\ref{thm:kinematic}, this approach is not based on an exact formula.  Nor does it apply to the random orthogonal
model, which is more natural than the Gaussian model for many applications.
\end{remark}

\subsection{Calculating the statistical dimension}
\label{sec:calc-sdim-intro}

The statistical dimension arises from deep considerations in conic integral geometry,
and we rely on this connection to prove that phase transitions occur in random convex optimization problems.  There is an alternative formulation that is often useful
for calculating the statistical dimension of specific cones.

\begin{proposition}[Statistical dimension: Metric characterization] \label{def:sdim}
The statistical dimension $\delta(C)$ of a closed convex cone $C$ in $\R^d$ satisfies 
\begin{equation} \label{eq:stat-dim-defn}
\sdim(C) = \Expect \big[ \enormsm{ \mtx{\Pi}_{C}(\vct{g}) }^2 \big],
\end{equation}
where $\vct{g} \in \R^d$ is a standard normal vector, $\enorm{\cdot}$ is the Euclidean norm,
and $\mtx{\Pi}_{C}$ denotes the Euclidean projection~\eqref{eq:cone-proj} onto the cone $C$.
\end{proposition}

\noindent
The proof of Proposition~\ref{def:sdim} appears in Section~\ref{sec:revisit-stat-dimen}.  The argument requires a classical result called the spherical Steiner formula~\cite[Thm.~6.5.1]{scwe:08}.

The metric characterization of the statistical dimension provides a surprising link between two perspectives on random convex optimization problems: our approach based on integral geometry and the alternative approach based on Gaussian process theory.  Indeed, the formula~\eqref{eq:stat-dim-defn} is closely related to the definition of another summary parameter for convex cones called the \term{Gaussian width}; see Section~\ref{sec:statdim_sqgw} for more information.  This connection allows us to perform statistical dimension calculations by adapting methods~\cite{RV:08,stojnic10,OH:10,CRPW:12} developed for the Gaussian width.  

We undertake this program in Sections~\ref{sec:calc-stat-dimens} and~\ref{sec:descent-cones} to estimate the statistical dimension for several important families of convex cones.
Although the resulting formulas are not substantially novel, we \emph{prove} for the first time that the error in these calculations is negligible.  Our contribution to this analysis forms a critical part of the rigorous computation of phase transitions.

\subsection{Regularized linear inverse problems with a random model}
\label{sec:line-inverse-probl}

Our first application of Theorem~\ref{thm:kinematic} concerns a generalization of the compressed sensing problem that has been studied in~\cite{CRPW:12}.  A \term{linear inverse problem} asks us to infer an unknown vector $\vct{x}_0 \in \R^d$ from an observed vector $\vct{z}_0 \in \R^m$ of the form
\begin{equation} \label{eq:lin-inv-obs}
	\vct{z}_0 = \mtx{A} \vct{x}_0,
\end{equation}
where $\mtx{A} \in \R^{m \times d}$ is
a matrix that describes a linear data acquisition process.  When the matrix is fat $(m < d)$,
the inverse problem is underdetermined.
In this situation, we cannot hope to identify $\vct{x}_0$ unless we take advantage of prior information about its structure.

\subsubsection{Solving linear inverse problems with convex optimization}

Suppose that $f : \R^d \to \overline{\R}$ is a proper convex function\footnote{The extended real numbers $\overline{\R} := \R \cup \{\pm \infty\}$.  A \term{proper} convex function has at least one finite value and never takes the value $-\infty$.}
that reflects the amount of ``structure'' in a vector.
We can attempt to identify the structured unknown $\vct{x}_0$
in~\eqref{eq:lin-inv-obs} by solving a convex optimization problem:
\begin{equation}
  \label{eq:lin-inv-gen}
\minimize f( \vct{x} )
\subjto \vct{z}_0 = \mtx{A} \vct{x}.
\end{equation}

\noindent
The function $f$ is called a \term{regularizer}, and
the formulation~\eqref{eq:lin-inv-gen} is called a \term{regularized linear inverse problem}.  To illustrate the kinds of regularizers that arise in practice,
we highlight two familiar examples.

\begin{example}[Sparse vectors]
When the vector $\vct{x}_0$ is known to be sparse, we can minimize the $\ell_1$ norm to look for a sparse solution to the inverse problem.  Repeating~\eqref{eqn:l1-min}, we have the optimization
\begin{equation}
  \label{eq:l1-min-v2}
\minimize \pnorm{1}{ \vct{x} }
\subjto \vct{z}_0 = \mtx{A} \vct{x}.
\end{equation}
This approach was proposed by Chen et al.~\cite{CDS:01}, motivated by work in geophysics~\cite{CM:73,SS:86}.
\end{example}

\begin{example}[Low-rank matrices]
Suppose that $\mtx{X}_0$ is a low-rank matrix, and we have acquired a vector of measurements of the form $\vct{z}_0 = \coll{A}(\mtx{X}_0)$, where $\coll{A}$ is a linear operator.  This process is equivalent with~\eqref{eq:lin-inv-obs}.  We can look for low-rank solutions to the linear inverse problem by minimizing the Schatten 1-norm:
\begin{equation} \label{eqn:S1-min}
\minimize \pnorm{S_1}{ \mtx{X} }
\subjto \vct{z}_0 = \coll{A}(\mtx{X}).
\end{equation}
This method was proposed in~\cite{RFP2010}, based on ideas from
control~\cite{MP:97} and optimization~\cite{Faz2002}.
\end{example}

We say that the regularized linear inverse problem~\eqref{eq:lin-inv-gen} \term{succeeds}
at solving~\eqref{eq:lin-inv-obs} when the convex program has a unique minimizer $\widehat{\vct{x}}$ that coincides with the true unknown; that is,
$\widehat{\vct{x}} = \vct{x}_0$.
To develop conditions for success, we introduce a convex cone associated with the
regularizer $f$ and the unknown $\vct{x}_0$.

\begin{definition}[Descent cone] \label{def:feas-cone}
The \term{descent cone} $\Desc(f,\vct{x})$ of a proper convex function $f : \R^d \to \overline{\R}$
at a point $\vct{x} \in \R^d$ is the conic hull of the
perturbations that do not increase $f$ near $\vct{x}$.
\begin{equation*}
  \Desc(f,\vct{x}) \defeq \bigcup_{\tau > 0}  \big\{\vct{y} \in \R^d :
  f(\vct{x} + \tau \vct{y} ) \leq f(\vct{x}) \big\}.
\end{equation*}
\end{definition}

\noindent
The descent cones of a proper convex function are always convex, but they may not be closed.
The descent cones of a smooth convex function are always halfspaces,
so this concept inspires the most interest when the function is nonsmooth.

To characterize when the optimization problem~\eqref{eq:lin-inv-gen} succeeds,
we write the primal optimality condition in terms of the descent cone;
cf.~\cite[Sec.~4]{RV:08} and \cite[Prop.~2.1]{CRPW:12}.

\begin{fact}[Optimality condition for linear inverse problems]
\label{fact:lin-inv-opt}
Let $f$ be a proper convex function.
The vector $\vct x_0$ is the unique optimal point of the convex program~\eqref{eq:lin-inv-gen}
if and only if $\Desc(f,\vct{x}_0) \cap \nullity(\mtx{A})  = \{\vct{0}\}.$
\end{fact}

\noindent
Figure~\ref{fig:geom-opt-cond} illustrates the geometry of this optimality condition.
Despite its simplicity, this result forges a crucial link between
the convex optimization problem~\eqref{eq:lin-inv-gen}
and the theory of conic integral geometry.

\begin{figure}
\begin{center}
\newcommand{\conelen}{2.75}
\newcommand{\myAlp}{12}
\newcommand{\myA}{2}
\begin{tikzpicture}[scale=1.5]
\fill[white!90!blue] (0,1) -- ++(225:\conelen) arc(225:315:\conelen) -- cycle;
\draw[thin,blue] (0,1) -- ++(225:\conelen) (0,1) -- ++(315:\conelen);

\fill[white!70!blue] (1,0) -- (0,1) -- (-1,0) -- (0,-1) -- cycle;
\draw[very thick] (1,0) -- (0,1) -- (-1,0) -- (0,-1) -- cycle;

\draw[very thick,black!30!red] (0,1) ++(180-\myAlp:\myA) -- ++(360-\myAlp:2*\myA); 
\draw[black!60!red,thick] (-1.1,1.4) node[above]{$\vct{x}_0 + \nullity(\mtx{A})$};

\draw[->,>=latex,thick,black!60!blue] (-1.25,0.6) node[left,scale=0.8]{$\{\vct{x} : f(\vct{x}) \leq f(\vct{x}_0) \}$} to [out=-15] (-0.67,0.33);

\fill (0,1) circle (1pt) node[above right=-1pt]{$\vct{x}_0$};

\path (0,-1.2) node[below,black!20!blue]{$\vct{x}_0 + \Desc(f,\vct{x}_0)$};

\end{tikzpicture}
\hspace{0.5in}
\renewcommand{\myAlp}{60}
\begin{tikzpicture}[scale=1.5]
\fill[white!90!blue] (0,1) -- ++(225:\conelen) arc(225:315:\conelen) -- cycle;
\draw[thin,blue] (0,1) -- ++(225:\conelen) (0,1) -- ++(315:\conelen);

\fill[white!70!blue] (1,0) -- (0,1) -- (-1,0) -- (0,-1) -- cycle;
\draw[very thick] (1,0) -- (0,1) -- (-1,0) -- (0,-1) -- cycle;

\draw[very thick,black!30!red] (0,1) ++(180-\myAlp:0.5*\myA) -- ++(360-\myAlp:2*\myA); 
\draw[black!60!red,thick] (-1.1,1.4) node[above]{$\vct{x}_0 + \nullity(\mtx{A})$};

\draw[->,>=latex,thick,black!60!blue] (-1.25,0.6) node[left,scale=0.8]{$\{\vct{x} : f(\vct{x}) \leq f(\vct{x}_0) \}$} to [out=-15] (-0.67,0.33);

\fill (0,1) circle (1pt) node[above right=-1pt]{$\vct{x}_0$};

\path (0,-1.2) node[below,black!20!blue]{$\vct{x}_0 + \mathscr{D}(f,\vct{x}_0)$};

\end{tikzpicture}

\end{center}
\caption{\textbf{The optimality condition for a regularized inverse problem.}
The condition for the regularized linear inverse problem~\eqref{eq:lin-inv-gen} to succeed requires that the descent cone $\Desc(f, \vct{x}_0)$ and the null space $\nullity(\mtx{A})$ do not share a ray.  \textbf{[left]}  The regularized linear inverse problem succeeds.  \textbf{[right]}  The regularized linear inverse problem fails.} 
\label{fig:geom-opt-cond}
\end{figure}
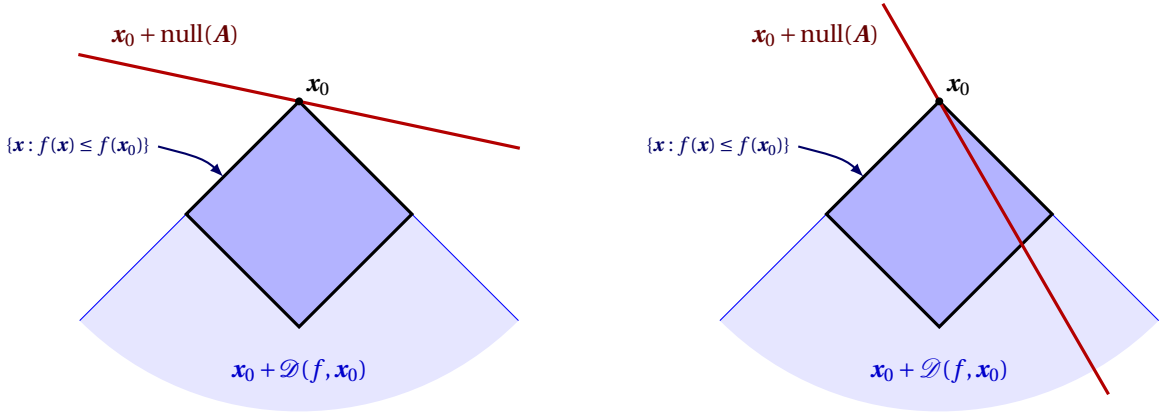

\subsubsection{Linear inverse problems with random data}

Our goal is to understand the power of convex regularization for solving linear
inverse problems, as well as the limitations inherent in this approach.
To do so, we consider the case where the measurements are \emph{generic}.
A natural modeling technique is to draw the measurement matrix $\mtx{A}$ at random
from the standard normal distribution on $\R^{m \times d}$.  In this case,
the kernel of the matrix $\mtx{A}$ is a randomly oriented subspace, so the
optimality condition, Fact~\ref{fact:lin-inv-opt},
requires us to calculate the probability that this
random subspace does not share a ray with the descent cone.

The kinematic formula, Fact~\ref{fact:kinem-intro}, gives an exact expression
for the probability that~\eqref{eq:lin-inv-gen} succeeds under the random model
for $\mtx{A}$.  By invoking the approximate kinematic formula,
Theorem~\ref{thm:kinematic},
we reach a simpler result that allows us to identify a sharp transition in
performance as the number $m$ of measurements varies.

\begin{bigthm}[Phase transitions in linear inverse problems with random measurements] \label{thm:phase-trans-lin-inv}
Fix a tolerance $\eta \in (0,1)$.
Let $\vct{x}_0 \in \R^d$ be a fixed vector, and
let $f : \R^d \to \overline{\R}$ be a proper convex function.
Suppose $\mtx{A} \in \R^{m \times d}$ has independent standard normal entries, and let $\vct{z}_0 = \mtx{A}\vct{x}_0$.
Then
\begin{align*}
m \leq \delta\big( \Desc(f, \vct{x}_0) \big) - a_{\eta} \sqrt{d}
&\quad\Longrightarrow\quad
\text{\eqref{eq:lin-inv-gen} succeeds with probability~$\leq \eta$}; \\
m \geq \delta\big( \Desc(f, \vct{x}_0) \big) + a_{\eta} \sqrt{d}
&\quad\Longrightarrow\quad
\text{\eqref{eq:lin-inv-gen} succeeds with probability~$\geq 1-\eta$.}
\end{align*}
The quantity $a_{\eta} := \sqrt{8\log(4/\eta)}$. \end{bigthm}

\begin{proof}
The standard normal distribution on $\R^{m \times d}$ is invariant under rotation, so the null space $L = \nullity(\mtx{A})$ is almost surely a uniformly random $(d-m)$-dimensional subspace of $\R^d$.  According to~\eqref{eq:sdim-dim-intro}, the statistical dimension $\delta(L) = d - m$ almost surely.  The result follows immediately when we combine the optimality condition, Fact~\ref{fact:lin-inv-opt}, and the kinematic bound, Theorem~\ref{thm:kinematic}.
\end{proof}

Under minimal assumptions, Theorem~\ref{thm:phase-trans-lin-inv} proves that we always encounter a phase transition when we use the regularized formulation~\eqref{eq:lin-inv-gen} to solve the linear inverse problem with random measurements.  The transition occurs where the number of measurements equals the statistical dimension of the descent cone: $m = \delta\big(\Desc(f, \vct{x}_0)\big)$.  The shift from failure to success takes place over a range of about $O(\sqrt{d})$ measurements.

Here is one way to think about this result.  We cannot identify $\vct{x}_0 \in \R^d$ from the observation $\vct{z}_0 \in \R^m$ by solving the linear system $\mtx{A} \vct{x} = \vct{z}_0$ because we only have $m$ equations.  Under the random model for $\mtx{A}$, the regularization in~\eqref{eq:lin-inv-gen} effectively adds $d - \sdim\big(\Desc(f, \vct{x}_0)\big)$ more equations to the system.  Therefore, we can typically recover $\vct{x}_0$ when $m \geq \sdim\big(\Desc(f,\vct{x}_0)\big)$.  This interpretation accords with the heuristic that the statistical dimension measures the dimension of a cone.

There are several reasons that the conclusions of Theorem~\ref{thm:phase-trans-lin-inv} are significant.  The first implication provides evidence about the minimum amount of information we need before we can use the convex method~\eqref{eq:lin-inv-gen} to solve the linear inverse problem.  The second implication tells us that we can solve the inverse problem reliably once we have acquired this quantum of information.  Furthermore, Theorem~\ref{thm:phase-trans-lin-inv} allows us to compare the performance of different regularizers because we know exactly how many measurements each one requires.

\begin{remark}[Prior work]
A variant of the success condition $m \geq \sdim\big(\Desc(f,\vct{x}_0)\big) + O(\sqrt{d})$ from Theorem~\ref{thm:phase-trans-lin-inv} already appears in the literature~\cite[Cor.~3.3(1)]{CRPW:12}.  This result depends on the argument of Rudelson \& Vershynin~\cite[Sec.~4]{RV:08}, which uses the ``escape from the mesh'' theorem~\cite{gord:88} to verify the optimality condition, Fact~\ref{fact:lin-inv-opt}, for the optimization problem~\eqref{eq:lin-inv-gen}. 
There is some evidence that the success condition accurately describes the performance limit for~\eqref{eq:lin-inv-gen} with random measurements.  Stojnic~\cite{stojnic10} presents analysis and experiments for the $\ell_1$ norm, while Oymak \& Hassibi~\cite{OH:10} study the Schatten 1-norm.  Results from~\cite[Sec.~17]{DMM:09-supp} and~\cite{BayLelMon:12} imply that Stojnic's calculation is asymptotically sharp.

Nevertheless, the prior literature offers no hint that the statistical dimension determines the location of the phase transition for every convex regularizer.  In fact, we can derive a variant of the failure condition from
Theorem~\ref{thm:phase-trans-lin-inv} by supplementing Rudelson \& Vershynin's approach
with a polarity argument. A similar observation appeared in Stojnic's paper~\cite{Sto13:Regularly-Random} after our work was released.
\end{remark}

\subsubsection{Computer experiments}

We have performed some computer experiments to compare the theoretical and empirical phase transitions.  Figure~\ref{fig:lin-inv-1}[left] shows the performance of~\eqref{eq:l1-min-v2} for identifying a sparse vector in $\R^{100}$ from random measurements.  Figure~\ref{fig:lin-inv-1}[right] shows the performance of~\eqref{eqn:S1-min} for identifying a low-rank matrix in $\R^{30 \times 30}$ from random measurements.  In each case, the heat map indicates the observed probability of success with respect to the randomness in the measurement operator.  The 5\%, 50\%, and 95\% success isoclines are calculated from the data.  We also draft the theoretical phase transition curve, promised by Theorem~\ref{thm:phase-trans-lin-inv}, where the number $m$ of measurements equals the statistical dimension of the appropriate descent cone; the statistical dimension formulas are drawn from Sections~\ref{sec:descent-cones-ell_1} and~\ref{sec:descent-cones-trace}.  See Appendix~\ref{app:experiments} for the experimental protocol.

In both examples, the theoretical prediction of Theorem~\ref{thm:phase-trans-lin-inv} coincides almost perfectly with the 50\% success isocline.  Furthermore, the phase transition takes place over a range of $O(\sqrt{d})$ values of $m$, as promised.  Although Theorem~\ref{thm:phase-trans-lin-inv} does not explain why the transition region tapers at the bottom-left and top-right corners of each plot, we have established a more detailed version of Theorem~\ref{thm:kinematic} that allows us to predict this phenomenon as well; see Section~\ref{sec:discuss-approx-kinem}.

\begin{figure}[t]
  \centering
  \includegraphics[width=0.48\textwidth]{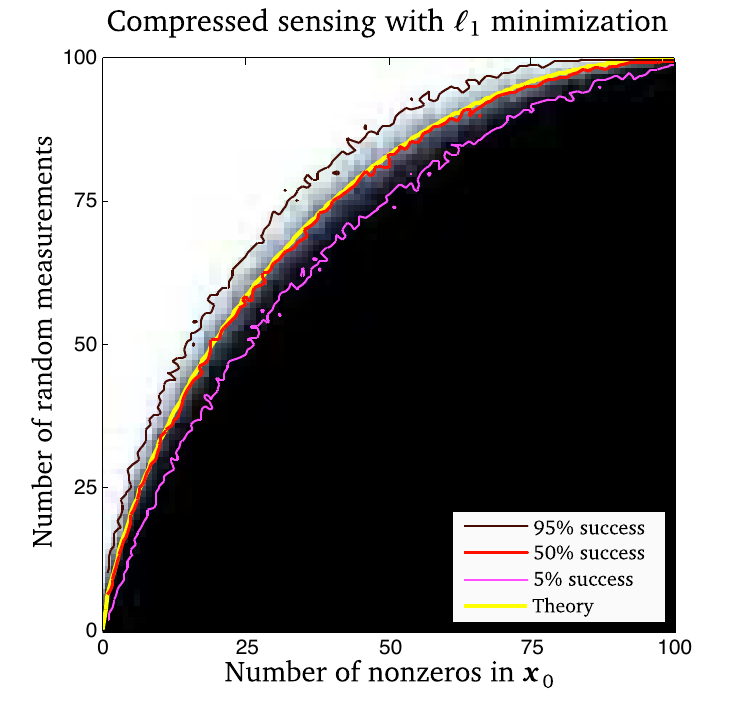}
  \includegraphics[width=0.48\textwidth]{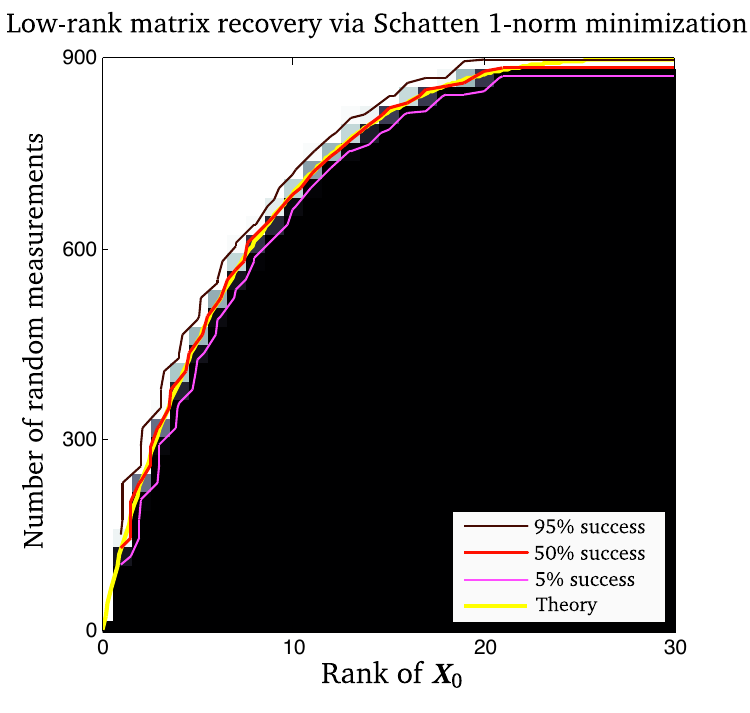}
  \caption{\textbf{Phase transitions for regularized linear inverse problems.}
  \textbf{[left] Recovery of sparse vectors.}  The empirical probability
  that the $\ell_1$ minimization problem~\eqref{eq:l1-min-v2} identifies a sparse vector $\vct{x}_0 \in \R^{100}$ given random linear measurements $\vct{z}_0 = \mtx{A}\vct{x}_0$. 
  \textbf{[right] Recovery of low-rank matrices.}  The empirical probability that the $S_1$ minimization problem~\eqref{eqn:S1-min} identifies a low-rank matrix $\mtx{X}_0 \in \R^{30 \times 30}$ given random linear measurements $\vct{z}_0 = \coll{A}(\mtx{X}_0)$.
	In each panel, the heat map indicates the empirical probability of success (black = 0\%; white = 100\%).  The yellow curve marks the theoretical prediction of the phase transition from Theorem~\ref{thm:phase-trans-lin-inv}; the red curve traces the $50\%$ success isocline calculated from the data.}
  \label{fig:lin-inv-1}
\end{figure}

\subsection{Demixing problems with a random model}
\label{sec:phase-trans-demix}

In a demixing problem~\cite{mctr:12}, we observe a superposition of two structured vectors, and we aim to extract the two constituents from the mixture.  More precisely, suppose that we have acquired a vector $\vct{z}_0 \in \R^d$ of the form
\begin{equation} \label{eq:demix-obs}
\vct{z}_0 = \vct{x}_0 + \mtx{U} \vct{y}_0
\end{equation}
where $\vct{x}_0, \vct{y}_0 \in \R^d$ are unknown and $\mtx{U} \in \R^{d \times d}$ is a known orthogonal matrix.  If we wish to identify the pair $(\vct{x}_0, \vct{y}_0)$, we must assume that each component is structured to reduce the number of degrees of freedom.  In addition, if the two types of structure are coherent (i.e., aligned with each other), it may be impossible to disentangle them, so it is expedient to include the matrix $\mtx{U}$ to model the relative orientation of the two constituent signals.

\subsubsection{Solving demixing problems with convex optimization}

Suppose that $f$ and $g$ are proper convex functions on $\R^d$ that promote the structures we expect to find in $\vct{x}_0$ and $\vct{y}_0$.  Then we can frame the convex optimization problem
\begin{equation} \label{eq:demix-solve}
\minimize f( \vct{x} )
\subjto g( \vct{y} ) \leq g( \vct{y}_0 )
\quad\text{and}\quad \vct{z}_0 = \vct{x} + \mtx{U} \vct{y}.
\end{equation}
In other words, we seek structured vectors $\vct{x}$ and $\vct{y}$ that are consistent with the observation $\vct{z}_0$.  This approach requires the side information $g(\vct{y}_0)$, so a Lagrangian formulation is sometimes more natural in practice~\cite[Sec.~1.2.4]{mctr:12}.
Here are two concrete examples of the demixing program~\eqref{eq:demix-solve} that are adapted from the literature.

\begin{example}[Sparse + sparse]
Suppose that the first signal $\vct{x}_0$ is sparse in the standard basis, and the second signal $\mtx{U}\vct{y}_0$ is sparse in a known basis $\mtx{U}$.  In this case, we can use $\ell_1$ norms to promote sparsity, which leads to the optimization
\begin{equation} \label{eq:l1+l1}
\minimize \pnorm{1}{ \vct{x} }
\subjto \pnorm{1}{ \smash{\vct{y}} } \leq \pnorm{1}{ \smash{ \vct{y}_0 } }
\quad\text{and}\quad \vct{z}_0 = \vct{x} + \mtx{U} \vct{y}.
\end{equation}
This approach for demixing sparse signals
is sometimes called \term{morphological component analysis}~\cite{SDC:2003,SED:2005,ESQD2005,BMS2006}.
\end{example}

\begin{example}[Low-rank + sparse]
Suppose that we observe $\mtx{Z}_0 = \mtx{X}_0 + \coll{U}(\mtx{Y}_0)$ where $\mtx{X}_0$ is a low-rank matrix, $\mtx{Y}_0$ is a sparse matrix, and $\coll{U}$ is a known orthogonal transformation on the space of matrices.
We can minimize the Schatten 1-norm to promote low rank, and we can constrain the $\ell_1$ norm to promote sparsity.  The optimization becomes
\begin{equation} \label{eq:l1+S1}
\minimize \pnorm{S_1}{ \mtx{X} }
\subjto \pnorm{1}{\mtx{Y}} \leq \pnorm{1}{ \mtx{Y}_0 }
\quad\text{and}\quad \mtx{Z}_0 = \mtx{X} + \coll{U}( \mtx{Y} ).
\end{equation}
This demixing problem is called the \term{rank--sparsity decomposition}~\cite{CSPW2011}.
\end{example}

We say that the convex program~\eqref{eq:demix-solve} for demixing \term{succeeds} when it has a unique solution $(\widehat{\vct{x}}, \widehat{\vct{y}})$ that coincides with the vectors that generate the observation: $(\widehat{\vct{x}}, \widehat{\vct{y}}) = (\vct{x}_0, \vct{y}_0)$.  As in the case of a linear inverse problem, we can express the primal  optimality condition in terms of descent cones; cf.~\cite[Lem.~2.4]{mctr:12}.

\begin{fact}[Optimality condition for demixing] \label{fact:demix-cond}
Let $f$ and $g$ be proper convex functions.  The pair $(\vct{x}_0,\vct{y}_0)$ is the unique
optimal point of the convex program~\eqref{eq:demix-solve} if and only if
$\Desc(f,\vct x_0)\cap (-\mtx{U}\Desc(g,\vct{y}_0)) = \{\vct{0}\}$.
\end{fact}

\noindent
Figure~\ref{fig:geom-opt-cond-2} depicts the geometry of this optimality condition.  The parallel with Fact~\ref{fact:lin-inv-opt}, the optimality condition for a regularized linear inverse problem, is striking.  Indeed, the two conditions coalesce when the function $g$ in~\eqref{eq:demix-solve} is the indicator of an appropriate affine space.  This observation shows that the regularized linear inverse problem~\eqref{eq:lin-inv-gen} is a special case of the convex demixing problem~\eqref{eq:demix-solve}.

\begin{center}
\begin{figure}
\newcommand{\conelen}{2.75}
\newcommand{\conelend}{2.75}
\newcommand{\myAlp}{200}
\newcommand{\myA}{.5}
\pgfmathsetmacro\myB{sqrt(2)*\myA}
\begin{tikzpicture}[scale=1.35]
\fill[white!90!blue] (0,1) -- ++(225:\conelen) arc(225:315:\conelen) -- cycle; \draw[thin,blue] (0,1) -- ++(225:\conelen) (0,1) -- ++(315:\conelen); 
\fill[white!50!red,nearly transparent] (0,1) -- ++(135+\myAlp:\conelend) arc(135+\myAlp:225+\myAlp:\conelend) -- cycle;
\draw[thin,red] (0,1) -- ++(135+\myAlp:\conelend) (0,1) -- ++(225+\myAlp:\conelend);

\fill[white!0!red,nearly transparent] (0,1) -- ++(135+\myAlp:\myB) -- ++(225+\myAlp:\myB) -- ++(315+\myAlp:\myB) -- cycle;
\draw[very thick,black!30!black] (0,1) -- ++(135+\myAlp:\myB) -- ++(225+\myAlp:\myB) -- ++(315+\myAlp:\myB) -- cycle;

\fill[white!75!blue] (1,0) -- (0,1) -- (-1,0) -- (0,-1) -- cycle;
\draw[very thick] (1,0) -- (0,1) -- (-1,0) -- (0,-1) -- cycle;

\draw[->,>=latex,thick,black!60!blue] (-1.25,0.6) node[left,scale=0.8]{$\{\vct{x} : f(\vct{x}) \leq f(\vct{x}_0) \}$} to [out=-15] (-0.67,0.33);

\draw[->,>=latex,thick,black!60!red] (-1,2) node[above,scale=0.8]{$\{\vct{x} : g(\mtx{U}^\transp(\vct{z}_0 - \vct{x})) \leq g(\vct{y}_0) \}$} to [out=-90,in=165] (0.22,1.5);

\fill[black] (0,1) circle(1pt);
\path (0,1) node[above=4pt,left=1pt]{$\vct{x}_0$};

\path (0,-1.2) node[below,black!20!blue]{$\vct{x}_0 + \mathscr{D}(f,\vct{x}_0)$};
\path (1.5,2.25) node[black!40!red]{$\vct{x}_0 -\mtx{U}\mathscr{D}(g,\vct{y}_0)$};

\end{tikzpicture}
\renewcommand{\myAlp}{165}
\begin{tikzpicture}[scale=1.35]
\fill[white!90!blue] (0,1) -- ++(225:\conelen) arc(225:315:\conelen) -- cycle; \draw[thin,blue] (0,1) -- ++(225:\conelen) (0,1) -- ++(315:\conelen); 
\fill[white!75!blue] (1,0) -- (0,1) -- (-1,0) -- (0,-1) -- cycle;
\draw[very thick] (1,0) -- (0,1) -- (-1,0) -- (0,-1) -- cycle;

\fill[white!50!red,nearly transparent] (0,1) -- ++(135+\myAlp:\conelend) arc(135+\myAlp:225+\myAlp:\conelend) -- cycle;
\draw[thin,red] (0,1) -- ++(135+\myAlp:\conelend) (0,1) -- ++(225+\myAlp:\conelend);

\fill[white!0!red,nearly transparent] (0,1) -- ++(135+\myAlp:\myB) -- ++(225+\myAlp:\myB) -- ++(315+\myAlp:\myB) -- cycle;
\draw[very thick,black!30!black] (0,1) -- ++(135+\myAlp:\myB) -- ++(225+\myAlp:\myB) -- ++(315+\myAlp:\myB) -- cycle;

\draw[->,>=latex,thick,black!60!blue] (-1.25,0.6) node[left,scale=0.8]{$\{\vct{x} : f(\vct{x}) \leq f(\vct{x}_0) \}$} to [out=-15] (-0.67,0.33);

\draw[->,>=latex,thick,black!60!red] (-1,2) node[above,scale=0.8]{$\{\vct{x} : g(\mtx{U}^\transp(\vct{z}_0 - \vct{x})) \leq g(\vct{y}_0) \}$} to [out=-90,in=120] (0.45,1.25);

\fill[black] (0,1) circle(1pt);
\path (0,1) node[above=4pt,left=1pt]{$\vct{x}_0$};

\path (0,-1.2) node[below,black!20!blue]{$\vct{x}_0 + \mathscr{D}(f,\vct{x}_0)$};
\path (1.75,1.35) node[black!40!red]{$\vct{x}_0 -\mtx{U}\mathscr{D}(g,\vct{y}_0)$};

\end{tikzpicture}
\caption{\textbf{The optimality conditions for convex demixing.}  The condition for the convex demixing problem~\eqref{eq:demix-obs} to succeed requires that the cones $\Desc(f, \vct{x}_0)$ and $-\mtx{U} \Desc(\vct{g}, \vct{y}_0)$ do not share a ray.  \textbf{[left]}  The convex demixing method succeeds. \textbf{[right]} The convex demixing method fails.}
\label{fig:geom-opt-cond-2}
\end{figure}
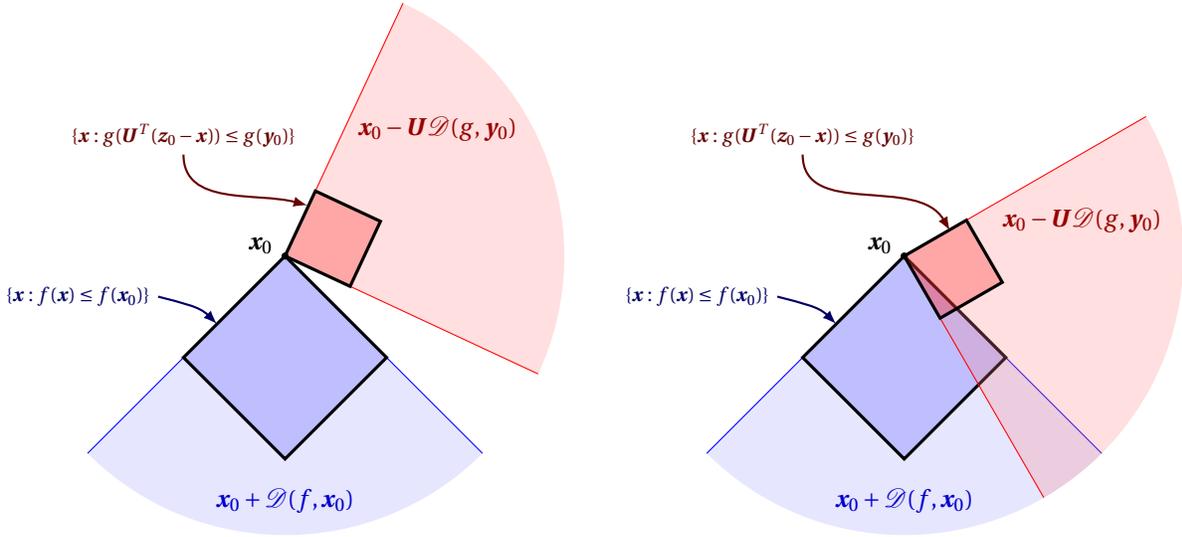
\end{center}

\subsubsection{Demixing with a random model for coherence}

Our goal is to understand the prospects for solving the demixing problem with a convex program of the form~\eqref{eq:demix-solve}.
To that end, we use randomness to model the favorable case where the two structures do not interact with each other.  More precisely, we choose the matrix $\mtx{U}$ to be a random orthogonal basis.  Under this assumption, Theorem~\ref{thm:kinematic} delivers a sharp transition in the performance of the optimization problem~\eqref{eq:demix-solve}.

\begin{bigthm}[Phase transitions in convex demixing with a random coherence model] \label{thm:demix}
Fix a tolerance $\eta \in (0,1)$.
Let $\vct{x}_0$ and $\vct{y}_0$ be fixed vectors in $\R^d$, and
let $f : \R^d \to \overline{\R}$ and $g : \R^d \to \overline{\R}$ be proper convex functions.
Draw a random orthogonal basis $\mtx{U} \in \R^{d \times d}$,
and let $\vct{z}_0 = \vct{x}_0 + \mtx{U} \vct{y}_0$.  Then
\begin{align*}
\sdim\big( \Desc( f, \vct{x}_0 ) \big) + \sdim\big( \Desc( g, \vct{y}_0 ) \big)
	\geq d + a_{\eta} \sqrt{d}
	&\quad\Longrightarrow\quad
	\text{\eqref{eq:demix-solve} succeeds with probability $\leq \eta$}; \\
\sdim\big( \Desc( f, \vct{x}_0 ) \big) + \sdim\big( \Desc( g, \vct{y}_0 ) \big)
	\leq d - a_{\eta} \sqrt{d}
	&\quad\Longrightarrow\quad
	\text{\eqref{eq:demix-solve} succeeds with probability $\geq 1-\eta$}.
\end{align*}
The quantity $a_{\eta} := \sqrt{8\log(4/\eta)}$. \end{bigthm}

\begin{proof}
This theorem follows immediately when we combine the optimality condition, Fact~\ref{fact:demix-cond}, with the kinematic bound, Theorem~\ref{thm:kinematic}.  To simplify the formulas, we invoke the rotational invariance of the statistical dimension, which follows from either Definition~\ref{def:sdim-int} or Proposition~\ref{def:sdim}.
\end{proof}

Under minimal assumptions, Theorem~\ref{thm:demix} establishes that there is always a phase transition when we use the convex program~\eqref{eq:demix-solve} to solve the demixing problem under the random model for $\mtx{U}$.  The optimization is effective if and only if the total statistical dimension of the two descent cones is smaller than the ambient dimension $d$. 

\subsubsection{Computer experiments}

Our numerical work confirms the analysis in Theorem~\ref{thm:demix}.  Figure~\ref{fig:MCA-phase}[left] shows when~\eqref{eq:l1+l1} can demix a sparse vector from a vector that is sparse in a random basis $\mtx{U}$ for $\R^{100}$.  Figure~\ref{fig:MCA-phase}[right] shows when~\eqref{eq:l1+S1} can demix a low-rank matrix from a matrix that is sparse in a random basis $\coll{U}$ for $\R^{35 \times 35}$.  In each case, the experiment provides an empirical estimate for the probability of success with respect to the randomness in the coherence model.  We display the 5\%, 50\%, and 95\% success isoclines, calculated from the data.  We also sketch the theoretical phase transition from Theorem~\ref{thm:demix}, which occurs when the total statistical dimension of the relevant cones equals the ambient dimension; the statistical dimensions of the descent cones are obtained from the formulas in Sections~\ref{sec:descent-cones-ell_1} and~\ref{sec:descent-cones-trace}.  See~\cite[Sec.~6]{mctr:12} for the details of the experimental protocol.

Once again, we see that the theoretical curve of Theorem~\ref{thm:demix} coincides almost perfectly with the empirical 50\% success isocline.  The width of the transition region is $O(\sqrt{d})$.  Although Theorem~\ref{thm:demix} does not predict the tapering of the transition in the top-left and bottom-right corners, the discussion in Section~\ref{sec:discuss-approx-kinem} exposes the underlying reason for this phenomenon.

\begin{figure}[t]
  \centering
  \includegraphics[width=0.48\columnwidth]{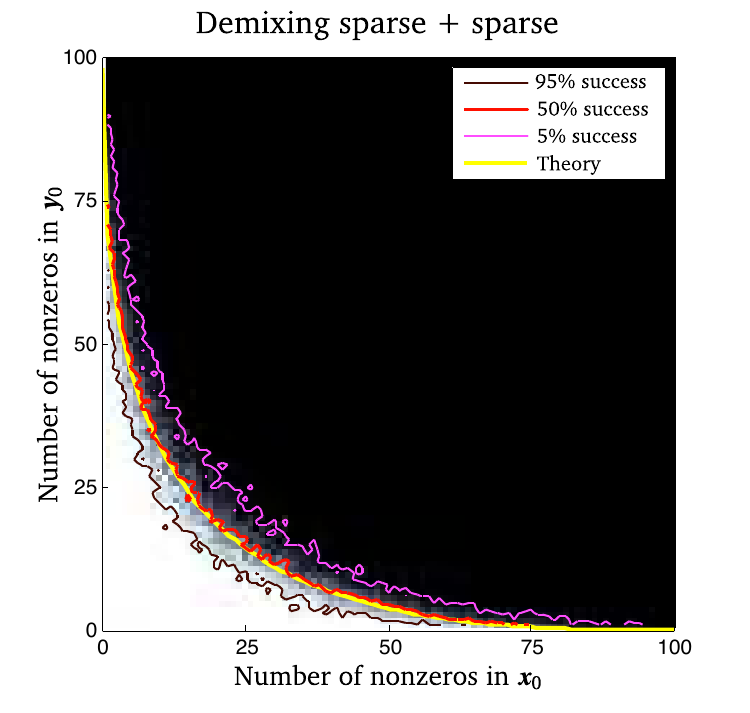}
  \includegraphics[width=0.48\columnwidth]{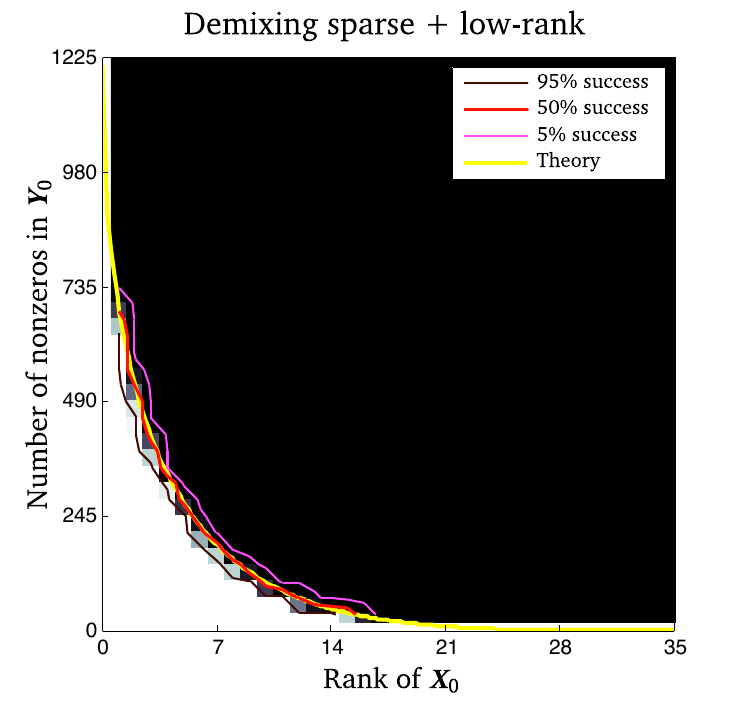}
  \caption{\textbf{Phase transitions for convex demixing.}
  \textbf{[left] Sparse + sparse.}  The empirical probability that
  the convex program~\eqref{eq:l1+l1} successfully demixes a vector
  $\vct{x}_0 \in \R^{100}$ that is sparse in the standard basis
  from a vector $\mtx{U} \vct{y}_0 \in \R^{100}$ that is sparse in the
  random basis $\mtx{U}$.  \textbf{[right] Low rank + sparse.}  The empirical
  probability that the convex program~\eqref{eq:l1+S1} successfully
  demixes a low-rank matrix $\mtx{X}_0 \in \R^{35 \times 35}$ from
  a matrix $\coll{U}(\mtx{Y}_0) \in \R^{35 \times 35}$ that is sparse in the
  random basis $\coll{U}$.
  In each panel, the heat map indicates the empirical probability
  of success (black = 0\%; white = 100\%).  The yellow curve
  marks the theoretical phase transition predicted by Theorem~\ref{thm:demix}.
  The red curve follows the empirical phase transition.}
  \label{fig:MCA-phase}
\end{figure}

\subsection{Contributions} 
\label{sec:contributions}

It takes a substantial amount of argument to establish the existence of phase transitions and to calculate their location for specific problems.  Some parts of our paper depend on prior work, but much of the research is new.  We conclude this section with a summary of our contributions.  Section~\ref{sec:conclusion} contains a detailed discussion of the literature; additional citations appear throughout the presentation.

This paper contains foundational research in conic integral geometry:

\begin{itemize} \setlength{\itemsep}{1mm}
\item	We define the statistical dimension as the mean value of the distribution of intrinsic volumes, and we argue that the statistical dimension canonically extends the linear dimension of a subspace to the class of convex cones.  (Definition~\ref{def:sdim-int} and Section~\ref{sec:sdim-canon})

\item	We demonstrate that the metric characterization of the statistical dimension coincides with the intrinsic characterization, and we use this connection to establish some properties of the statistical dimension.  The statistical dimension is also related to the Gaussian width.  (Definition~\ref{def:sdim-int}, Proposition~\ref{def:sdim}, Proposition~\ref{prop:prop-sdim}, Proposition~\ref{prop:sdim-int-vols}, and Proposition~\ref{prop:sdim-width})

\item	We prove that the distribution of intrinsic volumes of a convex cone concentrates sharply about the statistical dimension of the cone. (Theorem~\ref{thm:main-conc})

\item	The concentration of intrinsic volumes leads to an approximate version of the kinematic formula for cones.  This result uses the statistical dimension to bound the probability that a randomly rotated cone shares a ray with a fixed cone.  (Theorem~\ref{thm:kinematic} and Theorem~\ref{thm:approx-kinem})
\end{itemize}

Building on this foundation, we establish a number of applied results concerning phase transition phenomena in convex optimization problems with random data:

\begin{itemize} \setlength{\itemsep}{1mm}
\item	We prove that a regularized linear inverse problem with random measurements
must exhibit a phase transition as the number of random measurements increases.
The location and width of the transition are controlled by the statistical
dimension of a descent cone.  Our work confirms
and extends the earlier analyses based on polytope angles~\cite{dono:06,dota:09a,KXAH2011,XuHas:12}
and those based on Gaussian process theory~\cite{RV:08,stojnic10,OH:10,CRPW:12}.
(Theorem~\ref{thm:phase-trans-lin-inv}, Theorem~\ref{thm:approx-kinem}, and Proposition~\ref{prop:vec-list})

\item	The paper~\cite{mctr:12} proposes convex programming methods for decomposing a superposition of two structured, randomly oriented vectors into its constituents.  We prove that these methods exhibit a phase transition whose properties depend on the total statistical dimension of two descent cones.  This work confirms a conjecture~\cite[Sec.~4.2.2]{mctr:12} about the existence of phase transitions in these problems.  (Theorem~\ref{thm:demix} and Theorem~\ref{thm:approx-kinem}) 

\item	The work~\cite{ambu:11c} studies cone programs with random affine constraints.  Building on this analysis, we show that a cone program with random affine constraints displays a phase transition as the number of constraints increases.  We can predict the transition using the statistical dimension of the cone.  (Theorem~\ref{thm:cone-prog})

\item	Section~\ref{sec:descent-cones} contains a recipe for estimating the statistical dimension of a descent cone.
The approach is based on ideas from~\cite[App.~C]{CRPW:12}, but we provide the first proof that it delivers accurate estimates.  This result rigorously explains why the bounds computed in~\cite{stojnic10,OH:10} closely match observed phase transitions. (Theorem~\ref{thm:sharp-descent} and Propositions~\ref{prop:l1-sdim} and~\ref{prop:S1-sdim})

\item	The approximate kinematic formula also delivers information about the probability that a face of a polytope maintains its dimension under a random projection.  This argument clarifies the connection between the polytope-angle approach to random linear inverse problems and the approach based on Gaussian process theory. (Section~\ref{sec:face-plant})
\end{itemize}

As an added bonus, we provide the final ingredient needed to resolve a series of conjectures~\cite{DonMalMon:09,DonJohMon:11,DonGavMon:13} about the coincidence between the minimax risk of denoising and the location of phase transitions in linear inverse problems.  Indeed, Oymak \& Hassibi~\cite{OymHas:12} have recently shown that the minimax risk is equivalent with the statistical dimension of an appropriate cone, and our results prove that the phase transition occurs at precisely this spot.  See Section~\ref{sec:minimax} for further details.

\section{Calculating the statistical dimension}
\label{sec:calc-stat-dimens}

Section~\ref{sec:conic-phase} demonstrates that the statistical
dimension is a fundamental quantity in conic integral geometry.
Through the approximate kinematic formula, Theorem~\ref{thm:kinematic},
the statistical dimension drives phase transitions
in random linear problems and random demixing problems.
A natural question, then, is how we can determine the value
of the statistical dimension for a specific cone.

This section explains how to compute the statistical dimension
directly for a few basic cones.
In Section~\ref{sec:prop-stat-dimens}, we present some useful properties of the statistical dimension.  Sections~\ref{sec:self-dual-cones}--\ref{sec:fin_refl_gr} contain some example calculations, in increasing order of difficulty.
See Table~\ref{tab:stadim} for a summary.  We discuss general descent cones later, in Section~\ref{sec:descent-cones}.

\subsection{Basic facts about the statistical dimension}
\label{sec:prop-stat-dimens}

The statistical dimension has a number of valuable properties.
These facts provide useful tools for making computations,
and they strengthen the analogy between the statistical
dimension of a cone and the linear dimension of a subspace.

\begin{proposition}[Properties of statistical dimension] \label{prop:prop-sdim}
Let $C$ be a closed convex cone in $\R^d$.
The statistical dimension obeys the following laws.

\begin{enumerate} \setlength{\itemsep}{1mm}
\item	\textbf{Intrinsic formulation.}
The statistical dimension is defined as
\begin{equation} \label{eq:sdim-intr}
\delta(C) := \sum_{k=0}^d k \, v_k(C)
\end{equation}
where $v_0, \dots, v_d$ denote the conic intrinsic volumes.
(See Section~\ref{sec:conic-intr-vol}.)

\item	\textbf{Gaussian formulation.} The statistical dimension satisfies
\begin{equation} \label{eq:sdim-gauss}
\delta(C) = \Expect{} \big[ \enormsm{ \Proj_{C}(\vct{g}) }^2 \big]
\quad\text{where $\vct{g} \sim \normal(\vct{0}, \Id_d)$.} 
\end{equation}

\item	\textbf{Spherical formulation.} An equivalent expression is
\begin{equation} \label{eq:sdim-circ-expect}
\delta(C) = d \, \Expect{} \big[ \enormsm{ \Proj_{C}(\vct{\theta}) }^2 \big]
\quad\text{where $\vct{\theta} \sim \uniform(\sphere{d-1})$.}
\end{equation}

\item	\textbf{Polar formulation.} The statistical dimension can be expressed in terms of the polar cone:
\begin{equation} \label{eq:delta-dist}
\delta(C) = \Expect \big[ \dist^2( \vct{g}, C^\polar ) \big].
\end{equation}

\item \textbf{Mean-squared-width formulation.}  Another alternative formulation reads
  \begin{equation}\label{eq:sdim-sup-expect}
\delta(C) = \Expect{} \bigl[ \bigl( \sup\nolimits_{\vct{y} \in C \cap \ball{d}} \ip{\smash{\vct{y}}}{ \smash{\vct{g}}} \bigr)^2 \bigr]\quad
\text{where $\vct{g} \sim \normal(\zerovct, \Id_d)$.}
  \end{equation}

\item	\textbf{Rotational invariance.} The statistical dimension does not depend on the orientation of the cone:
\begin{equation} \label{eq:sdim-rot-inv}
\delta(\mtx{U} C) = \delta(C)
\quad\text{for each orthogonal matrix $\mtx{U} \in \R^{d\times d}$.}
\end{equation}

\item	\textbf{Subspaces.} \label{item:sdim-subspace}
For each subspace $L$, the statistical dimension satisfies $\delta(L) = \dim(L)$.

\item	\textbf{Complementarity.} The sum of the statistical dimension of a cone and that of its polar equals the ambient dimension:
\begin{equation} \label{eq:sdim-polarity-sum}
\delta(C) + \delta(C^\polar) = d.
\end{equation}
This generalizes the property $\dim(L) + \dim(L^\perp) = d$
for each subspace $L \subset \R^d$.

\item	\textbf{Direct products.} For each closed convex cone $K\subset\R^d$,
\begin{equation} \label{eq:sdim-direct-product}
\delta(C \times K) = \delta(C) + \delta(K).
\end{equation}
In particular, the statistical dimension is invariant under embedding:
$$
\delta\big(C \times \{ \vct{0}_{d'} \} \big) = \delta(C).
$$
The relation~\eqref{eq:sdim-direct-product} generalizes the rule
$\dim(L \times M) = \dim(L) + \dim(M)$ for linear subspaces $L$ and $M$.

\item	\textbf{Monotonicity.} For each closed convex cone $K \subset \R^d$, the inclusion
$C \subset K$ implies that $\delta(C) \leq \delta(K)$.
\end{enumerate}
\end{proposition}

\noindent
We verify the equivalence of the intrinsic volume formulation~\eqref{eq:sdim-intr}
and the metric formulation~\eqref{eq:sdim-gauss}
below in Proposition~\ref{prop:sdim-int-vols}.  It is possible to establish
the remaining facts on the basis of either~\eqref{eq:sdim-intr} or~\eqref{eq:sdim-gauss}.
In Appendix~\ref{sec:proof-prop-refpr},
we use the metric characterization~\eqref{eq:sdim-gauss} to prove the rest of the proposition.
Many of these elementary results have appeared in~\cite{stojnic10,CRPW:12}
in a slightly different form.

\begin{table}[t]
\centering
\caption{\textbf{The statistical dimensions of some convex cones}.}
\begin{tabu}{lccc}
    \toprule
     Cone & Notation &  Statistical dimension  & Location
    \\ 
    \midrule
    The nonnegative orthant & $\R_+^d$ & $\half d$ & Sec.~\ref{sec:self-dual-cones} \\
    The second-order cone & $\mathbb{L}^{d+1}$ & $\half (d+1)$ & Sec.~\ref{sec:self-dual-cones} \\
    \begin{minipage}{5cm}
    Symmetric positive- \\ semidefinite matrices
    \end{minipage} & $\mathbb{S}_+^{n \times n}$ & $\tfrac{1}{4} n(n+1)$ & Sec.~\ref{sec:self-dual-cones} \\
	\begin{minipage}{5cm} Descent cone of the $\ell_\infty$ norm \\
	at an $s$-saturated vector $\vct{x}$ in $\R^d$
	\end{minipage} & $\Desc(\pnorm{\infty}{\cdot}, \vct{x})$ & $d - \half s$ & Sec.~\ref{sec:desc-linf} \\
    Circular cone in $\R^d$ of angle $\alpha$ & $\Circ_d(\alpha)$ &  $d\sin^2(\alpha) + O(1)$ & Sec.~\ref{sec:circular-cones} \\
    \begin{minipage}{4.5cm}
    Chambers of finite reflection groups acting on $\R^d$
    \end{minipage} &
    \begin{minipage}{2cm} \begin{equation*}\begin{array}{c} C_A \\[2pt] C_{BC}     \end{array}\end{equation*}\end{minipage} &
    \begin{minipage}{4cm}
      \begin{equation*}
      \begin{array}{c}
            \log(d) +O(1) \\[2pt]
           	\half \log(d) + O(1)           \end{array}
      \end{equation*}
  \end{minipage}
 & Sec.~\ref{sec:fin_refl_gr}
    \\
    \bottomrule
  \end{tabu}
  \label{tab:stadim}
\end{table}

\subsection{Self-dual cones}
\label{sec:self-dual-cones}

We say that a cone $C$ is \term{self-dual} when $C^{\polar} = -C$.  Self-dual cones are ubiquitous in the theory and practice of convex optimization.  Here are three important examples:

\begin{enumerate} \setlength{\itemsep}{1mm}
\item	\textbf{The nonnegative orthant.}  The cone $\R_+^d := \big\{ \vct{x} \in \R^d : x_i \geq 0 \ \text{for $i = 1, \dots, d$} \big\}$ is self-dual.

\item	\textbf{The second-order cone.}  The cone $\mathbb{L}^{d+1} := \big\{ (\vct{x}, \tau) \in \R^{d +1} : \enorm{\vct{x}} \leq \tau \big\}$ is self-dual.  This example is sometimes called the \term{Lorentz cone} or the \term{ice-cream cone}.

\item	\textbf{Symmetric positive-semidefinite matrices.}  The cone $\mathbb{S}_+^{n \times n} := \big\{ \mtx{X} \in \R^{n \times n}_{\text{sym}} : \mtx{X} \psdge \mtx{0} \big\}$ is self-dual, where the curly inequality denotes the semidefinite order.  Note that the linear space $\R^{n \times n}_{\text{sym}}$ of $n \times n$ symmetric matrices has dimension $\half n(n+1)$.
\end{enumerate}

\noindent
For a self-dual cone, the computation of the statistical dimension is particularly simple; cf.~\cite[Cor.~3.8]{CRPW:12}.  The first three entries in Table~\ref{tab:stadim} follow instantly from this result.

\begin{proposition}[Self-dual cones] \label{prop:self-dual}
Let $C$ be a self-dual cone in $\R^d$. The statistical dimension
$\delta(C) = \half d$.
\end{proposition}

\begin{proof}
Just observe that
$\delta(C) = \half \big[ \delta(C) + \delta(C^\polar) \big] = \half d$.
The first identity holds because of the self-dual property of the cone and the rotational invariance~\eqref{eq:sdim-rot-inv} of statistical dimension.  The second equality follows from the complementarity law~\eqref{eq:sdim-polarity-sum}.
\end{proof}

\subsection{Descent cones of the \texorpdfstring{\(\ell_\infty\)}{linf} norm}
\label{sec:desc-linf}

Recall that the $\ell_\infty$ norm of a vector $\vct{x} \in \R^d$
is defined as
$$
\pnorm{\infty}{\vct{x}} := \max_{i = 1, \dots, d}\ \abs{x_i}.
$$
The descent cones of the $\ell_\infty$ norm have a simple form that
allows us to calculate their statistical dimensions easily.
We have drawn this analysis from~\cite[Sec.~6.2.4]{McC:13}.
Variants of this result drive the applications in~\cite{DT:10,JFF12:Anti-Sparse-Coding}.
Section~\ref{sec:descent-cones} develops a method for studying more general types of
descent cones.

\begin{proposition}[Descent cones of the $\ell_\infty$ norm] \label{prop:linf-desc}
Consider an $s$-saturated vector $\vct{x} \in \R^d$; that is,
$$
\# \big\{ i : \abs{x_i} = \pnorm{\infty}{\vct{x}} \big\} = s
\quad\text{where $1 \leq s \leq d$.}
$$
Then
$$
\Desc( \pnorm{\infty}{\cdot}, \vct{x} ) = d - \half s.
$$
\end{proposition}

\begin{proof}[Proof sketch]
The $\ell_\infty$ norm is homogeneous and invariant under signed permutation,
so we may assume that the $s$-saturated vector $\vct{x} \in \R^d$ takes the form
$$
\vct{x} = (1, \dots, 1, x_{s+1}, \dots, x_d)^\transp
\quad\text{where $1 > x_{s+1} \geq \cdots \geq x_d \geq 0$.}
$$
The descent cone of the $\ell_{\infty}$ norm at the vector $\vct{x}$ can be written as
$$
\Desc(\pnorm{\infty}{\cdot}, \vct{x})
	= \big(- \R_+^s \big) \times \R^{d-s}.
$$
We may now calculate the statistical dimension:
$$
\delta\big(\Desc(\pnorm{\infty}{\cdot}, \vct{x})\big)
	= \delta\big( \R_+^s \big) + \delta\big(\R^{d-s}\big)
	= \half s + (d-s)
	= d - \half s.
$$
The first identity follows from the direct product law~\eqref{eq:sdim-direct-product}
and the rotational invariance~\eqref{eq:sdim-rot-inv} of the statistical dimension.
The next relation depends on the rule for subspaces and the calculation of
the statistical dimension of the nonnegative orthant from Section~\ref{sec:self-dual-cones}.
\end{proof}

\subsection{Circular cones}
\label{sec:circular-cones}

The \term{circular cone} $\Circ_d(\alpha)$ in $\R^d$ with angle $0 \leq \alpha \leq \tfrac{\pi}{2}$ is defined as
$$
\Circ_d(\alpha) := \big\{ \vct{x} \in \R^d : x_1 \geq \enorm{\vct{x}} \cos(\alpha) \big\}.
$$
In particular, the cone $\Circ_d(\tfrac{\pi}{4})$ is isometric to the second-order cone $\mathbb{L}^d$.
Circular cones have numerous applications in optimization; we refer the reader
to~\cite[Sec.~4]{BV:04},~\cite[Sec.~3]{B-TN:01}, and~\cite{AG:03} for details.

We can obtain an accurate expression for the statistical dimension of a circular
cone by expressing the spherical formulation~\eqref{eq:sdim-circ-expect}
in spherical coordinates and administering a dose of asymptotic analysis.

\begin{proposition}[Circular cones] \label{prop:circ-cone}
The statistical dimension of a circular cone satisfies
\begin{equation} \label{eq:sdim-cone-est}
\delta\big(\Circ_d(\alpha) \big) = d \sin^2(\alpha) + O(1).
\end{equation}
The error term is approximately equal to $\cos(2\alpha)$.
See Figure~\ref{fig:statdim-comput}[left] for a plot
of~\eqref{eq:sdim-cone-est}.
\end{proposition}

\noindent
Turn to Appendix~\ref{app:circ-cone} for the proof, which seems to be novel.
Even though the formula in Proposition~\ref{prop:circ-cone} is simple,
it already gives an excellent approximation in moderate
dimensions.  See~\cite[Sec.~6.3]{McCTro:13} or~\cite{MuHuaWri:13}
for refinements of Proposition~\ref{prop:circ-cone}
that appeared after our work.

\begin{figure}
\centering
\def\myX{3}  \def\myY{1}
\def\conelen{4.2}
\def\conelend{3.0}
\begin{tikzpicture}[scale=0.75]
\draw[white](0,-5.0) --(0,0); \draw[thin,blue] (\myX,-\myY) -- (\myX,-\myY+\conelen) arc(90:225:\conelen) -- cycle;
\draw[very thick,white] (\myX,-\myY+\conelen) arc(90:225:\conelen);
\fill[white!90!blue] (\myX,-\myY) -- (\myX,-\myY+\conelen) arc(90:225:\conelen) -- cycle;
\draw[thin,black] (\myX,-\myY) -- (\myX+\conelend,-\myY) arc(360:315:\conelend) -- cycle;
\draw[very thick,white] (\myX+\conelend,-\myY) arc(360:315:\conelend);
\fill[blue!70!black] (\myX,-\myY) -- (\myX+\conelend,-\myY) arc(360:315:\conelend) -- cycle;
\draw (-1.2*\myX,0) -- (1.2*\myX,0) (0,-1.2*\myX) -- (0,1.2*\myX);

\draw[very thick] (\myX,\myY) -- (\myY,\myX) -- (-\myY,\myX) -- (-\myX,\myY) -- (-\myX,-\myY) -- (-\myY,-\myX) -- (\myY,-\myX) -- (\myX,-\myY) -- cycle;
\fill (\myX,-\myY) circle(2pt) node[above=1mm, left]{\(\vct x\)};
\path (\myX+1,-\myY+4)
node{}
(\myX+3,-\myY-0.5)
node{}
(-3,-2.7) node{\(\mathcal{P}_\pm(\vct x)\)};
\end{tikzpicture} \hspace{0.25in}
\def\myX{3}  \def\myY{2}
\def\conelen{4.5}
\def\conelend{3.0}
\begin{tikzpicture}[scale=0.75]

\draw[thin,blue] (\myX,-\myY) -- (\myX,-\myY+\conelen) arc(90:225:\conelen) -- cycle;
\draw[very thick,white] (\myX,-\myY+\conelen) arc(90:225:\conelen);
\fill[white!90!blue] (\myX,-\myY) -- (\myX,-\myY+\conelen) arc(90:225:\conelen) -- cycle;
\draw[thin,black] (\myX,-\myY) -- (\myX+\conelend,-\myY) arc(360:315:\conelend) -- cycle;
\draw[very thick,white] (\myX+\conelend,-\myY) arc(360:315:\conelend);
\fill[blue!70!black] (\myX,-\myY) -- (\myX+\conelend,-\myY) arc(360:315:\conelend) -- cycle;
\draw (-1.2*\myX,0) -- (1.2*\myX,0) (0,-1.2*\myX) -- (0,1.2*\myX);
\draw[very thick] (\myX,\myY) -- (\myY,\myX) -- (-\myY,\myX) -- (-\myX,\myY) -- (-\myX,-\myY) -- (-\myY,-\myX) -- (\myY,-\myX) -- (\myX,-\myY) -- cycle;
\fill (\myX,-\myY) circle(2pt) node[above=1mm, left]{\(\vct x\)};
\path (\myX+1,-\myY+4)
node{}
(\myX+3,-\myY-0.5)
node{}
(-3.2,-3.2) node{\(\mathcal{P}_\pm(\vct x)\)};
\end{tikzpicture} \caption{
\textbf{Normal cone at the vertex of a permutahedron.}
The signed permutahedron $\coll{P}_{\pm}(\vct{x})$ generated by \textbf{[left]} the vector $\vct{x} = (3,-1)$ and \textbf{[right]} the vector $\vct{x} = (3,-2)$.
In each panel, the darker cone is the normal cone $\NormC(\coll{P}_{\pm}(\vct{x}), \vct{x})$, and the lighter cone is its polar.  Note that the normal cone does not depend on the generator $\vct{x}$ provided that the entries of $\vct{x}$ are distinct.}
\label{fig:signed_permutahedron}
\end{figure}
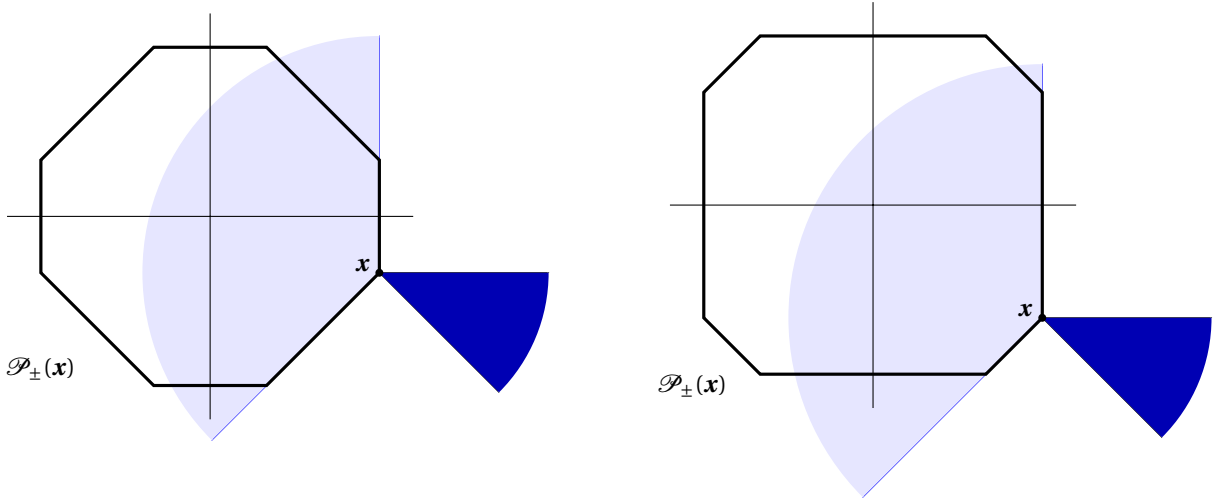

\subsection{Normal cones of a permutahedron}
\label{sec:fin_refl_gr} 
We close this section with a more sophisticated example.
The \term{(signed) permutahedron} generated by a vector $\vct{x} \in \R^d$ is the convex hull of all (signed) coordinate permutations of the vector:
\begin{align}
	\mathcal{P}(\vct{x}) &:= \conv\big\{ \sigma(\vct{x}) : \text{$\sigma$ a coordinate permutation} \big\} \label{eq:perm} \\
	\mathcal{P}_{\pm}(\vct{x}) &:= \conv\big\{ \sigma_{\pm}(\vct{x}) : \text{$\sigma_{\pm}$ a signed coordinate permutation} \big\} \label{eq:sperm}.
\end{align}
A signed permutation permutes the coordinates of a vector and may change the sign of each coordinate.
The \term{normal cone} $\NormC( E, \vct{x} )$
 of a convex set $E \subset \R^d$ at a point $\vct{x} \in E$
 consists of the outward normals to all hyperplanes that support $E$ at $\vct{x}$, i.e.,
 \begin{equation} \label{eq:ncone}
 \NormC( E, \vct{x} ) := \big\{ \vct{s} \in \R^d : \ip{ \vct{s} }{ \smash{\vct{y}} - \vct{x} } \leq 0
 \quad\text{for all $\vct{y} \in E$} \big\}.
 \end{equation}
Figure~\ref{fig:signed_permutahedron} displays two signed permutahedra along with the normal cone at a vertex of each one.

We can develop an exact formula for the statistical dimension of the normal cone of a nondegenerate permutahedron.  In Section~\ref{sec:appl-permutahedron}, we use this calculation to study a signal processing application proposed in~\cite[p.~812]{CRPW:12}.

\begin{proposition}[Normal cones of permutahedra] \label{prop:perm-cone}
Suppose that $\vct{x}$ is a vector in $\R^d$ with distinct entries.  The statistical dimension of the normal cone at a vertex of the (signed) permutahedron generated by $\vct{x}$ satisfies
$$
\delta\big( \NormC( \mathcal{P}(\vct{x}), \vct{x} ) \big) = \mathrm{H}_d
\quad\text{and}\quad
\delta\big( \NormC( \mathcal{P}_{\pm}(\vct{x}), \vct{x} ) \big) = \half \mathrm{H}_d,
$$
where $\mathrm{H}_d := \sum_{i=1}^d i^{-1}$ is the $d$th harmonic number.
\end{proposition}

\noindent
The proof of Proposition~\ref{prop:perm-cone} appears in Appendix~\ref{sec:chamb-finite-refl}.  The argument relies on the intrinsic formulation~\eqref{eq:sdim-intr} of the statistical dimension, and it illustrates some deep connections between conic geometry and classical combinatorics.

\section{The statistical dimension of a descent cone}
\label{sec:descent-cones}

Theorems~\ref{thm:phase-trans-lin-inv} and~\ref{thm:demix} allow us to locate the phase transition for a class of convex optimization problems with random data.
To apply these results, however, we must be able to compute the statistical dimension for the descent cone of a convex function.  In this section, we describe a recipe that delivers an accurate estimate for the statistical dimension of a descent cone.

In Section~\ref{sec:descent-upper}, we explain how to obtain an upper bound for the statistical dimension of a descent cone.  Section~\ref{sec:descent-lower} provides an error estimate for this method.
In Sections~\ref{sec:descent-cones-ell_1} and~\ref{sec:descent-cones-trace},
we apply these ideas to study the descent cones of the $\ell_1$ norm and the Schatten 1-norm.

\subsection{A recipe for the statistical dimension of a descent cone}
\label{sec:descent-upper}

There is an elegant way to obtain an upper bound for the statistical dimension of a descent cone of a convex function.
Recall that the \term{subdifferential} $\partial f(\vct{x})$ of a proper convex function $f : \R^d \to \overline{\R}$ at a point $\vct{x} \in \R^d$ is the closed convex set
\begin{equation*} \partial f(\vct{x}) :=
	\big\{ \vct{s} \in \R^d : f(\vct{y}) \geq f(\vct{x}) + \ip{\vct{s}}{\smash{\vct{y}} - \vct{x}} \text{ for all $\vct{y} \in \R^d$} \big\}.
\end{equation*}
There is a classical duality between descent cones and subdifferentials~\cite[Chap.~23]{Rock}.
As a consequence, we can convert questions about the statistical dimension of a descent cone into questions about the subdifferential.

\begin{proposition}[The statistical dimension of a descent cone] \label{prop:sdim-descent}
Let $f : \R^d \to \overline{\R}$ be a proper convex function, and let $\vct{x} \in \R^d$.
Assume that the subdifferential $\partial f(\vct{x})$ is nonempty, compact, and does not contain the origin.
Define the function
\begin{equation} \label{eq:Ftau}
\distsubdiff(\tau) :=
\distsubdiff(\tau; \ \partial f(\vct{x})) :=
\Expect \big[ \dist^2\big(\vct{g}, \tau \cdot \partial f(\vct{x}) \big) \big]
\quad\text{for $\tau \geq 0$}
\end{equation}
where $\vct{g} \sim \normal(\vct{0}, \Id)$.  We have the upper bound
\begin{equation} \label{eq:sdim-descent}
\sdim\big( \Desc(f, \vct{x}) \big)
	\leq \inf_{\tau \geq 0} \ \distsubdiff(\tau).
\end{equation}
Furthermore, the function $\distsubdiff$ is strictly convex, continuous at $\tau = 0$, and differentiable for $\tau \geq 0$.  It achieves its minimum at a unique point.
\end{proposition}

\begin{proof} The inequality~\eqref{eq:sdim-descent} generalizes some specific arguments from~\cite[App.~C]{CRPW:12}, and it is easy to establish.  We use the polarity relation~\eqref{eq:delta-dist} to compute the statistical dimension:
$$
\delta\big( \Desc(f, \vct{x}) \big)
	= \delta\big( \Desc(f, \vct{x})^{\polar\polar} \big)
	= \Expect \big[ \dist^2\big(\vct{g},{} \Desc(f, \vct{x})^\polar \big) \big]
	= \Expect \left[ \dist^2\left(\vct{g}, \bigcup_{\tau \geq 0} \tau \cdot \partial f(\vct{x}) \right) \right]
	= \Expect{} \bigg[ \inf_{\tau \geq 0}\ \dist^2\big(\vct{g}, \tau \cdot \partial f(\vct{x}) \big) \bigg].
$$
Indeed, the statistical dimension of the descent cone equals the statistical dimension
of its closure, which can be expressed as a double polar.  The second identity
uses~\eqref{eq:delta-dist} and the fact that polarity is an involution on closed convex cones.
The third identity follows from the fact that, under our technical assumptions, the polar of the descent cone is the cone generated by the subdifferential~\cite[Cor.~23.7.1]{Rock}; see Appendix~\ref{sec:repr-desc-cones} for details.
The last identity holds because the distance to a union is the infimal distance to any one of its members.  To reach~\eqref{eq:sdim-descent}, we pass the expectation through the infimum.

The analytic properties of the function $\distsubdiff$ are new, and they demand substantial effort.  Lemma~\ref{slem:distance-conic} in Appendix~\ref{sec:dist-conic} contains the proof of the remaining claims. \end{proof}

Proposition~\ref{prop:sdim-descent} suggests a method for studying the statistical dimension of a descent cone: Minimize the function $\distsubdiff$ by setting its derivative to zero.  We formalize this approach in Recipe~\ref{fig:desc-recipe}.  In Section~\ref{sec:descent-cones-ell_1} and~\ref{sec:descent-cones-trace}, we discuss some examples where this recipe is provably effective.

\begin{remark}[Prior work]
Rudelson \& Vershynin~\cite[Sec.~4]{RV:08} established a bound for the Gaussian width of the descent cone of the $\ell_1$ norm at a sparse vector using a geometric argument.  Stojnic~\cite{stojnic10} obtained a substantial refinement of this estimate using linear programming duality.  Oymak \& Hassibi~\cite{OH:10} developed a related method to study the descent cone of the Schatten 1-norm at a low-rank matrix.  The paper~\cite[App.~C]{CRPW:12} clarifies the calculations of Stojnic and Oymak \& Hassibi using geometric polarity arguments.  The result~\eqref{eq:sdim-descent} extends the latter approach to a general convex function.
\end{remark}

\subsection{An error estimate for the descent cone recipe}
\label{sec:descent-lower}

The papers~\cite{stojnic10,OH:10} contain computational evidence that the ideas behind Recipe~\ref{fig:desc-recipe} lead to accurate upper bounds in some special cases.  Among the major contributions of this paper is an error estimate that explains why the descent cone recipe works so well.  This result is an essential ingredient in our method for locating
the phase transition of a random convex program.  Indeed, we need this theorem to ensure that we can calculate the statistical dimension of a descent cone correctly.

\begin{theorem}[Error bound for descent cone recipe] \label{thm:sharp-descent}
Let $f$ be a norm on $\R^d$, and fix a nonzero point $\vct{x} \in \R^d$.  Then
\begin{equation} \label{eq:sharp-descent}
0 \leq \bigg( \inf_{\tau \geq 0} \  \distsubdiff(\tau;\ \partial f(\vct{x})) \bigg) - \sdim \big( \Desc(f, \vct{x}) \big)
	\leq \frac{2 \, \sup\big\{ \enorm{\vct{s}} : \vct{s} \in \partial f(\vct{x}) \big\}}{f(\vct{x} / \enorm{\vct{x}})},
\end{equation}
where the function $\distsubdiff$ is defined in~\eqref{eq:Ftau}.
\end{theorem}

\noindent
The proof of Theorem~\ref{thm:sharp-descent} is technical in nature, so we defer the details to Appendix~\ref{app:sharp-descent}.  The application of this result requires some care because many different vectors $\vct{x}$ can generate the same subdifferential $\partial f(\vct{x})$ and hence the same descent cone $\Desc(f,\vct{x})$.  From this class of vectors, we ought to select one that maximizes the value $f(\vct{x}/\norm{\vct{x}})$.

\begin{remark}[Related work]
In an independent paper that appeared shortly after this work was released, Foygel \& Mackey~\cite{FoyMac:13} developed another error bound for the descent cone recipe.  These two results operate under different assumptions, and the two bounds are effective in different regimes.  It remains an open question to find an optimal error estimate
for Recipe~\ref{fig:desc-recipe}.
\end{remark}

\begin{recipe}[t]
\centering
\caption{\textbf{The statistical dimension of a descent cone.}}
\label{fig:desc-recipe}
\framebox{
\begin{minipage}{0.9\textwidth}
\textbf{Assume} that $f : \R^d \to \overline{\R}$ is a proper convex function and $\vct{x} \in \R^d$ \\
\textbf{Assume} that the subdifferential $\partial f(\vct{x})$ is nonempty, compact, and does not contain the origin
\vspace{1mm}
\begin{enumerate} \setlength{\itemsep}{1mm}
\item	Identify the subdifferential $S = \partial f(\vct{x})$.

\item	For each $\tau \geq 0$, compute $\distsubdiff(\tau) = \Expect \big[ \dist^2(\vct{g}, \tau S) \big]$.

\item	Find the unique solution, if it exists, to the stationary equation $\distsubdiff'(\tau) = 0$.

\item	If the stationary equation has a solution $\tau_{\star}$, then $\sdim\big( \Desc(f, \vct{x}) \big) \leq \distsubdiff(\tau_{\star})$.

\item	Otherwise, the bound is vacuous: $\sdim\big( \Desc(f,\vct{x}) \big) \leq \distsubdiff(0) = d$.
\end{enumerate}
\vspace{1mm}
\end{minipage}}
\end{recipe}

\subsection{Descent cones of the \texorpdfstring{\(\ell_1\)}{l1} norm}
\label{sec:descent-cones-ell_1}

When we wish to solve an inverse problem with a sparse unknown,
we often use the $\ell_1$ norm as a regularizer;
cf.~\eqref{eq:l1-min-v2},~\eqref{eq:l1+l1}, and~\eqref{eq:l1+S1}.
Our next result summarizes the calculations required to obtain
the statistical dimension of the descent cone of the $\ell_1$ norm
at a sparse vector.  When we combine this proposition
with Theorems~\ref{thm:phase-trans-lin-inv} and~\ref{thm:demix},
we obtain the exact location of the phase transition
for $\ell_1$ regularized inverse problems whose dimension is large.

\begin{proposition}[Descent cones of the $\ell_1$ norm] \label{prop:l1-sdim}
Let $\vct{x}$ be a vector in $\R^d$ with $s$ nonzero entries.
Then the normalized statistical dimension of the descent cone of the $\ell_1$ norm at $\vct{x}$ 
satisfies the bounds
\begin{equation} \label{eq:l1-sdim}
\psi(s/d) - \frac{2}{\sqrt{sd}}
	\leq \frac{\sdim\big( \Desc(\pnorm{1}{\cdot}, \vct{x} ) \big)}{d} 
	\leq \psi(s/d).
\end{equation}
The function $\psi : [0,1] \to [0,1]$ is defined as
\begin{equation} \label{eq:l1-sdim-curve}
\psi(\rho) := \inf_{\tau \geq 0} \ \left\{ \rho (1 + \tau^2) + (1 - \rho )
	\int_\tau^\infty (u - \tau)^2 \cdot \phi(u) \idiff{u} \right\}.
\end{equation}
The integral kernel $\phi(u) := \sqrt{2/\pi} \, \econst^{-u^2/2}$ is a probability density supported on $[0, \infty)$.
The infimum in~\eqref{eq:l1-sdim-curve} is achieved for the unique positive $\tau$
that solves the stationary equation
\begin{equation} \label{eq:l1-stationary}
\int_\tau^\infty \left(\frac{u}{\tau} - 1 \right) \cdot \phi(u) \idiff{u}
	= \frac{\rho}{1 - \rho}.
\end{equation}
See Figure~\ref{fig:statdim-comput}[center] for a plot of the function~\eqref{eq:l1-sdim-curve}.
\end{proposition}

\noindent
Proposition~\ref{prop:l1-sdim} is a direct consequence of Recipe~\ref{fig:desc-recipe} and the error bound in Theorem~\ref{thm:sharp-descent}.  See Appendix~\ref{app:l1-sdim} for details of the proof;
Appendix~\ref{sec:addend-comp-stat} explains the numerical aspects.

Let us emphasize the following consequences of Proposition~\ref{prop:l1-sdim}.
When the number $s$ of nonzeros in the vector $\vct{x}$ is proportional to the ambient dimension $d$, the error in the statistical dimension calculation~\eqref{eq:l1-sdim} is vanishingly small relative to the ambient dimension.  When $\vct{x}$ is sparser, it is more appropriate to compare the error with the statistical dimension itself.  Thus,
$$
0 \leq  \frac{\big(d \cdot \psi(s/d)\big) - \sdim\big(\Desc(\pnorm{1}{\cdot}, \vct{x})\big)}{\sdim\big(\Desc(\pnorm{1}{\cdot}, \vct{x})\big)} 
	\leq \frac{2}{\sqrt{\sdim\big(\Desc(\pnorm{1}{\cdot},\vct{x})\big)}}
	\quad\text{when $s \geq\sqrt{d} + 1$.}
$$
We have used the observation that $\sdim\big(\Desc(\pnorm{1}{\cdot},\vct{x})\big) \geq s - 1$, which holds because $\Desc(\pnorm{1}{\cdot}, \vct{x})$ contains the $(s-1)$-dimensional subspace parallel with the minimal face of the $\ell_1$ ball containing $\vct{x}$.

\begin{remark}[Prior work]
Except for the first inequality in~\eqref{eq:l1-sdim},
the calculations and the resulting formulas in Proposition~\ref{prop:l1-sdim}
are not substantially novel.  Most of the existing analysis concerns the
phase transition in compressed sensing,
i.e., the $\ell_1$ minimization problem~\eqref{eq:l1-min-v2}
with Gaussian measurements.  In this setting,
Donoho~\cite{dono:06b} and Donoho \& Tanner~\cite{dota:09a}
obtained an asymptotic upper bound, equivalent to the upper bound in~\eqref{eq:l1-sdim},
from polytope angle calculations.
Stojnic~\cite{stojnic10} established the same asymptotic upper bound using
a precursor of Recipe~\ref{fig:desc-recipe};
see also~\cite[App.~C]{CRPW:12}. In addition, there are some heuristic arguments,
based on ideas from statistical physics, that lead to the same result,
cf.~\cite{DonMalMon:09} and~\cite[Sec.~17]{DMM:09-supp}.  Very recently,
Bayati et al.~\cite{BayLelMon:12} have shown
that, in the asymptotic setting, the compressed sensing problem undergoes
a phase transition at the location predicted by~\eqref{eq:l1-sdim}.
\end{remark}

\begin{figure}
  \centering
  \includegraphics[width=\textwidth]{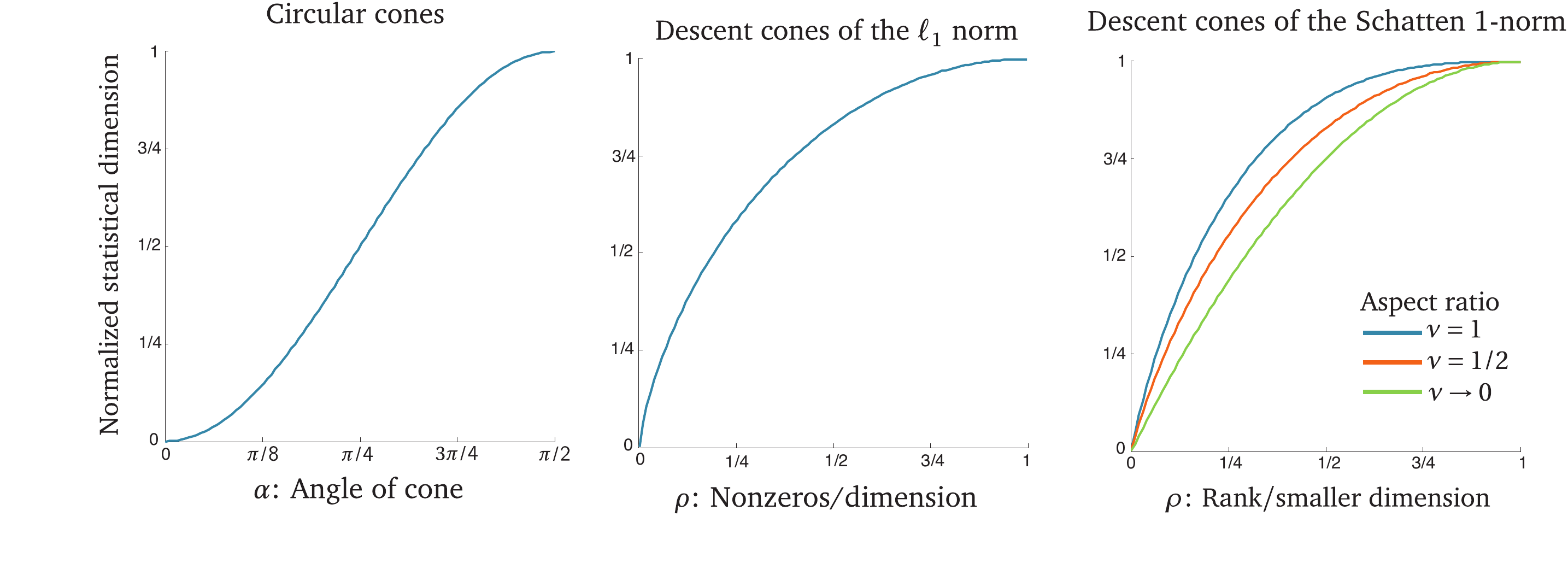}  
  \caption{\textbf{Asymptotic statistical dimension computations.}
  In each panel, we take the dimensional parameters to infinity.
  \textbf{[left] Circular cones.}  The plot shows the normalized
  statistical dimension $\delta(\cdot)/d$ of the circular cone $\Circ_d(\alpha)$.
  \textbf{[center] $\ell_1$ descent cones.}  The curve traces the 
  normalized statistical dimension $\delta(\cdot)/d$ of the descent
  cone of the $\ell_1$ norm on $\R^d$ at a vector with $\lfloor \rho d \rfloor$ nonzero
  entries.
  \textbf{[right] Schatten 1-norm descent cones.}  The normalized statistical
  dimension $\delta(\cdot)/(mn)$ of the descent cone of the $S_1$ norm
  on $\R^{m \times n}$ at a matrix with rank $\lfloor \rho m \rfloor$ for several fixed
  aspect ratios $\nu = m/n$.
  As the aspect ratio $\nu \to 0$, the limiting curve is $\rho \mapsto 2\rho - \rho^2$.}
  \label{fig:statdim-comput}
\end{figure}

\subsection{Descent cones of the Schatten 1-norm}
\label{sec:descent-cones-trace}

When we wish to solve an inverse problem whose unknown is a low-rank matrix, we often use the Schatten 1-norm $S_1$ as a regularizer, as in~\eqref{eqn:S1-min} and~\eqref{eq:l1+S1}.  The following result gives an asymptotically exact expression for the statistical dimension of the descent cone of the $S_1$ norm at a low-rank matrix.  Together with Theorems~\ref{thm:phase-trans-lin-inv} and~\ref{thm:demix}, this proposition allows us
to identify the exact location of the phase transition for $S_1$ regularized
inverse problems as the ambient dimension goes to infinity.

\begin{proposition}[Descent cones of the $S_1$ norm] \label{prop:S1-sdim}
Consider a sequence $\{\mtx{X}(r,m,n)\}$ of matrices where $\mtx{X}(r,m,n)$ has rank $r$ and dimension $m \times n$ with $m \leq n$.  Suppose that $r, m, n \to \infty$ with limiting ratios $r/m \to \rho \in (0,1)$ and $m/n \to \nu \in (0,1]$.  Then
\begin{equation} \label{eq:S1-sdim}
\frac{\sdim\big( \Desc\big( \snorm{\cdot}, \mtx{X}(r,m,n) \big) \big)}{mn}
	\to \psi(\rho, \nu).
\end{equation}
The function $\psi : [0,1]\times[0,1] \to [0,1]$ is defined as
\begin{equation} \label{eq:S1-curve}
\psi(\rho, \nu) := \inf_{\tau \geq 0} \ \left\{ \rho \nu + (1 - \rho \nu) \left[
	\rho \big(1 + \tau^2 \big) + (1 - \rho)
	\int_{a_- \vee \tau}^{a_+} (u-\tau)^2 \cdot \phi_y(u) \idiff{u} \right] \right\}.
\end{equation}
The quantity $y := (\nu - \rho \nu)/(1-\rho \nu)$, and the limits of the integral are $a_{\pm} := 1 \pm \sqrt{y}$.  The integral kernel $\phi_y$ is a probability density supported on $[a_-, a_+]$:
$$
\phi_y(u) := \frac{1}{\pi y u} \sqrt{(u^2 - a_-^2)(a_+^2 - u^2)}
\quad\text{for $u \in [a_-, a_+]$.}
$$
The optimal value of $\tau$ in~\eqref{eq:S1-curve} satisfies the stationary equation
\begin{equation} \label{eq:S1-stationary}
\int_{a_- \vee \tau}^{a_+} \left(\frac{u}{\tau} - 1\right) \cdot \phi_y(u) \idiff{u}
	= \frac{\rho}{1-\rho}.
\end{equation}
See Figure~\ref{fig:statdim-comput}[right] for a visualization of
the curve~\eqref{eq:S1-curve} as function of $\rho$ for several choices of $\nu$.
The operator $\vee$ returns the maximum of two numbers.  
\end{proposition}

\noindent
See Appendix~\ref{sec:feas-cone-schatt} for a proof of Proposition~\ref{prop:S1-sdim}.
Appendix~\ref{sec:addend-comp-stat} contains details of the numerical calculation.

\begin{remark}[Prior work]
The literature contains several papers that, in effect, contain loose upper bounds for the statistical dimension of the descent cones of the Schatten 1-norm~\cite{RXH:11,OKH:11}.  We single out the work~\cite{OH:10} of Oymak \& Hassibi, which identifies an empirically sharp upper bound via a laborious argument.  The approach here is more in the spirit of the weak upper bound in~\cite[App.~C]{CRPW:12}, but our argument leads to the asymptotically correct estimate.
\end{remark}

\section{Conic integral geometry and the statistical dimension}
\label{sec:intr-conic-integr}

To prove that the statistical dimension controls the location of
phase transitions in random convex optimization problems,
we rely on methods from \term{conic integral geometry},
the field of mathematics concerned with geometric properties
of convex cones that remain invariant under rotations, reflections, and embeddings.
Here are some of the guiding questions in this area:
\begin{itemize}
\item \textit{
What is the probability that a random unit vector lies at most a specified distance from a fixed cone?}
\item \textit{
What is the probability that a randomly rotated cone shares a ray with a fixed cone?}
\end{itemize}
The theory of conic integral geometry offers beautiful and precise answers to these questions,
phrased in terms of a set of geometric invariants called
\term{conic intrinsic volumes}.

In Section~\ref{sec:conic-intr-vol}, we formally introduce the intrinsic volumes of a cone and
we compute the intrinsic volumes of some basic cones.  We state
the key facts about intrinsic volumes in Section~\ref{sec:int-vol-props}.
Sections~\ref{sec:steiner} and~\ref{sec:kinematic} contain more advanced
formulas from conic integral geometry, which are essential tools in our approach
to phase transitions.
In Section~\ref{sec:revisit-stat-dimen},
we establish the equivalence of the two characterizations of the statistical dimension, given in
Definition~\ref{def:sdim-int} and Proposition~\ref{def:sdim}.
Section~\ref{sec:sdim-canon} explains why the statistical dimension
is canonical.

The material in this section is adapted from the book~\cite[Sec.~6.5]{scwe:08} and 
the dissertation~\cite{am:thesis}.  The foundational research in this area is due to
Santal{\'o}~\cite[Part IV]{sant:76}.  Modern treatments depend on the work
of Glasauer~\cite{Gl,Gl:Summ}.  In these sources, the theory is presented in
terms of spherical geometry, rather than in terms of conical geometry.
As noted in~\cite{ambu:11c}, the two approaches are equivalent, but
the conic viewpoint provides simpler formulas and has other benefits that
are revealed by deeper structural investigations.

\subsection{Conic intrinsic volumes}
\label{sec:conic-intr-vol}

We begin with the definition of the intrinsic volumes of a
convex cone.  Recall that a cone is \emph{polyhedral} if it can be written as the intersection of a finite number of halfspaces.  Polyhedral cones are automatically closed and convex.

\begin{definition}[Intrinsic volumes: Polyhedral case] \label{def:intr-vol}
Let $C$ be a polyhedral cone in $\R^d$.  For each $k = 0, 1, 2, \dots, d$, the $k$th \term{(conic) intrinsic volume} $v_k(C)$ is given by
$$
v_k(C) := \Prob\big\{ \text{$\Proj_C(\vct{g})$ lies in the relative interior of a $k$-dimensional face of $C$} \big\}.
$$
As usual, $\vct{g}$ is a standard normal vector in $\R^d$.
See Figure~\ref{fig:conic-int-vols} for an illustration.
\end{definition}

\begin{figure}
\begin{center}
\includegraphics[height=2in]{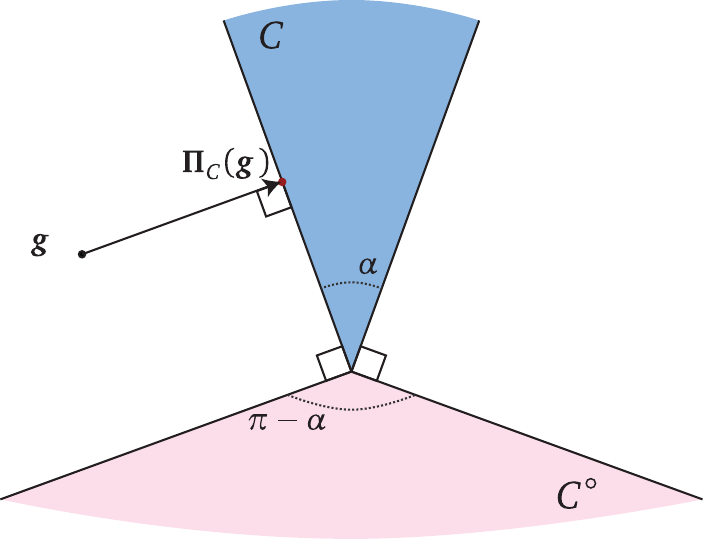}
\caption{\textbf{The intrinsic volumes of a convex cone in two dimensions.} The closed convex cone $C \subset \R^2$ has four faces: one two-dimensional face (dark shading), two one-dimensional faces (the boundary rays), and one zero-dimensional face (the origin). The projection $\Proj_C(\vct{g})$ of a standard normal vector $\vct{g}$ onto $K$ lies in the two-dimensional face when $\vct{g}$ is in the dark region, on a one-dimensional face when $\vct{g}$ is in the white region, and in the zero-dimensional face when $\vct{g}$ is in the light region.  Each intrinsic volume of the cone $C$ can be expressed in terms of the solid angle $\alpha$ measured in radians: $v_2(C) = \alpha / (2\pi)$ and $v_1(C) = 1/2$ and $v_0(C) = (\pi - \alpha)/(2\pi)$.}
\label{fig:conic-int-vols}
\end{center}
\end{figure}

For a polyhedral cone, it is clear that the sequence of intrinsic volumes
forms a probability distribution on $\{0, 1, 2, \dots, d\}$.
The definition also delivers insight about several fundamental examples.

\begin{example}[Linear subspaces] \label{ex:subspace}
Let $L_j$ be a $j$-dimensional subspace in $\R^d$.
Then $L_j$ is a polyhedral cone with precisely one
face, so the map $\Proj_{L_j}$ projects every point onto this $j$-dimensional face.  Thus,
$$
v_k(L_j) = \begin{cases}
	1, & k = j \\
	0, & k \neq j
\end{cases}
\quad\text{for $k = 0, 1, 2, \dots, d$.}
$$
Apply the intrinsic formulation~\eqref{eq:sdim-intr} of the statistical dimension
to confirm that $\delta(L_j) = j$. 
\end{example}

\begin{example}[The nonnegative orthant] \label{ex:orthant}
The nonnegative orthant $\R_+^d$ is a polyhedral cone.
The projection $\Proj_{\R_+^d}(\vct{g})$ lies in the relative interior of a $k$-dimensional
face of the orthant if and only if exactly $k$
coordinates of $\vct{g}$ are positive.  Each coordinate of $\vct{g}$
is positive with probability one-half and negative with probability one-half,
and the coordinates are independent.  Therefore, the intrinsic volumes of the orthant
are given by
$$
v_k\big(\R_+^d\big) = 2^{-d} {d \choose k}
\quad\text{for $k = 0, 1, 2, \dots, d$.}
$$
In other words, the intrinsic volumes coincide with the probability
density of a $\textsc{Binomial}(d, \half)$ random variable.
Apply the intrinsic formulation~\eqref{eq:sdim-intr} of the statistical dimension
to confirm that $\delta\big(\R_+^d\big) = \half d$.
\end{example}

Let us explain briefly how to extend the definition of conic intrinsic
volumes to the general case.
We can equip the family of closed convex cones in $\R^d$ with the
\term{conic Hausdorff metric}\footnote{The conic Hausdorff metric is obtained by identifying each closed convex cone $C \subset \R^d$
with the spherical convex set $C \cap \sphere{d-1}$.  Then we invoke the familiar construction of the
Hausdorff metric on the sphere.}
to form a compact metric space.
The polyhedral cones are dense in this metric space, and
the conic intrinsic volumes are continuous with respect to the metric.
Therefore, we may define the intrinsic volumes of a general closed convex cone
by approximation.

\begin{definition}[Intrinsic volumes: General case] \label{def:intr-vol-gen}
Let $C$ be a closed convex cone in $\R^d$, and let $\{ C_i : i = 1, 2, 3, \dots \} \subset \R^d$
be a sequence of polyhedral cones that converges to $C$ in the
conic Hausdorff metric.  For each $k = 0, 1, 2, \dots, d$, the $k$th
(conic) intrinsic volume $v_k(C)$ is given by the limit
$$
v_k(C) := \lim_{i \to \infty} v_k(C_i).
$$
It can be shown that this limit does not depend on the approximating sequence.
\end{definition}

\noindent
Let us warn the reader that the projection formula in Definition~\ref{def:intr-vol}
breaks down for a general closed convex cone because the limiting process
does not preserve facial structure.
To learn more about the construction behind Definition~\ref{def:intr-vol-gen},
see the book~\cite[Sec.~6.5]{scwe:08}, the thesis~\cite{am:thesis},
or the paper~\cite{McCTro:13}.
The spherical Steiner formula, Fact~\ref{fact:steiner},
provides an alternative geometric interpretation of the intrinsic volumes.

\subsection{Properties of conic intrinsic volumes}
\label{sec:int-vol-props}

The intrinsic volumes of a closed convex cone satisfy a number of important
relationships that we outline here.

\begin{fact}[Properties of intrinsic volumes]
Let $C$ be a closed convex cone in $\R^d$.  The
intrinsic volumes of the cone obey the following laws.
\begin{enumerate}
\item	\textbf{Distribution.}  The intrinsic volumes describe a probability
distribution on $\{0, 1, \dots, d\}$:
\begin{equation} \label{eq:probab-dist}
\sum_{k=0}^d v_k(C) = 1
\quad\text{and}\quad
v_k(C) \geq 0
\quad\text{for $k = 0, 1, 2, \dots, d$.}
\end{equation}

\item	\textbf{Polarity.}  The intrinsic volumes reverse under polarity:
\begin{equation} \label{eq:polarity-reverse}
v_k(C) = v_{d-k}(C^{\polar})
\quad\text{for $k = 0, 1, 2, \dots, d$.}
\end{equation}

\item	\textbf{Gauss--Bonnet formula.}  When $C$ is not a subspace,
\begin{equation} \label{eq:gauss-bonnet}
\sum_{\substack{k=0 \\ \text{$k$ even}}}^d v_k(C)
= \sum_{\substack{k=1 \\ \text{$k$ odd}}}^d v_k(C)
= \frac{1}{2}.
\end{equation}

\item	\textbf{Direct products.}  For each convex cone $K \subset \R^{d'}$,
\begin{equation} \label{eq:prod}
v_k(C \times K) = \sum_{i+j = k} v_i(C) \cdot v_j(K)
\quad\text{for $k = 0, 1, 2, \dots, d + d'$.}
\end{equation}
\end{enumerate}
\end{fact}

\noindent
The facts~\eqref{eq:probab-dist}, \eqref{eq:polarity-reverse},
and~\eqref{eq:gauss-bonnet} are drawn from~\cite[Sec.~6.5]{scwe:08}.
See~\cite[Prop.~4.4.13]{am:thesis} or~\cite[Cor.~5.1]{McCTro:13} for a
proof of the product rule~\eqref{eq:prod}.

\subsection{The spherical Steiner formula}
\label{sec:steiner}

We continue with a selection of more sophisticated results from conic integral
geometry.  These formulas provide detailed answers, expressed in terms of conic
intrinsic volumes, to the geometric questions posed at the beginning of
Section~\ref{sec:intr-conic-integr}.  To state the first result,
we introduce a family of geometric functions.

\begin{definition}[Tropic functions]
Let $L_k$ be a $k$-dimensional subspace of $\R^d$.  Define
\begin{equation} \label{eq:tropic}
I_k^d(\eps) := \Prob\big\{ \enormsm{\Proj_{L_k}(\vct{\theta})}^2 \geq \eps \big\}
\quad\text{for $\eps \in [0,1]$,}
\end{equation}
where $\vct{\theta}$ is uniformly distributed on the unit sphere in $\R^d$.
\end{definition}

\noindent
Basic geometric reasoning reveals that $I_k^d(\eps)$ is the proportion of
points on the sphere $\sphere{d-1}$ that lie within an angle 
$\arccos(\sqrt{\eps})$ of the subspace $L_k$.  Our terminology
derives from the approximate geographical fact that the tropics
lie within a fixed angle ($23^\circ \, 26'$) of the equator; the
usual term \term{regularized incomplete beta function} is longer and
less evocative.

The core fact in conic integral geometry is the
\term{spherical Steiner formula}~\cite{hot:39,weyl:39,herg:43,alle:48,sant:50},
which describes the fraction of points on the sphere that lie at most a
fixed angle from a closed convex cone.

\begin{fact}[Spherical Steiner formula] \label{fact:steiner}
Let $C$ be a closed convex cone in $\R^d$.  For each $\eps \in [0,1]$,
\begin{equation} \label{eq:steiner}
\Prob\big\{ \enormsm{\Proj_{C}(\vct{\theta})}^2 \geq \eps \big\}
	= \sum_{k=0}^d v_k(C) \, I_k^d(\eps).
\end{equation}
\end{fact}

\noindent
The spherical Steiner formula often serves as the \emph{definition} of
conic intrinsic volumes. The formula~\eqref{eq:steiner}
can also be derived from the definition here.
For a proof of Fact~\ref{fact:steiner} in the
spirit of this work, see~\cite[Thm.~6.5.1]{scwe:08} or~\cite[Prop.~3.4]{McCTro:13}.

\subsection{The conic kinematic formula}
\label{sec:kinematic}

Next, we present another major result from the theory of
conic integral geometry.  This statement involves partial sums of
the intrinsic volumes.

\begin{definition}[Tail functionals]
Let $C$ be a closed convex cone in $\R^d$.  For each $k = 0, 1, 2, \dots, d$, the
$k$th \term{tail functional} is given by
\begin{align} \label{eq:tail-fn}
t_k(C) &:= v_k(C) + v_{k+1}(C) + \dots = \sum_{\quad j=k\quad}^d v_j(C).
\intertext{The $k$th \term{half-tail functional} is defined as}
\label{eq:half-tail-fn}
h_k(C) &:= v_k(C) + v_{k+2}(C) + \dots = \sum_{\substack{j=k \\ \text{$j-k$ even}}}^d v_j(C).
\end{align}
\end{definition}

\noindent
The two types of tail functionals are related through the following interlacing
inequality.

\begin{proposition}[Interlacing] \label{prop:interlacing}
For each closed convex cone $C$ in $\R^d$ that is not a linear subspace,
$$
2 \, h_k(C) \geq t_k(C) \geq 2 \, h_{k+1}(C)
\quad\text{for each $k = 0, 1, 2, \dots, d-1$.}
$$
\end{proposition}

\noindent
We establish Proposition~\ref{prop:interlacing} in Appendix~\ref{sec:proof-backgr-results}.

With this notation, we can present a modern formulation of the
\term{conic kinematic formula}, which provides an exact expression
for the probability that a randomly oriented convex cone has a
nontrivial intersection with a fixed convex cone.

\begin{fact}[Conic kinematic formula] \label{fact:kinematic}
Let $C$ and $K$ be closed convex cones in $\R^d$, and assume that $C$ is not a subspace.  Then
\begin{align} \label{eq:kinematic}
\Prob\big\{ C \cap \mtx{Q} K \neq \{ \vct{0} \} \big\}
	&= 2 \, h_{d+1}(C \times K).
\intertext{For a linear subspace $L_{d-m}$ in $\R^d$ with dimension $d-m$, this expression
reduces to the \term{Crofton formula}}
  \label{eq:Crofton}
  \Prob\big\{ C \cap \mtx{Q} L_{d-m} \neq \{ \vct{0} \} \big\}
  &= 2 \, h_{m+1}(C).
\end{align}
\end{fact}

\noindent
The compact notation here disguises the equivalence
between~\eqref{eq:kinematic} and Fact~\ref{fact:kinem-intro}.
To verify this point, expand the half-tail functional
using~\eqref{eq:half-tail-fn} and apply the direct
product rule~\eqref{eq:sdim-direct-product}.
See~\cite[p.~261]{scwe:08} for a proof of Fact~\ref{fact:kinematic}.

\begin{remark}[Extended kinematic formula]
By induction, the kinematic formula generalizes to a family $\{ C, K_1, \dots, K_r \}$ of closed convex cones in $\R^d$.  If $C$ is not a subspace, then
\begin{equation} \label{eq:kinem-extend}
\Prob\big\{ C \cap \mtx{Q}_1 K_1 \cap \dots \cap \mtx{Q}_r K_r \neq \{ \vct{0} \} \big\}
	= 2 \, h_{rd + 1}(C \times K_1 \times \dots \times K_r).
\end{equation}
Each matrix $\mtx{Q}_i \in \R^{n \times n}$ is a random orthogonal basis, chosen independently
from the others.  This result can be used to analyze demixing problems with more than
two constituents. See the follow-up work~\cite{McCTro:13a} for details.
\end{remark}

\subsection{Characterizations of the statistical dimension}
\label{sec:revisit-stat-dimen}

In Section~\ref{sec:conic-phase}, we presented two different ways
of thinking about the statistical dimension.  Definition~\ref{def:sdim-int},
offers an intrinsic characterization in terms of the conic intrinsic volumes,
and it links the statistical dimension to the theory of conic integral geometry.
Proposition~\ref{def:sdim} provides a metric characterization that leads to
powerful tools for calculating the statistical dimension for specific cones.
The following result applies the spherical Steiner formula to verify that the
two formulations coincide.

\begin{proposition}[Statistical dimension: Equivalent characterizations] \label{prop:sdim-int-vols}
For each closed convex cone $C$ in $\R^d$, 
\begin{equation*} \label{eq:sdim=E[v]}
\delta(C) = \Expect \big[ \norm{\smash{\Proj_{C}(\vct{g})}}^2 \big]
	= \sum_{k=0}^d k \, v_k(C).
\end{equation*}
\end{proposition}

\begin{proof}
The Gaussian formulation~\eqref{eq:sdim-gauss} and the spherical formulation~\eqref{eq:sdim-circ-expect} of statistical dimension coincide, so
$$
\Expect \big[ \enormsm{ \Proj_C(\vct{g}) }^2 \big]
	= d \, \Expect \big[ \enormsm{ \Proj_C(\vct{\theta}) }^2 \big]
	= d \int_0^1 \Prob\big\{ \enormsm{ \Proj_C(\vct{\theta}) }^2 \geq \eps \big\} \idiff{\eps}.
$$
We have used integration by parts to express the expectation as an integral of tail probabilities.
The Steiner formula~\eqref{eq:steiner} and the definition~\eqref{eq:tropic} of the tropic function
allow us to write the probability as a sum:
$$
\Expect \big[ \enormsm{ \Proj_C(\vct{g}) }^2 \big]
 	= d \sum_{k=0}^d v_k(C)
	\int_0^1 \Prob\big\{ \enormsm{ \Proj_{L_k}(\vct{\theta}) }^2 \geq \eps \big\} \idiff{\eps}
	= \sum_{k=0}^d v_k(C) \, \big( d \, \Expect\big[ \enormsm{ \Proj_{L_k}(\vct{\theta}) }^2 \big] \big)
	= \sum_{k=0}^d k \, v_k(C).
$$
where $L_k$ is an arbitrary $k$-dimensional subspace.  The last identity follows from elementary
geometric reasoning.
\end{proof}

\subsection{The statistical dimension is canonical}
\label{sec:sdim-canon}

The intrinsic characterization of the statistical dimension, Definition~\ref{def:sdim-int}, has a significant consequence
from the point of view of integral geometry.  We summarize the ideas for the benefit of geometers; other readers may prefer to skip this material.

Let $\mathscr{C}_d$ denote the family of closed convex cones in $\R^d$, equipped with the conic Hausdorff metric to form a compact metric space.
A geometric functional $v : \mathscr{C}_d \to \R$ is called a \term{continuous, rotation-invariant valuation} if it satisfies the properties

\begin{enumerate} \setlength{\itemsep}{1mm}
\item	\textbf{Valuation I.}  For the trivial cone, $v(\{\vct{0}\}) = 0$.
\item	\textbf{Valuation II.}	If $C, K \in \mathscr{C}_d$ and $C \cup K \in \mathscr{C}_d$, then $v(C \cup K) + v(C \cap K) = v(C) + v(K)$.
\item	\textbf{Rotation invariance.}	For each $C \in \mathscr{C}_d$,	we have $v(\mtx{U} C) = v(C)$ for each orthogonal matrix $\mtx{U} \in \R^{d \times d}$.
\item	\textbf{Continuity.} 	If $C_i \to C$ in $\mathscr{C}_d$, then $\lim_{i \to \infty} v(C_i) \to v(C)$.
\end{enumerate}

\noindent
Continuous, rotation-invariant valuations are natural geometric measures defined on convex cones.  Many of the valuations that arise in conic geometry are also \term{localizable}, which is a subtle technical property~\cite[p.~254]{scwe:08}.

In particular, each intrinsic volume $v_k$ is a continuous, rotation-invariant, localizable valuation on the set of closed convex cones.  It follows from the intrinsic formulation~\eqref{eq:sdim-intr} that the statistical dimension inherits these technical properties.
It is known that each continuous, rotation-invariant, and localizable valuation on $\mathscr{C}_d$ is determined by the values it takes on linear subspaces; see \cite[Satz 4.2.2]{Gl},~\cite[Thm.~5]{Gl:Summ}, or~\cite[Thm.~6.5.4]{scwe:08}.  Therefore,
\begin{quotation} \it
The statistical dimension $\delta$ is the \emph{unique} continuous, rotation-invariant, localizable valuation on the set $\mathscr{C}_d$ of closed convex cones that satisfies $\delta(L) = \dim(L)$ for each subspace $L$.
\end{quotation}
In other words, the statistical dimension canonically extends the linear dimension to the class of closed convex cones.

The long-standing \term{spherical Hadwiger conjecture} states that the condition of localizability is unnecessary here.  More precisely, the conjecture posits that every continuous and rotation-invariant valuation on the set of closed convex cones can be expressed as a linear combination of the conic intrinsic volumes.  For a discussion of the spherical Hadwiger conjecture, see the works~\cite[p.~976]{McMu:93},~\cite[Sec.~11.5]{KR:97}, and~\cite[p.~263]{scwe:08}.  The conjecture currently stands open for $d \geq 4$.

\section{Intrinsic volumes concentrate at the statistical dimension} \label{sec:conc-intr-volum}

The main technical result in this paper describes a new property of conic intrinsic volumes: The intrinsic volumes of a closed convex cone concentrate near the statistical dimension of the cone on a scale determined by the statistical dimension.
This phenomenon is depicted in Figure~\ref{fig:int-vol-conc}.

\begin{theorem}[Concentration of intrinsic volumes] \label{thm:main-conc}
Let $C$ be a closed convex cone. Define the transition width
$$
\omega(C) := \sqrt{\delta(C) \wedge \delta(C^\polar)},
$$
and introduce the function
\begin{equation} \label{eq:volume-tail}
p_C( \lambda ) := 4 \exp\left( \frac{-\lambda^2/8}{\omega^2(C) + \lambda} \right)
\quad\text{for $\lambda \geq 0$.}
\end{equation}
Then
\begin{align}
k_- \leq \delta(C) - \lambda + 1
&\quad\Longrightarrow\quad
t_{k_-}(C) \geq 1 - p_C(\lambda); \label{eq:tk-} \\
k_+ \geq \delta(C) + \lambda \phantom{\ +1}
&\quad\Longrightarrow\quad
t_{k_+}(C) \leq \phantom{1 - \ } p_C(\lambda). \label{eq:tk+}
\end{align}
The tail functional $t_k$ is defined in~\eqref{eq:tail-fn}.
The operator $\wedge$ returns the minimum of two numbers.
\end{theorem}

See Section~\ref{sec:conc-discuss} for a discussion of Theorem~\ref{thm:main-conc}.
In Section~\ref{sec:heur-proof} and~\ref{sec:proof-main-conc},
we summarize the intuition behind the proof,
and we follow up with the technical details.  Later, Sections~\ref{sec:kinem-cons-conc}--\ref{sec:appl-permutahedron} highlight applications in conic geometry, optimization theory, and signal processing.  The follow-up work~\cite{McCTro:13} contains some improvements on Theorem~\ref{thm:main-conc}.

\subsection{Discussion} \label{sec:conc-discuss}

Theorem~\ref{thm:main-conc} states that the sequence $\{ t_k(C) : k = 0, 1, 2, \dots, d \}$ of tail functionals drops from one to zero near the statistical dimension $\delta(C)$, and the transition occurs over a range of $O(\omega(C))$ indices.  Owing to the fact~\eqref{eq:probab-dist} that the intrinsic volumes form a probability distribution, we must conclude that the intrinsic volumes $v_k(C)$ are all negligible in size, except for those whose index $k$ is close to the statistical dimension $\delta(C)$.  We learn that the intrinsic volumes of a convex cone $C$ with statistical dimension $\delta(C)$ are qualitatively similar with the intrinsic volumes of a subspace with dimension about $\delta(C)$.

Theorem~\ref{thm:main-conc} contains additional information about the rate at which the tail functionals of a cone transit from one to zero.  To extract this information, it helps to note the weaker inequality
\begin{equation} \label{eq:pC-bd}
p_C(\lambda) \leq \begin{cases}
	4 \, \econst^{-\lambda^2/(16\,\omega^2(C))}, & 0 \leq \lambda \leq \omega^2(C) \\
	4 \, \econst^{-\lambda/16}, & \phantom{0 \leq \; } \lambda \geq \omega^2(C).
\end{cases}
\end{equation}
We see that~\eqref{eq:tk-} and~\eqref{eq:tk+} are vacuous until $\lambda \approx 4 \, \omega(C)$.  As $\lambda$ increases, the function $p_C(\lambda)$ decays like the tail of a Gaussian random variable with standard deviation $\leq 2\sqrt{2} \, \omega(C)$.  When $\lambda$ reaches $\omega^2(C)$, the decay slows to match the tail of an exponential random variable with mean $\leq 16$.  In particular, the behavior of the tail functionals depends on the intrinsic properties of the cone, rather than the ambient dimension.

\subsection{Heuristic proof of Theorem~\ref{thm:main-conc}}
\label{sec:heur-proof}

The basic ideas behind the argument are easy to summarize, but the details demand some effort.  Let $C$ be a closed convex cone in $\R^d$.  Recall the spherical Steiner formula~\eqref{eq:steiner}:
\begin{equation} \label{eq:conc-proof-heur}
\Prob\big\{ \enormsm{\Proj_C(\vct{\theta})}^2 \geq \eps \big\}
	= \sum_{k=0}^d v_k(C) \, I_k^d(\eps),
\end{equation}
where $\vct{\theta}$ is uniformly distributed on the sphere $\sphere{d-1}$
and the tropic function $I_k^d$ is defined in~\eqref{eq:tropic}.

Concentration of measure on the sphere implies that the random variable $\enormsm{\Proj_C(\vct{\theta})}^2$ is typically very close to its expected value $\delta(C)/d$, determined by~\eqref{eq:sdim-circ-expect}.  Thus, the left-hand side of~\eqref{eq:conc-proof-heur} is very close to one when $\eps d < \delta(C)$ and very close to zero when $\eps d > \delta(C)$.

As for the right-hand side of~\eqref{eq:conc-proof-heur}, recall that the tropic function $I_k^d(\eps)$ is the proportion of points on the sphere within a distance of $\sqrt{1 - \eps}$ from a fixed $k$-dimensional subspace.  Once again, concentration of measure ensures that $I_k^d(\eps)$ is close to zero when $k < \eps d$ and close to one when $k > \eps d$.  Therefore, the sum on the right-hand side of~\eqref{eq:conc-proof-heur} is approximately equal to the tail functional $t_{\eps d}(C)$.

Combining these two observations, we conclude that the sequence $\{ t_k(C) : k = 0, 1, 2, \dots, d \}$ of tail functionals makes a sharp transition from one to zero when $k \approx \delta(C)$.  It remains to make this reasoning rigorous and to determine the range of $k$ over which the transition takes place.

\subsection{Proof of Theorem~\ref{thm:main-conc}} \label{sec:proof-main-conc}

Let $C$ be a closed convex cone in $\R^d$, and define $\eps := k_+/d$.  The first part of the argument requires a technical lemma that we prove in Appendix~\ref{sec:proof-backgr-results}.  This result quantifies how much of the sphere in $\R^d$ lies within an angle $\arccos(k/d)$ of a $k$-dimensional subspace.

\begin{lemma}[The tropics] \label{lem:median-mean}
For all integers $0 \leq k \leq d$, the tropic function $I_{k}^d(k/d) \geq 0.3$.
\end{lemma}

\noindent 
We begin by expressing the tail functional $t_{k_+}(C)$ in terms of the probability that a spherical variable lies near the cone $C$.
\begin{align}
t_{k_+}(C) &= \sum_{k = k_+}^d v_k(C) \big[ I_k^d(\eps) + \big(1 - I_{k}^d(\eps)\big) \big] \notag \\
	&\leq \sum_{k=0}^d v_k(C) \, I_k^d(\eps)
	+ \big(1 - I_{k_+}^d(\eps) \big) \sum_{k=k_+}^d v_k(C) \notag \\
	&\leq \Prob\big\{ \enormsm{\Proj_C(\vct{\theta})}^2 \geq \eps \big\} + 0.7 \, t_{k_+}(C).
	\label{eq:tail-tail}
\end{align}
The first identity is the definition~\eqref{eq:tail-fn} of the tail function.
To reach the second inequality, we inspect the definition~\eqref{eq:tropic} to
see that $I_{k}^d(\eps)$ is a decreasing function of $k$ when the other parameters are fixed.  In the third line, we invoke the Steiner formula~\eqref{eq:steiner} to rewrite the first sum.  The bound for the second sum follows and the definition~\eqref{eq:tail-fn} of the tail functional and from Lemma~\ref{lem:median-mean} seeing that $\eps = k_+/d$.  

Rearranging~\eqref{eq:tail-tail}, we obtain the bound
\begin{equation} \label{eq:tail-fn-bd}
t_{k_+}(C) \leq 4 \, \Prob\big\{ d \, \enormsm{\Proj_C(\vct{\theta})}^2 \geq d \eps \big\}
\leq 4 \, \Prob\big\{ d \, \enorm{\Proj_C(\vct{\theta})}^2 \geq \delta(C) + \lambda \big\}.
\end{equation}
The last inequality depends on the fact that $\eps = k_+/d$ and the definition~\eqref{eq:tk+} of $k_+$.  In other words, the tail functional is dominated by the probability that a random point on the sphere is close to the cone.

To estimate the probability in~\eqref{eq:tail-fn-bd}, we need a tail bound for the squared norm of the projection of a spherical variable onto a cone.  This result is encapsulated in the following lemma.  The approach is more or less standard, so we defer the details to Appendix~\ref{sec:proof-backgr-results}.

\begin{lemma}[Tail bound for conic projections] \label{lem:concentration}
For each closed convex cone $C$ in $\R^d$,
\begin{equation} \label{eq:conc-claim}
\Prob\left\{ d \, \enormsm{ \Proj_C(\vct{\theta}) }^2 \geq
	\delta(C) + \lambda \right\} 
	\leq 
	\exp\left( \frac{-\lambda^2/8}{\omega^2(C) + \lambda} \right)
	\quad\text{for $\lambda \geq 0$.}
\end{equation}
\end{lemma}

\noindent
Introducing~\eqref{eq:conc-claim} into~\eqref{eq:tail-fn-bd}, we reach the upper bound~\eqref{eq:tk+} on the tail functional.

To develop the lower bound~\eqref{eq:tk-} on the tail functional $t_{k_-}(C)$,
we use a polarity argument.  Note that
\begin{equation} \label{eq:tk-bd1}
t_{k_-}(C) = \sum_{k=k_-}^d v_k(C) = \sum_{k=0}^{d-k_-} v_k(C^\polar) = 1 - t_{d-k_- + 1}(C^\polar).
\end{equation}
The first identity is the definition~\eqref{eq:tail-fn} of the tail functional $t_{k_-}(C)$.  The second
relation holds because of the fact~\eqref{eq:polarity-reverse} that polarity reverses intrinsic volumes, and the last part relies on~\eqref{eq:tail-fn} and the property~\eqref{eq:probab-dist} that the intrinsic volumes sum to one.  Owing to the complementarity law~\eqref{eq:sdim-polarity-sum} and the definition~\eqref{eq:tk-} of $k_-$,
$$
d - k_- + 1 = \delta(C^\polar) + \delta(C) - k_- + 1
	\geq \delta(C^\polar) + \lambda.
$$
Therefore, we may apply~\eqref{eq:tk+} to obtain an upper bound on the tail functional $t_{d-k_-+1}(C^\polar)$.  Substitute this bound into~\eqref{eq:tk-bd1} to establish the lower bound on the tail functional $t_{k_-}(C)$ stated in~\eqref{eq:tk-}.

\section{Approximate kinematic bounds}
\label{sec:kinem-cons-conc}

We are now prepared to establish an approximate version of the conic kinematic formula, expressed in terms of the statistical dimension.  Most of the applied results in this paper ultimately depend on this theorem.  The proof combines the exact kinematic formula~\eqref{eq:kinematic} with the concentration of intrinsic volumes, guaranteed by Theorem~\ref{thm:main-conc}.

\begin{theorem}[Approximate kinematics] \label{thm:approx-kinem}
Assume that $\lambda \geq 0$.
Let $C$ be a convex cone in $\R^d$.
For a $(d-m)$-dimensional subspace $L_{d-m}$, it holds that
\begin{equation}
\begin{split} \label{eq:subspace-kinem}
m \geq \delta(C) + \lambda
&\quad\Longrightarrow\quad
\Prob\big\{ C \cap \mtx{Q}L_{d-m} \neq \{ \vct{0} \} \big\}
	\leq \phantom{1 - \;\, }p_{C}(\lambda); \\
m \leq \delta(C) - \lambda
&\quad\Longrightarrow\quad
\Prob\big\{ C \cap \mtx{Q}L_{d-m} \neq \{ \vct{0} \} \big\} \geq 1 - p_{C}(\lambda).
\end{split}
\end{equation}
For any convex cone $K$ in $\R^d$, it holds that
\begin{equation} \label{eq:cone-kinem}
\begin{split}
\delta(C) + \delta(K) \leq d - 2\lambda
&\quad\Longrightarrow\quad
\Prob\big\{ C \cap \mtx{Q}K \neq \{ \vct{0} \} \big\}
	\leq \phantom{1 - (} p_{C}(\lambda) + p_{K}(\lambda); \\
\delta(C) + \delta(K) \geq d + 2\lambda 
&\quad\Longrightarrow\quad
\Prob\big\{ C \cap \mtx{Q}K \neq \{ \vct{0} \} \big\}
	\geq 1 - (p_{C}(\lambda) + p_{K}(\lambda)).
\end{split}
\end{equation}
The functions $p_C$ and $p_K$ are defined by the expression~\eqref{eq:volume-tail}.
\end{theorem}

We discuss Theorem~\ref{thm:approx-kinem} below in Section~\ref{sec:discuss-approx-kinem}.
The proof appears in Section~\ref{sec:pf-approx-kinem}.  In Section~\ref{sec:pf-kinem},
we derive Theorem~\ref{thm:kinematic} from a similar, but slightly
easier argument.  The follow-up work~\cite{McCTro:13a}
contains an improvement on Theorem~\ref{thm:approx-kinem}.

\subsection{Discussion} \label{sec:discuss-approx-kinem}

Theorem~\ref{thm:approx-kinem} has an attractive interpretation.  The first statement~\eqref{eq:subspace-kinem} shows that a randomly oriented subspace with codimension $m$ is unlikely to share a ray with a fixed cone $C$, provided that the codimension $m$ is larger than the statistical dimension $\delta(C)$ of the cone.  When the codimension $m$ is smaller than the statistical dimension $\delta(C)$, the subspace and the cone are likely to share a ray.

The transition in behavior expressed in~\eqref{eq:subspace-kinem} takes place when the codimension $m$ of the subspace changes by about $\omega(C) = \sqrt{\delta(C) \wedge \delta(C^\polar)}$.  This point explains why the empirical success curves taper in the corners of the graphs in Figure~\ref{fig:lin-inv-1}.  Indeed, on the bottom-left side of each panel, the relevant descent cone is small; on the top-right side of each panel, the descent cone is large, so its polar is small.  In these regimes, the result~\eqref{eq:subspace-kinem} shows that the phase transition must occur over a narrow range of codimensions.

The second statement~\eqref{eq:cone-kinem} provides analogous results for the probability that a randomly oriented cone shares a ray with a fixed cone.  This event is unlikely when the total statistical dimension of the two cones is smaller than the ambient dimension; it is likely to occur when the total statistical dimension exceeds
the ambient dimension.

For the case of two cones, it is harder to analyze the size of the transition region.
Since the probability bounds in~\eqref{eq:cone-kinem}
are controlled by the sum $p_C(\lambda) + p_K(\lambda)$,
we can only be certain that the probability estimate
depends on the larger of the two quantities.
It follows that the width of the transition does not exceed the
\emph{larger} of $\omega(C) = \sqrt{\delta(C) \wedge \delta(C^\polar)}$
and $\omega(K) = \sqrt{\delta(K) \wedge \delta(K^\polar)}$.
This observation is sufficient to explain why the empirical
success curves taper at the top-left and bottom-right of the graphs in
Figure~\ref{fig:MCA-phase}.  Indeed, these are the regions
where one of the descent cones is small and the polar
of the other descent cone is small.

\subsection{Proof of Theorem~\ref{thm:approx-kinem}}
\label{sec:pf-approx-kinem}

We may assume that $C$ and $K$ are both closed because $\Prob\{ C \cap \mtx{Q} K \neq \{\vct{0}\} \} = \Prob\{ \overline{C} \cap \mtx{Q} \overline{K} \neq \{\vct{0}\} \}$, where the overline denotes closure.  This is a subtle point that follows from the discussion of touching probabilities located in~\cite[pp.~258--259]{scwe:08}.

Let us begin with the first set~\eqref{eq:subspace-kinem} of results, concerning the probability that a randomly oriented subspace intersects a fixed cone along a ray.  Consider the first implication, which operates when $m \geq \delta(C) + \lambda$. The implication clearly holds when \(C\) is a subspace.  When \(C\) is not a subspace, the Crofton formula~\eqref{eq:Crofton} shows that
$$
\Prob\big\{ C \cap \mtx{Q}L \neq \{ \vct{0} \} \big\}
	= 2 \, h_{m+1}(C)
	\leq t_m(C),
$$
where the inequality depends on the interlacing result, Proposition~\ref{prop:interlacing}.  The concentration of intrinsic volumes, Theorem~\ref{thm:main-conc}, demonstrates that the tail functional satisfies the bound
$$
t_m(C) \leq p_{C}(\lambda)
\quad\text{when}\quad
m \geq \delta(C) + \lambda.
$$
This completes the first bound.  The second result, which holds when $m \leq \delta(C) - \lambda$, follows from a parallel argument.

The conic kinematic formula is required for the second set~\eqref{eq:cone-kinem} of results, which concern the probability that a randomly oriented cone intersects a fixed cone nontrivially.
Consider the situation where $\delta(C) + \delta(K) \leq d - 2\lambda$.  The kinematic formula~\eqref{eq:kinematic} yields
\begin{equation} \label{eq:kinem-in-action}
\Prob\big\{ C \cap \mtx{Q}K \neq \{ \vct{0} \} \big\}
	= 2 \, h_{d+1}(C \times K) \leq t_d(C \times K),
\end{equation}
where the inequality follows from the interlacing result, Proposition~\ref{prop:interlacing}.

We rely on a simple lemma to bound the tail functional of the product in
terms of the individual tail functionals.

\begin{lemma}[Tail functionals of a product] \label{lem:prod-tail}
Let $C$ and $K$ be closed convex cones.  Then
$$
t_{\lceil \delta(C) + \delta(K) + 2 \lambda \rceil}(C \times K)
	\leq t_{\lceil \delta(C) + \lambda \rceil}(C)
	+ t_{\lceil \delta(K) + \lambda \rceil}(K).
$$
\end{lemma}

\noindent
  The proof appears in
Appendix~\ref{sec:proof-backgr-results}.  See the follow-up work~\cite{McCTro:13}
for an improvement on this result.

Since the tail functionals are weakly decreasing,
our assumption that $\delta(C) + \delta(K) \leq d - 2\lambda$
implies that
$$
t_d(C \times K)
	\leq t_{\lceil \delta(C) + \delta(K) + 2\lambda \rceil}(C \times K)
	\leq t_{\lceil \delta(C) + \lambda \rceil}(C)
	+ t_{\lceil \delta(K) + \lambda \rceil}(K).
$$
Theorem~\ref{thm:main-conc} delivers an upper bound of
$p_C(\lambda) + p_K(\lambda)$ for the right-hand side.  Introduce these
bounds into the probability inequality~\eqref{eq:kinem-in-action}
to complete the proof of the first statement in~\eqref{eq:cone-kinem}.
The second result follows from an analogous argument. \qed

\subsection{Proof of Theorem~\ref{thm:kinematic}}
\label{sec:pf-kinem}

The simplified kinematic bound of Theorem~\ref{thm:kinematic} involves
an argument similar with the proof of Theorem~\ref{thm:approx-kinem}.
First, assume that $\delta(C) + \delta(K) \leq d - \lambda$.
As before, the kinematic formula~\eqref{eq:kinematic}
and the interlacing result, Proposition~\ref{prop:interlacing},
ensure that
\begin{equation} \label{eq:kinem-in-action-2}
\Prob\big\{ C \cap \mtx{Q}K \neq \{ \vct{0} \} \big\}
	= 2 \, h_{d+1}(C \times K) \leq t_d(C \times K).
\end{equation}
The product rule~\eqref{eq:sdim-direct-product} for cones shows that
$\delta(C \times K) = \delta(C) + \delta(K) \leq d - \lambda$, so the implication~\eqref{eq:tk+}
in Theorem~\ref{thm:main-conc} yields
\begin{equation} \label{eq:simp-kinem-1}
t_{d}(C\times K)
	\leq p_{C \times K}(\lambda)
	\leq 4 \exp\left( \frac{-\lambda^2 /8}{\delta(C \times K) + \lambda} \right)
	\leq 4\,\econst^{-\lambda^2/(8d)}.
\end{equation}
Substitute the inequality~\eqref{eq:simp-kinem-1} into the kinematic bound~\eqref{eq:kinem-in-action-2}.  Then make the change of variables $\lambda \mapsto a_{\eta} \sqrt{d}$, where $a_{\eta} := \sqrt{8\log(4/\eta)}$, to obtain the estimate
$$
\Prob\big\{ C \cap \mtx{Q}K \neq \{ \vct{0} \} \big\} \leq \eta.
$$
This establishes the first part of Theorem~\ref{thm:kinematic}.
The argument for the second part follows the same pattern. \qed

\section{Application: Cone programs with random constraints}
\label{sec:feas-cone-progr}

The concentration of intrinsic volumes has far-reaching consequences for the theory of optimization.  This section describes a new type of phase transition phenomenon that appears in a cone program with random affine constraints.  We begin with a theoretical result, and then we exhibit some numerical examples that confirm the analysis.

\subsection{Cone programs}

A \term{cone program} is a convex optimization problem with the following structure:
\begin{equation} \label{eq:cone-prog}
	\minimize \ip{ \vct{u} }{ \vct{x} }
	\subjto \mtx{A}\vct{x} = \vct{b}
	\quad\text{and}\quad \vct{x} \in C,
\end{equation}
where $C \subset \R^d$ is a closed convex cone.  The decision variable $\vct{x} \in \R^d$,
and the problem data consists of a vector $\vct{u} \in \R^d$, a matrix $\mtx{A} \in \R^{m \times d}$, and another vector $\vct{b} \in \R^m$.
This formalism includes several fundamental classes of
convex programs:

\vspace{1mm}

\begin{enumerate} \setlength{\itemsep}{1mm}
\item	\textbf{Linear programs.}  If $C = \R_+^d$, then~\eqref{eq:cone-prog} reduces to a linear program in standard form.

\item	\textbf{Second-order cone programs.}  If $C = \mathbb{L}^{d+1}$, then~\eqref{eq:cone-prog} is a type of second-order cone program.

\item	\textbf{Semidefinite programs.}  When $C = \mathbb{S}_+^{n \times n}$, we recover the class of (real) semidefinite programs.
\end{enumerate}

\vspace{1mm}

\noindent
In addition to their flexibility and modeling power, cone programs enjoy effective algorithms and a crisp theory.  We refer to~\cite{B-TN:01} for further details.

The cone program~\eqref{eq:cone-prog} can exhibit several interesting behaviors.  Let us remind the reader of the terminology.  A point $\vct{x}$ that satisfies the constraints $\mtx{A}\vct{x} = \vct{b}$ and $\vct{x} \in C$ is called a \term{feasible point}, and the cone program is \term{infeasible} when no feasible point exists.  The cone program is \term{unbounded} when there exists a sequence $\{\vct{x}_k\}$ of feasible points with the property $\ip{\vct{u}}{\vct{x}_k} \to - \infty$.

\begin{figure}[t]
  \centering
  \includegraphics[width=\columnwidth]{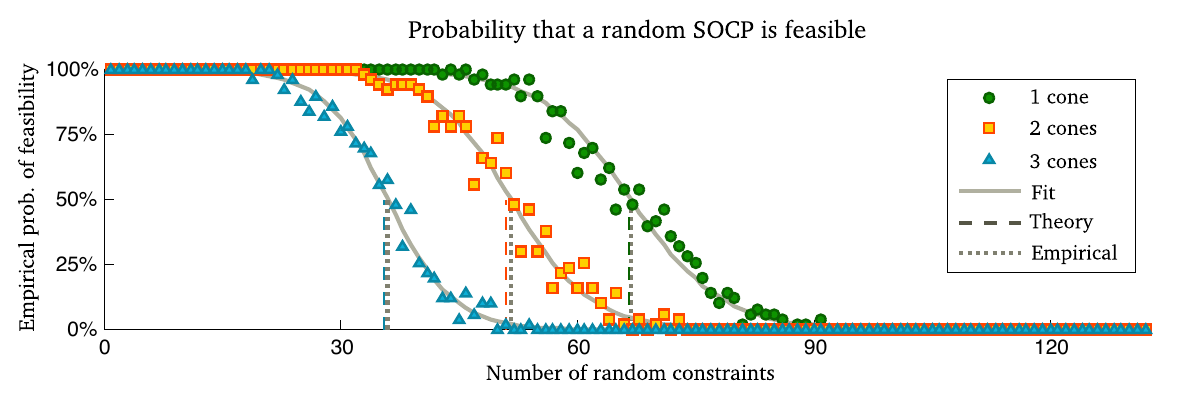}
  \caption{\textbf{Phase transitions in random cone programs.}
  For each cone $C_i$ in~\eqref{eq:cones}, we plot
  the empirical probability that the cone program~\eqref{eq:cone-prog}
  with random affine constraints is feasible.  The solid gray curve
  traces the logistic fit to the data, and the finely dashed line is the empirical
  $50\%$ success threshold, computed from the regression model.  The
  coarsely dashed line marks the statistical dimension $\delta(C_i)$,
  which is the theoretical estimate for the location of the phase
  transition.}
  \label{fig:SOCP}
\end{figure}

Our theory allows us analyze the properties of a random cone program.
It turns out that the number $m$ of
affine constraints controls whether the cone program is infeasible or unbounded.

\begin{theorem}[Phase transitions in cone programming] \label{thm:cone-prog}
Let $C$ be a closed convex cone in $\R^d$.  Consider the cone program~\eqref{eq:cone-prog} where $\vct{b}$ is a fixed nonzero vector, while the vector $\vct{u} \in \R^d$ and the matrix $\mtx{A} \in \R^{m \times d}$ have independent standard normal entries.  Then
\begin{align*}
m \leq \delta(C) - \lambda
	&\quad\Longrightarrow\quad
	\text{\eqref{eq:cone-prog} is unbounded with probability $\geq 1 - p_C(\lambda)$}; \\
m \geq \delta(C) + \lambda 	&\quad\Longrightarrow\quad
	\text{\eqref{eq:cone-prog} is infeasible with probability $\geq 1 - p_C(\lambda)$}.
\end{align*}
The function $p_C$ is defined by the expression~\eqref{eq:volume-tail}.
\end{theorem}

\begin{proof}
Amelunxen \& B{\"u}rgisser~\cite[Thm.~1.3]{ambu:11c} have shown that the intrinsic volumes of the cone $C$ control the properties of the random cone program~\eqref{eq:cone-prog}:
\begin{align*}
   \Prob\{\text{\eqref{eq:cone-prog} is infeasible}\} & = 1 - t_m(C);
\\ \Prob\{\text{\eqref{eq:cone-prog} has a unique minimizer}\} & = v_m(C);
\\ \Prob\{\text{\eqref{eq:cone-prog} is unbounded}\} & = t_{m+1}(C).
\end{align*}
We apply Theorem~\ref{thm:main-conc} to see that the tail functional $t_{m+1}(C)$ is extremely close to one when the number $m$ of constraints is smaller than the statistical dimension $\delta(C)$.  Likewise, $t_{m}(C)$ is extremely close to zero when the number $m$ of constraints is larger than the statistical dimension.  We omit the details, which are analogous with the proof of Theorem~\ref{thm:approx-kinem}.
\end{proof}

\subsection{A numerical example}

We have conducted a computer experiment to compare the predictions of Theorem~\ref{thm:cone-prog} with the empirical behavior of a generic cone program.  For this purpose, we study some random second-order cone programs.
In each case, the ambient dimension $d = 396$, and we consider three options for the cone $C$ in~\eqref{eq:cone-prog}:
\begin{subequations} \label{eq:cones}
\begin{align}
	C_1 &:= \Circ_d(\alpha_1); \\
	C_2 &:= \Circ_{d/2}(\alpha_2) \times \Circ_{d/2}(\alpha_2); \\
	C_3 &:= \Circ_{d/3}(\alpha_3) \times \Circ_{d/3}(\alpha_3) \times \Circ_{d/3}(\alpha_3).
\end{align}
\end{subequations}
The angles satisfy $\tan^2(\alpha_1) = \tfrac{1}{5}$ and $\tan^2(\alpha_2) = \tfrac{1}{7}$ and $\tan^2(\alpha_3) = \tfrac{1}{11}$.
Using the product rule~\eqref{eq:sdim-direct-product} and the integral
expression~\eqref{eq:sdim-circ} for the statistical dimension of a circular cone,
numerical quadrature yields
$$
\delta(C_1) \approx 66.67; \quad
\delta(C_2) \approx 51.00; \quad
\delta(C_3) \approx 35.50.
$$
Theorem~\ref{thm:cone-prog} indicates that a cone program~\eqref{eq:cone-prog} with the cone $C_i$ and generic constraints is likely to be feasible when the number $m$ of affine constraints is smaller than $\delta(C_i)$; it is likely to be infeasible when the number $m$ of affine constraints is larger than $\delta(C_i)$.

We can test this prediction numerically.  For each $i = 1, 2, 3$ and each $m \in \big\{ 1, 2, 3, \dots, \big[ \tfrac{1}{3} d \big] \big\}$, we perform the following steps 50 times:
\begin{enumerate}
\item	Independently draw a standard normal matrix $\mtx{A} \in \R^{m \times d}$ and standard normal $\vct{u} \in \R^d$ and $\vct{b} \in \R^m$.

\item	Use the \textsc{Matlab} package \texttt{CVX} to solve the cone program~\eqref{eq:cone-prog} with $C = C_i$.

\item	Report failure if \texttt{CVX} declares the cone program infeasible.
\end{enumerate}
For each $i = 1, 2, 3$, Figure~\ref{fig:SOCP} displays the empirical success probability, along with a logistic fit (Appendix~\ref{sec:logistic}).  We also mark the theoretical estimate for the location of the phase transition, which equals the statistical dimension $\delta(C_i)$.  Table~\ref{tab:locations-SOCP} reports the discrepancy between the theoretical and empirical behaviors.

\begin{table}[t]
  \centering
  \caption{\textbf{Phase transitions in random cone programs.}
  For each cone $C_i$ listed in~\eqref{eq:cones}, we compare the theoretical
  location of the phase transition, equal to the statistical
  dimension $\delta(C_i)$, with the empirical location $\mu_i$,
  computed from the logistic regression model in Figure~\ref{fig:SOCP}.  The last
  column lists the errors $\abs{\delta(C_i) - \mu_i}/d$, relative to the dimension $d = 396$ of the
  problem.}

    \begin{tabu}{ccccc}
    \toprule Cone & 
    Theoretical & Empirical &
    Error \\
    \midrule
    $C_1$ & 
    $66.67$ & $66.88$ &	 $0.054 \%$ \\
    $C_2$	& 
    $51.00$ & $51.70$ & $0.177\%$ \\
    $C_3$ &	 
    $35.50$ & $36.06$ & $0.141\%$ 
    \\
    \bottomrule
  \end{tabu}
  \label{tab:locations-SOCP}
\end{table}

\section{Application: Vectors from lists?}
\label{sec:appl-permutahedron}

This section describes a situation where our results prove that a particular linear inverse problem does \emph{not} provide an effective way to recover a structured vector.  Indeed, a significant contribution of our theory, which has no parallel in the current literature, is that we can obtain negative results as well as positive results.

In~\cite[Sec.~2.2]{CRPW:12}, Chandrasekaran et al.~propose a method for recovering a vector from an unordered list of its entries, along with some linear measurements.  Here is one way to frame this problem.  Suppose that $\vct{x}_0 \in \R^d$ is an unknown vector.  We are given the vector $\vct{y}_0 = \vct{x}_0^\downarrow$, whose entries list the components of $\vct{x}_0$ in weakly decreasing order.  We also collect data $\vct{z}_0 = \mtx{A} \vct{x}_0$ where $\mtx{A}$ is an $m \times d$ matrix. To identify $\vct{x}_0$, we must solve a structured linear inverse problem.

To solve this problem, Chandrasekaran et al.~propose to use a convex regularizer $f$ that exploits the information in $\vct{y}_0$.  They consider the Minkowski gauge of the permutahedron~\eqref{eq:perm} generated by $\vct{y}_0$.
$$
\norm{ \vct{x} }_{\coll{P}(\vct{y}_0)}
	:= \inf\big\{ \tau > 0 : \vct{x} \in \tau \, \coll{P}(\vct{y}_0) \big\},
$$
and they frame the regularized linear inverse problem
\begin{equation} \label{eq:perm-inv}
\minimize \ \norm{ \vct{x} }_{\coll{P}(\vct{y}_0)}
\subjto \vct{z}_0 = \mtx{A}\vct{x}.
\end{equation}
It is natural to ask how many linear measurements we need to be able to solve this inverse problem reliably.  Our theory allows us to answer this question decisively when the measurements are random.

\begin{proposition}[Vectors from lists?] \label{prop:vec-list}
Let $\vct{x}_0 \in \R^d$ be a fixed vector with distinct entries.  Suppose we are given the data $\vct{y}_0 = \vct{x}_0^\downarrow$ and $\vct{z}_0 = \mtx{A} \vct{x}_0$, where the matrix $\mtx{A} \in \R^{m \times d}$ has standard normal entries.  In the range
$0 \leq \lambda < \sqrt{\mathrm{H}_d}$, it holds that
\begin{align*}
m \leq d - \mathrm{H}_d - \lambda \sqrt{\mathrm{H}_d}
&\quad\Longrightarrow\quad
\text{\eqref{eq:perm-inv} succeeds with probability $\leq \phantom{1 - \ } 4 \, \econst^{-\lambda^2/16}$;} \\
m \geq d - \mathrm{H}_d + \lambda \sqrt{\mathrm{H}_d}
&\quad\Longrightarrow\quad
\text{\eqref{eq:perm-inv} succeeds with probability $\geq 1 - 4\,\econst^{-\lambda^2/16}$.}
\end{align*}
The $d$th harmonic number $\mathrm{H}_d$ satisfies $\log d < \mathrm{H}_d < 1 + \log d$.
\end{proposition}

Proposition~\ref{prop:vec-list} yields the depressing assessment that we need a near-complete set of linear measurements to resolve our uncertainty about the ordering of the vector.  Nevertheless, we do not need \emph{all} of the measurements.  It would be interesting to understand how much the situation improves for vectors with many duplicated entries.

\begin{proof}
This result follows from Fact~\ref{fact:lin-inv-opt} and the kinematic bound~\eqref{eq:subspace-kinem} in Theorem~\ref{thm:approx-kinem} as soon as we compute the statistical dimension of the descent cone of the regularizer $\norm{ \cdot }_{\coll{P}(\vct{y}_0)}$ at the point $\vct{x}_0$.
By construction, the unit ball of $\norm{ \cdot }_{\coll{P}(\vct{y}_0)}$ coincides with the permutahedron $\coll{P}(\vct{y}_0)$, which equals $\coll{P}(\vct{x}_0)$ by permutation invariance.  Therefore,
$$
\Desc( \norm{ \cdot }_{\coll{P}(\vct{y}_0)}, \vct{x}_0 )
	= \Desc( \norm{ \cdot }_{\coll{P}(\vct{x}_0)}, \vct{x}_0 )
	= \NormC( \coll{P}(\vct{x}_0), \vct{x}_0 )^\polar.
$$
The second identity holds because a closed descent cone coincides with the polar of the normal cone of the sublevel set; see Appendix~\ref{sec:repr-desc-cones} for details.
Figure~\ref{fig:signed_permutahedron} illustrates the corresponding facts
about signed permutahedra.
To compute the statistical dimension, we apply the complementarity law~\eqref{eq:sdim-polarity-sum} to see that
$$
\delta\big( \Desc( \norm{ \cdot }_{\coll{P}(\vct{y}_0)}, \vct{x}_0 ) \big)
	= d - \delta\big(\NormC( \coll{P}(\vct{x}_0), \vct{x}_0 )\big)
	= d - \mathrm{H}_d,
$$
where the second relation follows from Proposition~\ref{prop:perm-cone}.  Apply the kinematic result~\eqref{eq:subspace-kinem} for subspaces, and invoke~\eqref{eq:pC-bd} to simplify the error bound $p_C(\lambda)$.
\end{proof}

\begin{remark}[Signed vectors]
An equally disappointing result holds for the problem of reconstructing a general vector from an unordered list of the \emph{magnitudes} of its entries, along with some linear measurements.  In this case, the appropriate regularizer is the Minkowski gauge of the signed permutahedron~\eqref{eq:sperm}.  We can use Proposition~\ref{prop:perm-cone} to compute the statistical dimension of the descent cone.  For a $d$-dimensional vector with distinct entries, we need about $d - \half \mathrm{H}_d$ random measurements to succeed reliably.
\end{remark}

\subsection{A numerical example}

We present a computer experiment that confirms our pessimistic analysis.  Fix the ambient dimension $d = 100$.  Set $\vct{x}_0 = (1, 2, \dots, 100)$ and $\vct{y}_0 = \vct{x}_0^\downarrow$.  For each $m = 85, 86, \dots, 100$, we repeat the following procedure 50 times:
\begin{enumerate}
\item	Draw a matrix $\mtx{A} \in \R^{m \times d}$ with independent standard normal entries, and form $\vct{z}_0 = \mtx{A} \vct{x}_0$.

\item	Use the \textsc{Matlab} package \texttt{CVX} to solve the linear inverse problem~\eqref{eq:perm-inv}.

\item	Declare success if the solution $\widehat{\vct{x}}$ satisfies $\enormsm{\widehat{\vct{x}} - \vct{x}_0} \leq 10^{-5}$.
\end{enumerate}
Figure~\ref{fig:vect-from-list} displays the outcome of this experiment.  As usual, the phase transition predicted at the statistical dimension $d - \mathrm{H}_d$ is very close to the empirical 50\% mark, which we obtain by performing a logistic regression of the data (see Appendix~\ref{sec:logistic}).

\begin{figure}[t!]
  \centering
  \includegraphics[width=\columnwidth]{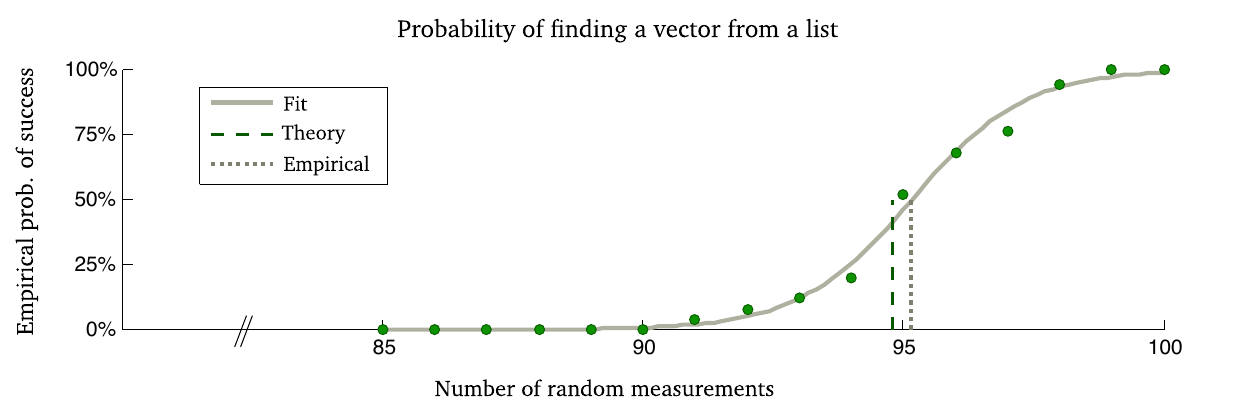}
  \caption{\textbf{Vectors from lists?}
  The empirical probability that the convex program~\eqref{eq:perm-inv} correctly identifies a vector $\vct{x}_0$ in $\R^{100}$ with distinct entries, provided an unordered list $\vct{y}_0$ of the entries of $\vct{x}_0$ and $m$ random linear measurements $\vct{z}_0 = \mtx{A}\vct{x}_0$.  The solid gray curve marks the logistic fit to the data.  The midpoint of the logistic curve $\mu = 95.17$ (finely dashed line), while the theory predicts a phase transition at the statistical dimension $\delta = 94.81$ (coarsely dashed line).  The error relative to the dimension $|\mu-\delta|/d = 0.36\%$.}
  \label{fig:vect-from-list}
\end{figure}

\section{Related work}
\label{sec:conclusion}

To conclude the body of the paper, we place our work in the context of the literature on geometric analysis of random convex optimization problems.  We trace four lines of thought on this subject.  The first draws from the theory of polytope angles; the second involves conic integral geometry; the third is based on comparison inequalities for Gaussian processes; and the last makes a connection with statistical decision theory.  Our results have some overlap with earlier work, but our discovery that the sequence of conic intrinsic volumes concentrates at the statistical dimension allows us to resolve several subtle but important questions that have remained open until now.

\subsection{Polytope-angle calculations}
\label{sec:polytope-angles}

The theory of polytope angles dates to the work of Schl{\"a}fli in the 1850s~\cite{schl:52}.  In pioneering research, Vershik \& Sporyshev~\cite{VS:86} used polytope-angle calculations to analyze random convex optimization problems.  They were able to estimate the average number of steps that the simplex algorithm requires to solve a linear program with random constraints as the number of decision variables tends to infinity.  This research inspired further theoretical work on the neighborliness of random polytopes~\cite{vesp:92,afsc:92,BH:99}.
More recently, Donoho~\cite{dono:06b} and Donoho \& Tanner~\cite{DT:05,dota:09a,DT:10,DT:10b} have used similar ideas to study specific regularized linear inverse problems with random data.  The papers~\cite{XH2011,KXAH2011} contain some additional work in this direction.  Let us offer a short, qualitative summary of this research.

Donoho~\cite{dono:06b} analyzed the performance of the convex program~\eqref{eqn:l1-min} for solving the compressed sensing problem described in Section~\ref{sec:cs}.  In the asymptotic regime where the number $s$ of nonzeros is proportional to the ambient dimension $d$, he obtained a lower bound $m \geq \psi(s)$ on the number $m$ of Gaussian measurements that are sufficient for the optimization to succeed (the \term{weak} threshold).  Numerical experiments~\cite{DonTan:09a} suggest that this bound is sharp, but the theoretical analysis in~\cite{dono:06b} falls short of establishing that a phase transition actually exists and identifying its location rigorously.  Finite-dimensional results with a similar flavor appear in~\cite{DT:10b}.

Donoho~\cite{dono:06b} also established an asymptotic lower bound on the number $m$ of random measurements sufficient to recover \emph{all} $s$-sparse vectors in $\R^d$ with high probability (the \term{strong} threshold).  Using different methods, Stojnic~\cite{stojnic10} has improved this bound for some values of the sparsity $s$.  These bounds are not subject to numerical interrogation, so we do not have reliable evidence about what actually happens.  Indeed, it remains an open question to prove that a strong phase transition exists and to identify its exact location in the regime where the sparsity is proportional to the ambient dimension.  

Donoho \& Tanner~\cite{dota:09a} have also made a careful study of the behavior of the convex program~\eqref{eqn:l1-min} in the asymptotic regime where the sparsity $s \ll d$.  In this case, they succeeded in proving that weak and strong thresholds exist, and they obtained exact formulas for the thresholds.  More precisely, at the computed thresholds, they show that the probability of success jumps from one to $1 - \eps$, where $\eps$ is positive.  Although these results do not ensure that certain failure awaits on the other side of the threshold curve, they do establish that the behavior changes.

Donoho \& Tanner~\cite{DT:05,dota:09a} provide similar results for the problem of recovering a sparse nonnegative vector by solving the $\ell_1$ minimization problem~\eqref{eqn:l1-min} with an additional nonnegativity constraint.  Once again, they obtain lower bounds on the number of Gaussian measurements required for weak and strong success.    These bounds are sharp in the ultrasparse regime $s = o(\log d)$.  The earlier work of Vershik \& Sporyshev~\cite{vesp:92} contains results closely related to the weak transition estimate from~\cite{dota:09a}.

Other authors have used polytope angle calculations to develop theory for related $\ell_1$ minimization problems.  For example, Khajehnejad et al.~\cite{KXAH2011} provide an analysis of the performance of weighted $\ell_1$ regularizers.  Xu \& Hassibi~\cite{XH2011} obtain lower bounds for the number of measurements required for \emph{stable} recovery of a sparse vector via $\ell_1$ minimization.  

Finally, let us mention that Donoho \& Tanner~\cite{DT:10} have obtained a more complete theory about the location of the phase transition for a regularized linear inverse problem where the $\ell_\infty$ norm is used as a regularizer.  Their results, formulated in terms of projections of the hypercube and orthant, are a consequence of a geometric theorem that goes back to Schl{\"a}fli~\cite{Sch:50}; see the discussion~\cite[p.~299]{scwe:08}.

\subsubsection{Faces of randomly projected polytopes}
\label{sec:face-plant}

It may be helpful if we elaborate on the relationship between the proof techniques in this paper and the approach described above.  For concreteness, we focus on the $\ell_1$ minimization method~\eqref{eqn:l1-min} for the compressed sensing problem,
but the ideas apply more broadly.

Suppose that we have acquired the vector $\vct{z}_0 = \mtx{A} \vct{x}_0$, where $\mtx{A}$ is a standard normal matrix and $\vct{x}_0$ is a vector with $s + 1$ nonzero entries.  We may assume that $\pnorm{1}{\vct{x}_0} = 1$.
Define the cross-polytope $P := \{ \vct{x} : \pnorm{1}{\vct{x}} \leq 1 \}$, and note that the sparse vector $\vct{x}_0$ lies in an $s$-dimensional face $F$ of the cross-polytope $P$.  Donoho shows that~\eqref{eqn:l1-min} succeeds if and only if $\mtx{A} F$ is an $s$-dimensional face of the projected polytope $\mtx{A} P$.  In other words, we must determine whether a face of the polytope ``survives'' a random projection.  Donoho~\cite{dono:06b} and Donoho \& Tanner~\cite{dota:09a} use the polytope-angle theory to address this question.

The approximate kinematic formula, Theorem~\ref{thm:kinematic}, provides a muscular approach to the same problem.  For a face $F$ of a polytope $P$, we define $\cone(F, P) := \cone( P - \vct{x} )$, where $\vct{x}$ is any point in the relative interior of the face~\cite[p.~297]{Gru:68}.
In the context of the $\ell_1$ minimization problem above, $\cone(F,P)$ coincides with the descent cone $\Desc(\pnorm{1}{\cdot}, \vct{x}_0)$.  The discussion in~\cite[Chap.~8.3]{scwe:08} yields the following observation.

\begin{fact} \label{fact:faces}
Let $F$ be an $s$-dimensional face of a polytope $P$ in $\R^d$, let $L$ be a $m$-dimensional subspace of $\R^d$, and let $\mtx{Q} \in \R^{d \times d}$ be a random rotation.  Then
$$
\Prob\big\{ \Proj_{\mtx{Q}L} F \text{ is an $s$-dimensional face of } \Proj_{\mtx{Q} L} P \big\}
	= \Prob\big\{ \mtx{Q}L^\perp \cap \cone(F, P) = \{ \vct{0} \} \big\}.
$$
\end{fact}

\noindent
In other words, the probability that a face $F$ of a polytope $P$ in $\R^d$ maintains its dimension under projection onto a random $m$-dimensional subspace is equal to the probability that a random $(d-m)$-dimensional subspace does not share a ray with $\cone(F,P)$.  When $m \leq \sdim(\cone(F,P))$ or, equivalently, $\sdim(\cone(F,P)) + (d - m) \geq d$, Theorem~\ref{thm:kinematic} shows that a random $(d-m)$-dimensional subspace is likely to share a ray with $\cone(F,P)$.  Therefore,
the face $F$ is unlikely to survive the projection onto an $m$-dimensional subspace.  Conversely, the face $F$ is likely to survive when $m \geq \sdim(\cone(F,P))$.
This analysis applies to any polytope, so our theory provides new results about a fundamental problem in combinatorial geometry.

\subsubsection{Commentary}

The analysis of structured inverse problems by means of polytope-angle computations has led to some striking conclusions, but this approach has inherent limitations.  First, the method is restricted to polyhedral cones, which means that it is silent about the behavior of many important regularizers, including the Schatten 1-norm.  Second, it requires detailed bounds on all angles of a given polytope (equivalently, all the intrinsic volumes of the normal cones of the polytope), which means that it is difficult to extend beyond a few highly symmetric examples.  For this reason, most of the existing results are asymptotic in nature.  Third, because of the intricacy of the calculations, this research has produced few definitive results of the form ``the probability of success jumps from zero to one at a specified location.''

We believe that our analysis supersedes most of the research on weak phase transitions for $\ell_1$ regularized linear inverse problems that is based on polytope angles.  We have shown for the first time that there is a transition from absolute success to absolute failure, and we have characterized the location of the threshold when the sparsity $s \geq \sqrt{d} + 1$.  On the other hand, our error estimate for the statistical dimension of the descent cone of the $\ell_1$ norm at a sparse vector is not very accurate when $s \leq \sqrt{d}$.  The independent work of Foygel \& Mackey~\cite[Prop.~1]{FoyMac:13} contains a bound that improves our result in this sparsity regime.

It is not hard to extend our analysis to the other settings discussed in this section.  Indeed, we can easily study regularized inverse problems involving weighted $\ell_1$ norms and $\ell_1$ norms with nonnegativity constraints.  We can effortlessly derive  phase transitions for $\ell_{\infty}$-regularized problems using Proposition~\ref{prop:linf-desc}.  Bounds for strong transitions are also accessible to our methods.  We have omitted all of this material for brevity.

\subsection{Conic intrinsic volumes}

Research on polytope angles has largely been supplanted
by spherical and conical integral geometry~\cite{Gl,scwe:08}.
Several authors have independently recognized the power of this
approach for analyzing random instances of convex optimization problems.

Amelunxen~\cite{am:thesis} and Amelunxen \& B{\"u}rgisser~\cite{ambu:11b,ambu:11c}
have shown that conic geometry offers an elegant way to perform
average-case and smoothed analysis of conic optimization problems.
Their work requires detailed computations of conic intrinsic volumes,
which can make it challenging to apply to particular cases.  We can
simplify some of their techniques using the new fact, Theorem~\ref{thm:main-conc},
that intrinsic volumes concentrate at the statistical dimension.
Theorem~\ref{thm:cone-prog} is based on their research.

McCoy \& Tropp~\cite{mctr:12} have used conic intrinsic volumes to study the
behavior of regularized linear inverse problems with random measurements
and regularized demixing problems under a random model.
This approach leads to both upper and lower bounds for weak and strong
phase transitions in a variety of problems.  As with Amelunxen's
work~\cite{am:thesis}, this research depends on detailed computations
of conic intrinsic volumes.  As a consequence, it was not possible to
rigorously locate the phase transition, nor was there any general
theory to inform us that phase transitions must exist in general.
By combining the ideas from~\cite{mctr:12} with Theorem~\ref{thm:approx-kinem},
the present work reaches stronger conclusions than~\cite{mctr:12}.

\subsection{Analysis based on Gaussian process theory}
\label{sec:statdim_sqgw}
The work we have discussed so far depends on various flavors
of integral geometry.  There is a completely different technique for
analyzing linear inverse problems with random data that depends on
a comparison inequality~\cite[Thm.~1.4]{gord:85} for Gaussian processes
due to Gordon.  Gordon~\cite{gord:88}
explains how to use this comparison to find accurate
bounds on the probability that a randomly oriented
subspace intersects a subset of the sphere. 

Rudelson \& Vershynin~\cite{RV:08} were the first authors
to observe that Gordon's work is relevant to the analysis
of the $\ell_1$ minimization method~\eqref{eqn:l1-min} for
the compressed sensing problem.
Stojnic~\cite{stojnic10} refined their method enough that
he was able to establish an empirically sharp condition
for $\ell_1$ minimization to succeed.  Oymak \& Hassibi~\cite{OH:10}
extended Stojnic's approach to study the low-rank recovery problem~\eqref{eqn:S1-min}
with random measurements.  Later, Chandrasekaran et al.~\cite{CRPW:12}
observed that the same approach can be used to analyze the
regularized linear inverse problem~\eqref{eq:lin-inv-gen} with random
measurements whenever the regularizer $f$ is convex.  
In particular, the result~\cite[Cor.~3.3(1)]{CRPW:12} is equivalent
with the success condition from Theorem~\ref{thm:phase-trans-lin-inv},
modulo the exact value of the constant $a_{\eta}$.

None of these authors has examined the failure condition
for random linear inverse problems, but we have
observed that it is possible to obtain such a result
by incorporating a polarity argument. We can also extend the Gaussian process approach to obtain a success condition
for demixing, but it does not yield a failure condition for these problems.
After this paper was written, Stojnic~\cite{Sto13:Regularly-Random}
developed some related results.

\subsubsection{The Gaussian width}

This line of work uses a summary parameter for convex cones
called the \term{Gaussian width}.  For a convex cone $C \subset \R^d$,
the width is defined as
\begin{equation*} \label{eq:gauss-width}
  w(C) := \Expect\bigl[\sup\nolimits_{\smash{\vct{y}} \in C \cap \sphere{d-1}} \
\ip{\smash{\vct{y}}}{\smash{\vct{g}}}\bigr].
\end{equation*}
The Gaussian width $w(C)$ is proportional to the classical mean width~\cite[p.~42]{Sch:93} of the set $C\cap \sphere{d-1}$.  The width $w(C)$ also has a numerical relationship with the statistical dimension $\delta(C)$, which is not surprising in view of the mean-squared-width formulation~\eqref{eq:sdim-sup-expect}.

\begin{proposition}[Statistical dimension \& Gaussian width]
\label{prop:sdim-width}
Let $C$ be a convex cone.  Then
\begin{equation} \label{eq:width-sdim}
w^2(C) \leq \delta(C) \leq w^2(C) + 1.
\end{equation}
\end{proposition}

\noindent
The lower bound in~\eqref{eq:width-sdim} requires little more than Jensen's inequality. The upper bound depends on some concentration
arguments.  See Appendix~\ref{app:sdim-width} for a short proof.

\subsubsection{Commentary}

In light of Proposition~\ref{prop:sdim-width}, the statistical dimension and the Gaussian width are essentially equivalent as summary parameters for cones.  This observation allows us to bridge the chasm between two perspectives on random convex optimization problems: the approach based on integral geometry and the approach based on Gaussian process theory.  Even in the simple case of $\ell_1$ minimization~\eqref{eqn:l1-min}, this connection is entirely new; cf.~\cite[pp.~312--313]{XuHas:12}.

This link yields many insights.  For example, all of the statistical dimension calculations here lead to analogous estimates for the Gaussian width.  In particular, Theorem~\ref{thm:sharp-descent} provides the first accurate lower bounds for the Gaussian width of a descent cone.  Furthermore, we can use the intrinsic characterization of the statistical dimension, Definition~\ref{def:sdim-int}, to study cones that are not accessible to the metric characterization, Proposition~\ref{def:sdim}.  For instance, consider the results for permutahedra in Proposition~\ref{prop:perm-cone}.

Despite these connections, we did not arrive at the statistical dimension by squaring the Gaussian width.  Rather, Definition~\ref{def:sdim-int} emerges from Theorem~\ref{thm:main-conc}, our result that the conic intrinsic volumes concentrate.  To travel from this definition to the Gaussian width, we must take several long steps: Proposition~\ref{prop:sdim-int-vols} leads to the metric characterization of the statistical dimension, Proposition~\ref{def:sdim}; we apply Proposition~\ref{prop:prop-sdim} to obtain the mean-squared-width formulation~\eqref{eq:sdim-sup-expect} of statistical dimension; and Proposition~\ref{prop:sdim-width} sandwiches the mean-squared-width formulation by the Gaussian width.

\subsection{Minimax denoising}
\label{sec:minimax}

Several authors~\cite{DonMalMon:09,DonJohMon:11,DonGavMon:13} have remarked on the power of statistical decision theory to empirically predict the location of the phase transition in a regularized linear inverse problem with random data.  For the compressed sensing problem, two recent papers~\cite{BayMon:12,BayLelMon:12} provide a rigorous explanation for this coincidence.  But there is no general theory that illuminates the connection between these two settings.  Our work and a recent paper of Oymak \& Hassibi~\cite{OymHas:12} together resolve this issue.  In short, Oymak \& Hassibi show that the minimax risk for denoising is essentially the same as the statistical dimension, while our research proves that a phase transition must occur at the statistical dimension.  Let us elaborate.

A classical problem in statistics is to estimate a target vector $\vct{x}_0$ given an observation of the form $\vct{z}_0 = \vct{x}_0 + \sigma \vct{g}$ where $\vct{g}$ is a standard normal vector and $\sigma$ is an unknown variance parameter.  When the unknown vector $\vct{x}_0$ has specified properties (e.g., sparsity), we can often construct a convex regularizer $f$ that promotes this type of structure~\cite{CRPW:12}.  A natural estimation procedure is to solve the convex optimization problem
\begin{equation} \label{eq:DN}
\widehat{\vct{x}}_{\gamma} :=
	\argmin_{\vct{x} \in \R^d} \ \gamma f(\vct{x}) + \tfrac{1}{2} \enormsm{\vct{z}_0 - \vct{x}}^2.
\end{equation}
The regularization parameter $\gamma > 0$ negotiates a tradeoff between the structural penalty and the data fidelity term.  One way to assess the performance of the estimator~\eqref{eq:DN} is the \term{minimax MSE risk},\footnote{The usual definition of the minimax risk involves an additional supremum over  a class of distributions on the target $\vct{x}_0$. In many applications, the symmetries in the regularizer $f$ allow a straightforward reduction to the case of a fixed target $\vct{x}_0$.  See~\cite[Sec.~6.4]{OymHas:12}.}
defined as
$$
\Rmm(\vct{x}_0) := \sup_{\sigma > 0} \ \inf_{\gamma > 0} \
	\frac{1}{\sigma^2} \Expect \big[ \enormsm{ \widehat{\vct{x}}_{\gamma} - \vct{x}_0 }^2 \big].
$$
In other words, the risk identifies the relative mean-square error for the best choice of tuning parameter $\gamma$ at the worst choice of the noise variance $\sigma^2$.

The papers~\cite{DonMalMon:09,DonJohMon:11,DonGavMon:13} examine several regularizers $f$ where the minimax risk empirically predicts the performance of the linear inverse problem~\eqref{eq:lin-inv-gen} with a Gaussian measurement matrix $\mtx{A}$.  The authors of this research propound a conjecture that may be expressed as follows.

\begin{conjecture}[Minimax risk predicts phase transitions] \label{conj:rmm}
Suppose that $\mtx{A} \in \R^{m \times d}$ is a matrix with independent standard normal entries, and let $f : \R^d \to \overline{\R}$ be a proper convex function.  Then
\begin{align*}
m \geq \Rmm(\vct{x}_0) + o(d)
&\quad\Longrightarrow\quad
\text{\eqref{eq:lin-inv-gen} succeeds with probability $1 - o(1)$;} \\
m \leq \Rmm(\vct{x}_0) + o(d)
&\quad\Longrightarrow\quad
\text{\eqref{eq:lin-inv-gen} succeeds with probability $o(1)$.}
\end{align*}
The order notation here should be interpreted heuristically.
\end{conjecture}

\noindent
To the best of our knowledge, this claim has been established rigorously only for the $\ell_1$ norm~\cite[Thm.~8]{BayLelMon:12} in the asymptotic setting.  The paper~\cite{BayLelMon:12} also includes analysis for a wider class of matrices.

Together, our paper and the recent paper~\cite{OymHas:12} settle Conjecture~\ref{conj:rmm} in the nonasymptotic setting for many regularizers of interest.  Indeed, Oymak \& Hassibi~\cite{OymHas:12} prove that
$$
\abs{ \Rmm(\vct{x}_0) - \delta\big( \Desc(f, \vct{x}_0) \big) } = O(\sqrt{d}).
$$
Their result holds under mild conditions on the regularizer $f$ that suffice to address most of the phase transitions conjectured in the literature.  Our result, Theorem~\ref{thm:phase-trans-lin-inv}, demonstrates that the phase transition in the linear inverse problem~\eqref{eq:lin-inv-gen} with a standard normal matrix $\mtx{A} \in \R^{m\times d}$ occurs when
$$
\abs{ m - \delta\big(\Desc(f, \vct{x}_0)\big) } = O(\sqrt{d}).
$$
Combining these two results, we conclude that, in some generality, the minimax risk coincides with the location of the phase transition in a regularized linear inverse problem with random measurements.

\appendix

\section{Computer experiments}
\label{app:experiments}

We confirm the predictions of our theoretical analysis by
performing computer experiments.  This appendix contains
some of the details of our numerical work.  All experiments
were performed using the \texttt{CVX} package~\cite{GB2013}
for \textsc{Matlab} with the default settings in place. 

\subsection{Linear inverse problems with random measurements}

This section describes the two experiments from Section~\ref{sec:line-inverse-probl}
that illustrate the empirical phase transition in compressed sensing via $\ell_1$ minimization and in low-rank matrix recovery via Schatten 1-norm minimization.

In the compressed sensing example, we fix the ambient dimension $d = 100$.  For each $m = 1, 2, 3, \dots, 100$ and each $s = 1, 2, 3, \dots, 100$, we repeat
the following procedure $50$ times:
\begin{enumerate}
\item	Construct a vector $\vct{x}_0 \in \R^d$ with $s$ nonzero entries.  The locations
of the nonzero entries are selected at random; each nonzero equals $\pm 1$ with equal probability.

\item	Draw a standard normal matrix $\mtx{A} \in \R^{m \times d}$, and form $\vct{z}_0 = \mtx{A} \vct{x}_0$.

\item	Solve~\eqref{eq:l1-min-v2} to obtain an optimal point $\widehat{\vct{x}}$.

\item	Declare success if $\enorm{\widehat{\vct{x}} - \vct{x}_0} \leq 10^{-5}$.
\end{enumerate}
All random variables are drawn independently in each step and at each iteration.  Figures~\ref{fig:l1-min}[left] and~\ref{fig:lin-inv-1}[left] show the empirical probability of success for this procedure.  We performed a similar experiment for ambient dimension $d = 600$.  In this case, we iterate over $s \in \{1, 7, 13, \dots, 595\}$ and $m \in \{0, 6, 12, \dots, 600\}$; Figure~\ref{fig:l1-min}[right] displays a subset of this data.

In the compressed sensing experiment, 50 trials suffice to estimate the probability because of the concentration phenomenon established in Theorem~\ref{thm:phase-trans-lin-inv}.  Furthermore, the experiment does not appear particularly sensitive to the success tolerance \(10^{-5}\) above. Limited tests confirm that tolerances with orders of magnitude \(10^{-3}\) through \(10^{-7}\) give essentially the same results.  The value \(10^{-5}\) was chosen as a stringent condition that would be insensitive to numerical errors produced by the \texttt{CVX} software package.
Nevertheless, we encountered some numerical problems in the experiment with $d = 600$, which led to a few spurious failures in Figure~\ref{fig:l1-min}[right].

We take a similar approach in the low-rank matrix recovery problem.  Fix $n = 30$, and consider square $n \times n$ matrices.  For each rank $r = 1, 2, \dots, n$ and each $m = 1, 29, 58, 87, \dots, n^2$, we repeat the following procedure $50$ times:
\begin{enumerate}
\item	If $r \geq \lceil \sqrt{m} \rceil + 1$, declare failure because the number of degrees of freedom in an $n\times n$ rank-$r$ matrix exceeds the number $m$ of measurements.

\item	Draw a rank-$r$ matrix $\mtx{X}_0 = \mtx{Q}_1 \mtx{Q}_2^\transp$, where $\mtx{Q}_1$ and $\mtx{Q}_2$ are independent $n \times r$ matrices with orthonormal columns, drawn uniformly from an appropriate Stiefel manifold~\cite{Mez2007}.

\item	Draw a standard normal matrix $\mtx{A} \in \R^{m \times n^2}$, and define $\coll{A}(\mtx{X}) := \mtx{A} \cdot \mathrm{\vec}(\mtx{X})$, where the vectorization operator stacks the columns of a matrix.  Form the vector of measurements $\vct{z}_0 = \coll{A}(\mtx{X}_0)$.

\item	Solve~\eqref{eqn:S1-min} to obtain an optimal point $\widehat{\mtx{X}}$.

\item	Declare success if $\fnorm{\smash{\widehat{\mtx{X}}- \mtx{X}_0}} \leq 10^{-5}$.
\end{enumerate}
As before, all random variables are chosen independently.
Readers interested in reproducing this experiment should be aware that this procedure required nearly one month to execute on a desktop workstation.
Figure~\ref{fig:lin-inv-1}[right] displays the results of this experiment.  Once again, the probability of success and failure is relatively insensitive to the precise tolerance \(10^{-5}\) used above.

\subsection{Statistical dimension curves}
\label{sec:addend-comp-stat}

The formulas~\eqref{eq:l1-sdim} and~\eqref{eq:S1-sdim} for the statistical dimension of the descent cones of the $\ell_1$ norm and the Schatten 1-norm do not have a closed form representation.  Nevertheless, we can evaluate these expressions using simple numerical methods.  Indeed, in each case, we solve the stationary equation~\eqref{eq:l1-stationary} and~\eqref{eq:S1-stationary} using the rootfinding procedure \texttt{fzero}, which works well because the left-hand side of each equation is a monotone function of $\tau$.  To evaluate the integral in~\eqref{eq:l1-sdim}, we use the command \texttt{erfc}.  To evaluate the integral in~\eqref{eq:S1-sdim}, we use the quadrature function~\texttt{quadgk}.

We have encountered some numerical stability problems evaluating~\eqref{eq:l1-sdim} when the proportional sparsity $\rho = s/d$ is close to zero or one.  Similarly, there are sometimes difficulties with~\eqref{eq:S1-sdim} when the proportional rank $\rho = r/m$ or the aspect ratio $\nu = m/n$ are close to zero or one.  Nevertheless, relatively simple code based on this approach is usually reliable.  Software is available online~\cite{snow-code}.

\subsection{Logistic regression}
\label{sec:logistic}

Several of the experiments involve fitting the logistic function
\begin{equation*}
  \ell(x) \defeq \frac{1}{1+\econst^{-(\beta_0 + \beta_1 x)}}
\end{equation*}
to the data, where $\beta_0,\beta_1 \in \R$ are parameters.
We use the command \texttt{glmfit} to accomplish this task.
The center $\mu\defeq -\beta_0/\beta_1$ of the logistic function is
the point such that $\ell(\mu)=\frac{1}{2}$.

\section{Background on conic geometry}
\label{sec:notation-conventions}

Sections~\ref{sec:repr-desc-cones}--\ref{sec:eucl-proj-onto-1} below provide some important facts from convex geometry that we will use liberally.  Section~\ref{sec:proof-prop-refpr} establishes the properties of the statistical dimension listed in Proposition~\ref{prop:prop-sdim}.

\subsection{Representation of descent cones}
\label{sec:repr-desc-cones}

Let $f : \R^d \to \overline{\R}$ be a
proper convex function.  Recall that the \term{descent cone} of $f$
at a point $\vct{x}$ is given by
$$
\Desc(f,\vct{x}) \defeq \bigcup_{\tau > 0}  \big\{\vct{y} \in \R^d :
  f(\vct{x} + \tau \vct{y} ) \leq f(\vct{x}) \big\}.
$$
The \term{normal cone} of a convex set $E$ at a point $\vct{x}$ is defined as
\begin{equation*}  \NormC( E, \vct{x} ) := \big\{ \vct{s} \in \R^d : \ip{ \vct{s} }{ \smash{\vct{y}} - \vct{x} } \leq 0
 \quad\text{for all $\vct{y} \in E$} \big\}.
 \end{equation*}
The polar of a descent cone has some attractive duality properties.
First, the polar of a descent cone coincides with the normal cone of a sublevel set:
\begin{equation} \label{eq:descent-normal}
\Desc(f, \vct{x})^\polar = \NormC(E, \vct{x})
\quad\text{where}\quad
E = \big\{ \vct{y} \in \R^d : f(\vct{y}) \leq f(\vct{x}) \big \}.
\end{equation}
If the descent cone is closed, we can apply the bipolar theorem~\cite[Thm.~14.1]{Rock} to transfer the polar in~\eqref{eq:descent-normal} from the normal cone to the descent cone.
Second, there is a duality between descent cones and subdifferentials.  Recall that
$$
\partial f(\vct{x}) :=
	\big\{ \vct{s} \in \R^d : f(\vct{y}) \geq f(\vct{x}) + \ip{\vct{s}}{\smash{\vct{y}} - \vct{x}} \text{ for all $\vct{y} \in \R^d$} \big\}.
$$
Assuming that the subdifferential $\partial f(\vct{x})$ is nonempty, compact, and does not contain the origin, the result~\cite[Cor.~23.7.1]{Rock} provides that
\begin{equation} \label{eqn:desc-polar}
\Desc(f, \vct{x})^\polar = \mathrm{\cone}(\partial f(\vct{x}))
	:= \bigcup_{\tau \geq 0} \tau \cdot \partial f(\vct{x}). 
\end{equation}
The expression $\tau \cdot \partial f(\vct{x})$ represents dilation of the subdifferential by a factor $\tau$.  The relation~\eqref{eqn:desc-polar} offers a powerful tool for computing the statistical dimension of a descent cone, as described in Proposition~\ref{prop:sdim-descent}.  Related identities hold under weaker technical conditions~\cite[Thm.~23.7]{Rock}; these results can be used to establish the formula~\eqref{eq:sdim-descent} without the compactness assumption.

\subsection{Euclidean projections onto sets} 
\label{sec:eucl-proj-onto}
Let $E \subset \R^d$ be a closed convex set. 
 The \term{Euclidean projection} onto the set $E$ is the map
 $$
 \vct{\pi}_{E} : \R^d \to E
 \quad\text{where}\quad
 \vct{\pi}_{E}(\vct{x}) :=
 \argmin\big\{ \enormsm{ \vct{x} - \vct{y} } : \vct{y} \in E \big\}.
 $$
 The projection takes a well-defined value because the norm is strictly convex.  
Let us note some properties of the distance and projection maps.
First, the function $\dist(\cdot, E)$ is convex~\cite[p.~34]{Rock}.
Next, the maps $\vct{\pi}_E$ and $\Id - \vct{\pi}_E$ are nonexpansive with respect
to the Euclidean norm~\cite[Thm.~31.5 et seq.]{Rock}:
\begin{equation} \label{eq:I-pi}
\enormsm{\vct{\pi}_E(\vct{x}) - \vct{\pi}_E(\vct{y})}
	\leq \enormsm{\vct{x} - \vct{y} }
	\quad\text{and}\quad
\enormsm{ (\Id - \vct{\pi}_E)(\vct{x}) - (\Id -\vct{\pi}_E)(\vct{y}) }
	\leq \enormsm{ \vct{x} - \vct{y} }
	\quad\text{for all $\vct{x}, \vct{y} \in \R^d$.}
\end{equation}
As a consequence, the projection $\vct{\pi}_E$ is continuous,
and the distance function is 1-Lipschitz with respect to the
Euclidean norm:
\begin{equation} \label{eq:dist-lip}
\abs{ \dist(\vct{x}, E) - \dist(\vct{y}, E) }
	\leq \enormsm{ \vct{x} - \vct{y} }
	\quad\text{for all $\vct{x}, \vct{y} \in \R^d$.}
\end{equation}
The squared distance is differentiable everywhere, and the derivative satisfies
\begin{equation} \label{eq:grad-dist}
\nabla \dist^2( \vct{x}, E ) = 2 \, (\vct{x} - \vct{\pi}_{E}(\vct{x}))
\quad\text{for all $\vct{x} \in \R^d$.}
\end{equation}
This point follows from~\cite[Thm.~2.26]{RW:98}.

\subsection{Euclidean projections onto cones} 
\label{sec:eucl-proj-onto-1}

Let $C \subset \R^d$ be a closed convex cone.  Recall that the Euclidean projection onto the cone $C$ is the map
$$
\Proj_C : \R^d \to C
\quad\text{where}\quad
\Proj_C(\vct{x}) := \argmin \big\{ \enormsm{ \vct{x} - \vct{y} } : \vct{y} \in C \big\}.
$$
We have used a separate notation for the projection onto a set because the projection onto a cone enjoys a number of additional properties~\cite[Sec.~III.3.2]{HUL:93a}.
First, the projection onto a cone is nonnegative homogeneous:
$\Proj_C( \tau \, \vct{x}) = \tau \, \Proj_C(\vct{x})$ for all $\tau \geq 0$.
Next, the cone $C$ induces an orthogonal decomposition of $\R^d$.
\begin{equation} \label{eq:orthog-decomp}
\vct{x} = \Proj_C(\vct{x}) + \Proj_{C^\polar}(\vct{x})
\quad\text{and}\quad
\ip{ \Proj_C(\vct{x}) }{ \Proj_{C^\polar}(\vct{x}) } = 0
\quad\text{for all $\vct{x} \in \R^d$.}
\end{equation}
The decomposition~\eqref{eq:orthog-decomp} yields the Pythagorean identity
\begin{equation} \label{eq:pythag}
\normsq{\vct{x}}
	= \normsq{ \Proj_C(\vct{x}) } + \normsq{ \Proj_{C^\polar}(\vct{x}) }
	\quad\text{for all $\vct{x} \in \R^d$.}	
\end{equation}
It also implies the distance formulas
\begin{equation} \label{eq:dist-proj}
\dist(\vct{x}, C) 	= \enormsm{ \vct{x} - \Proj_C(\vct{x}) }
	= \enormsm{ \Proj_{C^\polar}(\vct{x}) }
	\quad\text{for all $\vct{x} \in \R^d$}.
\end{equation}
The squared norm of the projection has a nice regularity
property, which follows from a short argument based
on~\eqref{eq:grad-dist},~\eqref{eq:orthog-decomp}, and~\eqref{eq:dist-proj}:
\begin{equation} \label{eq:grad-proj}
\nabla \enormsm{ \Proj_C(\vct{x}) }^2 = 2\,\Proj_C(\vct{x})
\quad\text{for all $\vct{x} \in \R^d$.}
\end{equation}
Indeed,
$
\nabla \enormsm{\Proj_C(\vct{x})}^2 	= \nabla \dist^2(\vct{x}, C^\polar) = 2 (\vct{x} - \Proj_{C^\polar}(\vct{x})) = 2 \,\Proj_C(\vct{x}).
$
Finally, the projection map decomposes under Cartesian products.
For two closed convex cones 
$C_1 \subset \R^{d_1}$ and $C_2\subset \R^{d_2}$ the product $C_1 \times C_2 \subset \R^{d_1 + d_2}$ satisfies
\begin{equation} \label{eq:proj-split}
\Proj_{C_1 \times C_2}\big((\vct{x}_1, \vct{x}_2)\big)
	= \big(\Proj_{C_1}(\vct{x}_1), \Proj_{C_2}(\vct{x}_2) \big)
	\quad\text{for all $\vct{x}_1 \in \R^{d_1}$ and $\vct{x}_2 \in \R^{d_2}$.} 
\end{equation}
The relation~\eqref{eq:proj-split} is easy to check directly.

\subsection{Proof of Proposition~\ref{prop:prop-sdim}}
\label{sec:proof-prop-refpr}

Let us verify the elementary properties of the statistical dimension delineated in Proposition~\ref{prop:prop-sdim}.  These results follow quickly from the orthogonal decomposition~\eqref{eq:pythag} and basic facts about a standard normal random vector.

\begin{enumerate} \setlength{\itemsep}{1mm}
\item	The intrinsic formulation~\eqref{eq:sdim-intr} simply restates Definition~\ref{def:sdim-int}.

\item	The Gaussian formulation~\eqref{eq:sdim-gauss} follows from Proposition~\ref{def:sdim} or the equivalent Proposition~\ref{prop:sdim-int-vols}.

\item	To derive the spherical formulation~\eqref{eq:sdim-circ-expect} from~\eqref{eq:sdim-gauss}, we introduce the spherical decomposition $\vct{g} = R \, \vct{\theta}$, where $R := \enormsm{\vct{g}}$ is a chi random variable with
$d$ degrees of freedom that is independent from the spherical variable $\vct{\theta}$.  By nonnegative homogeneity and independence,
\begin{equation*}
  \sdim(C) =\Expect\left[R^2 \enormsm{ \Proj_C(\vct{\theta}) }^2 \right] =\Expect\left[R^2\right] \cdot \Expect\left[\enormsm{ \Proj_C(\vct{\theta}) }^2 \right]  = d \,\Expect\left[\enormsm{\Proj_C(\vct{\theta})}^2\right],
\end{equation*}
where the last equality follows because $\Expect\left[\enormsm{\vct{g}}^2\right] = d$.

\item
The polar identity~\eqref{eq:delta-dist} is a direct consequence of the distance
formula~\eqref{eq:dist-proj}, which implies that $\enormsm{ \Proj_C(\vct{g}) } = \dist( \vct{g}, C^\polar )$.

\item  The supremum formulation~\eqref{eq:sdim-sup-expect} also follows from~\eqref{eq:sdim-gauss}.  Using the Pythagorean decomposition~\eqref{eq:pythag}, we observe that
$$
\sup\nolimits_{\vct{y} \in C \cap \ball{d}} \ \ip{ \smash{\vct{y}} }{ \smash{\vct{g}} }
	= \sup\nolimits_{\vct{y} \in C \cap \ball{d}} \ \ip{ \smash{\vct{y}} }{ \Proj_C(\smash{\vct{g}}) + \Proj_{C^\polar}(\smash{\vct{g}}) }
	\leq \sup\nolimits_{\vct{y} \in C \cap \ball{d}} \ \ip{ \smash{\vct{y}} }{ \Proj_C(\smash{\vct{g}}) }.
$$
The second inequality is immediate from the definition~\eqref{eq:polar-cone} of the polar cone.  Choose $\vct{y}$ to be the unit vector in the direction $\Proj_C(\vct{g})$ to see that
$$
\sup\nolimits_{\vct{y} \in C \cap \ball{d}} \ \ip{ \smash{\vct{y}} }{ \smash{\vct{g}} }
	= \enormsm{ \Proj_C(\vct{g}) }.
$$
Square the expression, and take the expectation to complete the argument.

\item
The rotational invariance property~\eqref{eq:sdim-rot-inv} follows immediately from the fact that a standard normal vector is rotationally invariant.

\item
To compute the statistical dimension of a subspace, note that the Euclidean projection of a standard normal vector onto a subspace has the standard normal distribution supported on that subspace, so its expected squared norm equals the dimension of the subspace.

\item
The complementarity law~\eqref{eq:sdim-polarity-sum}
follows from the Pythagorean identity~\eqref{eq:pythag}.

\item
We obtain the direct product rule~\eqref{eq:sdim-direct-product} from the
observation~\eqref{eq:proj-split} that projection splits over a direct
product, coupled with the fact that projecting a standard normal
vector onto each of two orthogonal subspaces results in two independent
standard normal vectors.

\item
Finally, we verify the monotonicity law.
Polarity reverses inclusion, so $K^\polar \subset C^\polar$.  Using
the polarity identity~\eqref{eq:delta-dist} twice, we obtain
$$
\delta(C) = \Expect \big[ \dist^2( \vct{g}, C^\polar ) \big]
	\leq \Expect \big[ \dist^2( \vct{g}, K^\polar ) \big]
	= \delta(K).
$$
\end{enumerate}

\noindent
This completes the recitation. 
\section{Theoretical results on descent cones}
\label{app:descent-cones}

This appendix contains the theoretical analysis
that permits us to calculate the statistical
dimension of a descent cone.  In particular, we prove Proposition~\ref{prop:sdim-descent}
and establish Theorem~\ref{thm:sharp-descent}.

\subsection{The expected distance to the subdifferential}
\label{sec:dist-conic}

In this section, we complete the proof of Proposition~\ref{prop:sdim-descent}.  This result forms the basis for Recipe~\ref{fig:desc-recipe}, which lets us compute the statistical dimension of a descent cone.  
We must show that the function
$\distsubdiff : \tau \mapsto \Expect\big[ \dist^2 \big(\vct{g}, \tau \cdot \partial f(\vct{x}) \big) \big]$
exhibits a number of analytic and geometric properties.
The hypotheses of the proposition ensure that
$\partial f(\vct{x})$ is a nonempty,
compact, convex set that does not contain the origin.
For clarity, we establish an abstract result that only depends
on the distinguished properties of the subdifferential.
Let us begin with a lemma about a related, but simpler, function.

\begin{lemma}[Distance to a dilated set] \label{slem:dist-slice}
Let $S$ be a nonempty, compact, convex subset of $\R^d$ that does
not contain the origin.  In particular, there are numbers that satisfy
$b \leq \enorm{\vct{s}} \leq B$
for all $\vct{s} \in S$.
Fix a point $\vct{u} \in \R^d$, and define the function
\begin{equation} \label{eq:Fx-def}
\distsubdiff_{\vct{u}} : \tau \mapsto \dist^2( \vct{u}, \tau S )
\quad\text{for $\tau \geq 0$.}
\end{equation}
The following properties hold.

\begin{enumerate} \setlength{\itemsep}{1mm}
\item	The function $\distsubdiff_{\vct{u}}$ is convex.

\item	The function $\distsubdiff_{\vct{u}}$ satisfies the lower bound
\begin{equation} \label{eq:Fx-unbdd}
\distsubdiff_{\vct{u}}(\tau) \geq (\tau b - \enorm{\vct{u}})^2
\quad\text{for all $\tau \geq \enorm{\vct{u}} / b$.}
\end{equation}
In particular, $\distsubdiff_{\vct{u}}$ attains its minimum value in the compact interval $[0, 2\norm{\vct{u}}/b]$.

\item	The function $\distsubdiff_{\vct{u}}$ is continuously differentiable,
and the derivative takes the form
\begin{equation} \label{eq:Fx-derivative}
\distsubdiff'_{\vct{u}}(\tau) =
- \frac{2}{\tau} \ip{ \vct{u} - \vct{\pi}_{\tau S}(\vct{u}) }{ \vct{\pi}_{\tau S}(\vct{u}) }
\quad\text{for $\tau > 0$.}
\end{equation}
The right derivative $\distsubdiff'_{\vct{u}}(0)$ exists, and $\distsubdiff'_{\vct{u}}(0) = \lim_{\tau \downarrow 0} \distsubdiff'_{\vct{u}}(\tau)$.

\item	The derivative admits the bound
\begin{equation} \label{eq:Fx'-bdd}
\abs{ \distsubdiff'_{\vct{u}}(\tau) } \leq 2B \, \big(\enorm{\vct{u}} + \tau B \big)
\quad\text{for all $\tau \geq 0$.}
\end{equation}

\item	Furthermore, the map $\vct{u} \mapsto \distsubdiff_{\vct{u}}'(\tau)$ is Lipschitz for each $\tau \geq 0$:
\begin{equation} \label{eq:Fx'-lip}
\abs{  \distsubdiff_{\vct{u}}'(\tau) - \smash{\distsubdiff_{\vct{y}}'(\tau)} }
	\leq 2 B \enormsm{ \vct{u} - \vct{y} }
	\quad\text{for all $\vct{u}, \vct{y} \in \R^d$.}
\end{equation}
\end{enumerate}
\end{lemma}

\begin{proof} These claims all take some work.  Along the way, we also need
to establish some auxiliary results to justify the main points.

\textit{Convexity.}  For $\tau > 0$, convexity follows from the representation
\begin{equation} \label{eq:Fx-convex}
\distsubdiff_{\vct{u}}(\tau)
	= \bigg( \inf_{\vct{s} \in S} \ \enormsm{ \vct{u} - \tau \vct{s} } \bigg)^2
	= \bigg( \tau \cdot \inf_{\vct{s} \in S} \ \enormsm{ \vct{u}/\tau - \vct{s} } \bigg)^2
	= \big( \tau \cdot \dist(\vct{u}/\tau, S) \big)^2.
\end{equation}
By way of justification, the distance to a closed convex set is a convex function~\cite[p.~34]{Rock}, the perspective transformation~\cite[Sec.~IV.2.2]{HUL:93a} of a convex function is convex, and the square of a nonnegative convex function is convex by a direct calculation.

\textit{Continuity.}  The representation~\eqref{eq:Fx-convex} shows that the
function $\distsubdiff_{\vct{u}}$ is continuous for $\tau > 0$ because the distance
to a convex set is a Lipschitz function, as stated in~\eqref{eq:dist-lip}.
To obtain continuity at $\tau = 0$, simply note that
$$
\abs{\distsubdiff_{\vct{u}}(\eps) - \distsubdiff_{\vct{u}}(0)}
	= \abs{ \enormsm{ \vct{u} - \vct{\pi}_{\eps S}(\vct{u}) }^2 - \enormsm{\vct{u}}^2 }
	\leq 2 \absip{ \vct{u} }{ \vct{\pi}_{\eps S}(\vct{u}) }
	+ \enormsm{ \vct{\pi}_{\eps S}(\vct{u}) }^2
	\leq 2 \enorm{ \vct{u} } (\eps B)  + (\eps B)^2
	\to 0
	\quad\text{as $\eps \to 0$.}
$$
Indeed, each point in $\eps S$ is bounded in norm by $\eps B$, so the projection
$\vct{\pi}_{\eps S}(\vct{u})$ admits the same bound.  Continuity implies that $\distsubdiff_{\vct{u}}$ is convex on the entire domain $\tau \geq 0$.

\textit{Attainment of minimum.}  Assume that $\tau \geq \enorm{\vct{u}} / b$.  Then
$$
\dist(\vct{u}, \tau S) = \inf_{\vct{s} \in S} \ \enorm{ \tau \vct{s} - \vct{u} }
	\geq \inf_{\vct{s} \in S} \ \big( \tau \enorm{\vct{s}} - \enorm{\vct{u}} \big)
	\geq \tau b - \enorm{ \vct{u} }
	\geq 0.
$$
Square this relation to reach~\eqref{eq:Fx-unbdd}.  As a consequence, for all $\tau > 2 \enorm{\vct{u}}/b$, we have $\distsubdiff_{\vct{u}}(\tau) > \distsubdiff_{\vct{u}}(0) = \enormsq{\vct{u}}$.  If follows that any minimizer of $\distsubdiff_{\vct{u}}$ must occur in the compact interval $[0, 2\enorm{\vct{u}}/b]$.  Since $\distsubdiff_{\vct{u}}$ is continuous, it attains its minimal value in this set.

\textit{Differentiability.}  We obtain the derivative from a direct calculation:
\begin{multline*} \label{eq:Fx-der1}
\distsubdiff'_{\vct{u}}(\tau) = \frac{\diff{}}{\diff{\tau}} \big( \tau^2 \dist^2(\vct{u}/\tau, S) \big)
	= 2\tau \cdot \dist^2(\vct{u}/\tau, S) +
	\tau^2 \ip{ 2\, \big( (\vct{u}/\tau) - \vct{\pi}_{S}(\vct{u}/\tau) \big)}{-\vct{u}/\tau^2} \\
	= \frac{2}{\tau} \big( \dist^2(\vct{u}, \tau S) -
	\ip{ \vct{u} - \vct{\pi}_{\tau S}(\vct{u}) }{ \vct{u} } \big)
	= - \frac{2}{\tau} \ip{\vct{u} - \vct{\pi}_{\tau S}(\vct{u})}
	{\vct{\pi}_{\tau S}(\vct{u})}
\end{multline*}
The first relation follows from~\eqref{eq:Fx-convex}.  The second relies
on the formula~\eqref{eq:grad-dist} for the derivative of the squared
distance.  To obtain the last relation, we express the squared distance
as $\enormsm{\vct{u} - \vct{\pi}_{\tau S}(\vct{u})}^2$.

\textit{Right derivative at zero.}
The right derivative $\distsubdiff'_{\vct{u}}(0)$ exists, and the limit formula holds because
$\distsubdiff_{\vct{u}}$ is a proper convex function that is continuous on $[0, \infty]$
and differentiable on $(0, \infty)$; see~\cite[Thm.~24.1]{Rock}.

\textit{Continuity of the derivative.}  The
expression~\eqref{eq:Fx-derivative} already implies that
$\distsubdiff'_{\vct{u}}$ is continuous for $\tau > 0$ because
the projection onto a convex set is continuous~\cite[Thm.~2.26]{RW:98}.
Continuity of the derivative at zero follows from the limit formula for the
right derivative at zero.

\textit{Bound for the derivative.}  Given the formula~\eqref{eq:Fx-derivative}, it is easy to control the derivative when $\tau > 0$:
$$
\abs{ \distsubdiff'_{\vct{u}}(\tau) }
	\leq \frac{2}{\tau} \enorm{ \vct{u} - \vct{\pi}_{\tau S}(\vct{u}) } \enorm{ \vct{\pi}_{\tau S}(\vct{u}) }
	\leq \frac{2}{\tau} \, (\enorm{\vct{u}} + \tau B)(\tau B)
	= 2B \, (\enorm{\vct{u}} + \tau B).
$$
We obtain the estimate for $\tau = 0$ by taking the limit.

\textit{Lipschitz property.}  We obtain the Lipschitz bound~\eqref{eq:Fx'-lip}
from~\eqref{eq:Fx-derivative} after some effort.  Fix $\tau > 0$.
The optimality condition~\cite[Thm.~III.3.1.1]{HUL:93a} for a projection
onto a closed convex set implies that
$$
\ip{ \vct{y} - \vct{\pi}_{\tau S}(\vct{y}) }{ \vct{\pi}_{\tau S}(\vct{y}) }
	\geq \ip{ \vct{y} - \vct{\pi}_{\tau S}(\vct{y}) }{ \vct{\pi}_{\tau S}(\vct{u}) }
	\quad\text{for all $\vct{u}, \vct{y} \in \R^d$.}
$$
As a consequence,
\begin{multline*}
\ip{ \vct{u} - \vct{\pi}_{\tau S}(\vct{u}) }{ \vct{\pi}_{\tau S}(\vct{u}) } -
\ip{ \vct{y} - \vct{\pi}_{\tau S}(\vct{y}) }{ \vct{\pi}_{\tau S}(\vct{y}) }
	\leq \ip{ \big(\vct{u} - \vct{\pi}_{\tau S}(\vct{u})\big)
	- \big(\vct{y} - \vct{\pi}_{\tau S}(\vct{y}) \big) }{ \vct{\pi}_{\tau S}(\vct{u})} \\
	\leq \enormsm{ (\Id - \vct{\pi}_{\tau S})( \vct{u} ) - (\Id - \vct{\pi}_{\tau S})(\vct{y})} \enormsm{ \vct{\pi}_{\tau S}(\vct{u}) }
	\leq \enormsm{ \vct{u} - \vct{y} } \cdot (\tau B).
\end{multline*}
The last relation relies on the fact~\eqref{eq:I-pi} that the map $\Id - \vct{\pi}_{\tau S}$ is nonexpansive.  Reversing the roles of $\vct{u}$ and $\vct{y}$ in the last calculation, we see that
$$
\abs{ \ip{ \vct{u} - \vct{\pi}_{\tau S}(\vct{u}) }{ \vct{\pi}_{\tau S}(\vct{u}) } -
\ip{ \vct{y} - \vct{\pi}_{\tau S}(\vct{y}) }{ \vct{\pi}_{\tau S}(\vct{y}) } }
	\leq (\tau B) \cdot \enormsm{ \vct{u} - \vct{y} }.
$$
Combining this estimate with the expression~\eqref{eq:Fx-derivative} for the derivative,
we reach
$$
\abs{ \distsubdiff'_{\vct{u}}(\tau) - \smash{\distsubdiff'_{\vct{y}}(\tau)} }
	\leq 2B \cdot \enormsm{ \vct{u} - \vct{y} }.
$$
For $\tau = 0$, the result follows when we take the limit as $\tau \downarrow 0$.
\end{proof}

With this result at hand, we are prepared to prove
a lemma that confirms the remaining claims from
Proposition~\ref{prop:sdim-descent}.

\begin{lemma}[Expected distance to a dilated set] \label{slem:distance-conic}
Let $S$ be a nonempty, compact, convex subset of $\R^d$ that does not contain the origin.
Consider the function
$$
\distsubdiff( \tau ) := \Expect{} \big[ \dist^2( \vct{g}, \tau S ) \big]
	= \Expect{}\big[  \distsubdiff_{\vct{g}}(\tau) \big]
\quad\text{for $\tau \geq 0$.}
$$
The function $\distsubdiff$ is strictly convex, continuous at $\tau = 0$, and
differentiable for $\tau > 0$.  It attains its minimum at a unique point.
Furthermore,
\begin{equation} \label{eq:F'=EF_g'}
\distsubdiff'(\tau) = \Expect{} \big[\distsubdiff'_{\vct{g}}(\tau)\big]
\quad\text{for all $\tau \geq 0$.}
\end{equation}
For $\tau = 0$, we interpret $\distsubdiff'(\tau)$ as a right derivative.
\end{lemma}

\begin{proof} These properties will follow from Lemma~\ref{slem:dist-slice}, and we continue using the notation from this result.

\textit{Continuity at zero.}
Imitating the continuity argument in Lemma~\ref{slem:dist-slice}, we find that
$$
\distsubdiff(\eps) - \distsubdiff(0)
	= \Expect \big[ \distsubdiff_{\vct{g}}(\eps) - \distsubdiff_{\vct{g}}(0) \big]
	\leq \Expect \big[2 \enormsm{ \vct{g} } \enormsm{\pi_{\eps S}(\vct{g})}
	+ \enormsm{\vct{\pi}_{\eps S}(\vct{g})}^2 \big]
	\leq 2 \sqrt{d} \cdot (\eps B) + (\eps B)^2
	\to 0
	\quad\text{as $\eps \to 0$.}
$$
This is all that is required.

\textit{Differentiability.}
This point follows from a routine application of the Dominated Convergence Theorem.
Indeed, for every $\tau \geq 0$, the function $\distsubdiff(\tau) = \Expect{}[ \distsubdiff_{\vct{g}}(\tau) ]$ takes a finite value, and
Lemma~\ref{slem:dist-slice} establishes that $\distsubdiff_{\vct{g}}'$
is continuously differentiable.
For each compact interval $I$, the bound~\eqref{eq:Fx'-bdd} ensures that
$$
\Expect{} \sup\nolimits_{\tau \in I} \ \big\vert \distsubdiff_{\vct{g}}'(\tau) \big\vert
	\leq \Expect{} \sup\nolimits_{\tau \in I} \ \big( 2B \, (\enormsm{\vct{g}} + \tau B) \big)
	\leq 2B \sqrt{d} + 2B^2 \left(\sup\nolimits_{\tau \in I} \ \tau \right)
	< \infty.
$$
The convergence theorem now implies that
$\distsubdiff'(\tau) = \tfrac{\diff{}}{\diff{\tau}} \Expect{} [ \distsubdiff_{\vct{g}}(\tau) ]
	= \Expect{} [ \distsubdiff_{\vct{g}}'(\tau) ]$
for all $\tau \geq 0$.

\textit{Convexity.}  The function $\distsubdiff$ is convex for $\tau \geq 0$ because it is an average of functions of the form $\distsubdiff_{\vct{g}}$, each of which is convex.

\textit{Strict convexity.}  We argue by contradiction.  We have shown that $\distsubdiff$ is convex and differentiable.  If $\distsubdiff$ were not strictly convex, its graph would contain a linear segment.
More precisely, there would be numbers $0 \leq \rho < \tau$ and $\eta \in (0,1)$ for which
\begin{equation} \label{eq:strict-contra}
\Expect{}\big[  \distsubdiff_{\vct{g}}\big((\eta \rho + (1-\eta) \tau ) S \big) \big]
	= \Expect \bigg[ \eta \cdot \distsubdiff_{\vct{g}}(\rho) + (1-\eta) \cdot \distsubdiff_{\vct{g}}(\tau) \bigg].
\end{equation}
The convexity of $\distsubdiff_{\vct{g}}$ ensures that, for each $\vct{g}$, the bracket on the right-hand side is no smaller than the bracket on the left-hand side.  Therefore, the relation~\eqref{eq:strict-contra} holds if and only if the two brackets are equal almost surely with respect to the Gaussian measure.  But note that
\begin{multline*}
\distsubdiff_{\vct{0}}\big( \eta \rho + (1-\eta)\tau \big)
	= \dist^2\big(\vct{0}, (\eta \rho + (1-\eta) \tau ) S \big)
	= \big(\eta \rho + (1-\eta) \tau \big)^2 \cdot \inf_{\vct{s} \in S} \ \enormsq{\vct{s}} \\
	< \big(\eta \rho^2 + (1-\eta) \tau^2 \big) \cdot \inf_{\vct{s} \in S} \ \enormsq{\vct{s}}
	= \eta \cdot \dist^2(\vct{0}, \rho S) + (1-\eta) \cdot \dist^2(\vct{0}, \tau S)
	= \eta \cdot \distsubdiff_{\vct{0}}(\rho) + (1-\eta) \cdot \distsubdiff_{\vct{0}}(\tau).
\end{multline*}
The strict inequality depends on the strict convexity of the square, together with the fact that the infimum is strictly positive.
On account of~\eqref{eq:dist-lip}, the squared distance to a convex set is a continuous function, so there is an open ball around the origin where the same relation holds.
That is, for some $\eps > 0$,
$$
\distsubdiff_{\vct{u}}( \eta \rho + (1- \eta) \tau )
	< \eta \cdot \distsubdiff_{\vct{u}}(\rho) + (1-\eta) \cdot \distsubdiff_{\vct{u}}(\tau)
	\quad\text{when $\enorm{\vct{u}} < \eps$.}
$$
This statement contravenes~\eqref{eq:strict-contra}.

\textit{Attainment of minimum.}
The median of the random variable $\enormsm{\vct{g}}$ does not exceed $\sqrt{d}$.  Therefore, when $\tau b \geq \sqrt{d}$, we have
$$
\distsubdiff(\tau) \geq \Expect\big[ \distsubdiff_{\vct{g}}(\tau) \, \big| \,
	\enormsm{\vct{g}} \leq \sqrt{d} \big] \cdot
	\Prob\big\{ \enormsm{\vct{g}} \leq \sqrt{d} \big\}
	\geq \frac{1}{2} \Expect\big[ 
	\big( \tau b - \enormsm{\vct{g}} \big)^2 \, \big| \,
	\enormsm{\vct{g}} \leq \sqrt{d} \big]
	\geq \frac{1}{2} \big(\tau b - \sqrt{d} \big)^2.
$$
The first inequality follows from the law of total expectation, and the second depends on~\eqref{eq:Fx-unbdd}.  In particular, $\distsubdiff(\tau) > \distsubdiff(0) = d$ when $\tau > 2b^{-1}\sqrt{d}$.  Thus, any minimizer of $\distsubdiff$ must occur in the compact interval $\big[0, 2b^{-1}\sqrt{d} \big]$.  Since $\distsubdiff$ is continuous and strictly convex, it attains its minimum at a unique point.
\end{proof}

\subsection{Error bound for descent cone calculations}
\label{app:sharp-descent}

In this section, we prove Theorem~\ref{thm:sharp-descent},
which provides an error bound for Proposition~\ref{prop:sdim-descent}.
We require a standard result concerning the variance
of a Lipschitz function of a standard normal vector.

\begin{fact}[Variance of a Lipschitz function]
\label{fact:gauss-lip}
Let $F : \R^d \to \R$ be a function that is Lipschitz with respect
to the Euclidean norm:
$$
\abs{ F( \vct{u} ) - F(\vct{y}) } \leq M \cdot \enorm{\vct{u} - \smash{\vct{y}}}
\quad\text{for all $\vct{u}, \vct{y} \in \R^d$.}
$$
For a standard normal vector $\vct{g}$ in $\R^d$, we have
\begin{equation} \label{eq:gauss-lip}
\Var\big( F(\vct{g}) \big) \leq M^2.
\end{equation}
\end{fact}

\noindent
Fact~\ref{fact:gauss-lip}
is a consequence of the Gaussian Poincar{\'e} inequality; see~\cite[Thm.~1.6.4]{Bog:98} or~\cite[p.~49]{ledo:01}.

\begin{proof}[Proof of Theorem~\ref{thm:sharp-descent}]

Let $f : \R^d \to \R$ be a norm, and fix a nonzero point $\vct{x} \in \R^d$.
According to~\cite[Ex.~VI.3.1]{HUL:93a}, the subdifferential of the norm satisfies
\begin{equation} \label{eq:norm-subdiff}
S := \partial f(\vct{x}) =
	\big\{ \vct{s} \in \R^d : \ip{\vct{s}}{\vct{x}} = f(\vct{x})
	\quad\text{and}\quad f^\polar(\vct{s}) = 1 \big\},
\end{equation}
where $f^\polar$ is the norm dual to $f$.  Thus, $S$ is nonempty, compact, convex, and it does not contain the origin.

As in Lemmas~\ref{slem:dist-slice} and~\ref{slem:distance-conic}, we introduce the functions
$$
\distsubdiff_{\vct{u}} : \tau \mapsto \dist^2( \vct{u}, \tau S )
\quad\text{and}\quad
\distsubdiff : \tau \mapsto \Expect \big[ \distsubdiff_{\vct{g}}(\tau) \big],
$$
where $\vct{g}$ is a standard normal vector.  Proposition~\ref{prop:sdim-descent} provides the upper bound $\Expect{} \inf_{\tau \geq 0} \ \distsubdiff_{\vct{g}}(\tau) \leq \inf_{\tau \geq 0} \ \distsubdiff(\tau)$.  Our objective is to develop a reverse inequality.

We establish the result by linearizing each function $\distsubdiff_{\vct{g}}$ around a suitable point.  Lemma~\ref{slem:distance-conic} shows that the function $\distsubdiff$ attains its minimum at a unique location, so we may define
$$
\tau_{\star} := \argmin_{\tau \geq 0} \ \distsubdiff(\tau).
$$
Similarly, for each $\vct{u} \in \R^d$, Lemma~\ref{slem:dist-slice} shows that $\distsubdiff_{\vct{u}}$ attains its minimum at some point $\tau_{\vct{u}} \geq 0$.  For the moment,
it does not matter how we select this minimizer.
Since $\distsubdiff_{\vct{u}}$ is convex and differentiable, we can bound its minimum value below using the tangent at $\tau_{\star}$.  That is,
$$
\inf_{\tau \geq 0}\ \distsubdiff_{\vct{u}}(\tau)
	= \distsubdiff_{\vct{u}}(\tau_{\vct{u}})
	\geq \distsubdiff_{\vct{u}}(\tau_{\star}) + (\tau_{\vct{u}} -\tau_{\star}) \cdot 
	\distsubdiff_{\vct{u}}'(\tau_{\star}).
$$
Should $\tau_{\star} = 0$, we interpret $\distsubdiff'_{\vct{u}}(\tau_{\star})$ as a right derivative.  Replacing $\vct{u}$ by the random vector $\vct{g}$ and
taking the expectation, we reach \begin{align} \label{eq:Fg-lower}
\Expect{} \bigg[ \inf_{\tau \geq 0}\ \distsubdiff_{\vct{g}}(\tau) \bigg]
	&\geq \Expect{} \big[ \distsubdiff_{\vct{g}}(\tau_{\star}) \big] +
	\Expect\big[ (\tau_{\vct{g}} - \tau_{\star}) \cdot \distsubdiff_{\vct{g}}'(\tau_{\star}) \big] \notag \\
	&= \distsubdiff(\tau_{\star}) +
	\Expect{}\big[ \big( \tau_{\vct{g}} - \tau_{\star} \big) \cdot
	\big( \distsubdiff'_{\vct{g}}(\tau_{\star}) - \Expect{} \big[ \distsubdiff'_{\vct{g}}(\tau_{\star}) \big] \big) \big]
	+ \Expect [ \tau_{\vct{g}} - \tau_{\star} ] \cdot \Expect\big[ \distsubdiff'_{\vct{g}}(\tau_{\star}) \big] \notag \\
	&= \distsubdiff(\tau_{\star}) +
	\Expect{}\big[ \big( \tau_{\vct{g}} - \Expect[ \tau_{\vct{g}}] \big) \cdot
	\big( \distsubdiff'_{\vct{g}}(\tau_{\star}) - \Expect{} \big[ \distsubdiff'_{\vct{g}}(\tau_{\star}) \big] \big) \big]
	+ \Expect [ \tau_{\vct{g}} - \tau_{\star} ] \cdot \distsubdiff'(\tau_{\star}) \notag \\
	&\geq \inf_{\tau \geq 0}\ \distsubdiff(\tau) - \big( \Var[ \tau_{\vct{g}} ] \cdot
	\Var\big[ \distsubdiff'_{\vct{g}}(\tau_{\star}) \big] \big)^{1/2}
	+ \Expect{} [\tau_{\vct{g}} - \tau_{\star} ] \cdot \distsubdiff'(\tau_{\star}).
\end{align}
To reach the second line, we add and subtract the constant $\Expect\big[ \distsubdiff_{\vct{g}}'(\tau_{\star}) \big]$.  Next we use the fact that the second parenthesis has zero mean to replace the constant $\tau_{\star}$ in the first parenthesis by the constant $\Expect[ \tau_{\vct{g}} ]$.  Then we invoke \eqref{eq:F'=EF_g'} to identify the derivative of $\distsubdiff$.  The last inequality depends on Cauchy--Schwarz.  From here, we obtain the conclusion~\eqref{eq:sharp-descent} as soon as we estimate the two error terms.  The advantage of this formulation is that the Lipschitz properties of the random variables allow us to control their variances.

First, let us demonstrate that the last term on the right-hand side of~\eqref{eq:Fg-lower} is nonnegative.  Abbreviate $e_1 := \Expect{} [ \tau_{\vct{g}} - \tau_{\star} ] \cdot \distsubdiff'(\tau_{\star})$.
There are two possibilities to consider.  When $\tau_{\star} > 0$,
the derivative $\distsubdiff'(\tau_{\star}) = 0$ because $\tau_{\star}$ minimizes $\distsubdiff$.  Thus, $e_1 = 0$.   On the other hand, when $\tau_{\star} = 0$, it must
be the case that the right derivative $\distsubdiff'(\tau_{\star}) \geq 0$,
or else the minimum of the convex
function $\distsubdiff$ would occur at a strictly positive value.  Since $\tau_{\vct{g}} \geq 0$, we see that the quantity $e_1 \geq 0$.

To compute the variance of $\tau_{\vct{g}}$, we need to devise a consistent method for selecting a minimizer $\tau_{\vct{u}}$ of $\distsubdiff_{\vct{u}}$.  Introduce the closed convex cone $K := \cone(S)$, and notice that
$$
\inf_{\tau \geq 0} \ \distsubdiff_{\vct{u}}(\tau)
	= \inf_{\tau \geq 0} \ \dist^2( \vct{u}, \tau S)
	= \dist^2( \vct{u}, K ).
$$
In other words, the minimum distance to one of the sets $\tau S$ is attained at the point $\Proj_K(\vct{u})$.  As such, it is natural to pick a minimizer $\tau_{\vct{u}}$ of $\distsubdiff_{\vct{u}}$ according to the rule
\begin{equation} \label{eq:tau-g}
\tau_{\vct{u}} := \inf\big\{ \tau \geq 0 : \Proj_{K}(\vct{u}) \in \tau S \big\}
	= \frac{\ip{ \Proj_K(\vct{u}) }{ \vct{x} }}{f(\vct{x})}.
\end{equation}
The latter identity follows from the expression~\eqref{eq:norm-subdiff} for the subdifferential $S$.  In light of~\eqref{eq:tau-g},
$$
\abs{ \tau_{\vct{u}} - \tau_{\vct{y}} }
	=\frac{1}{f(\vct{x})} \absip{ \Proj_K(\vct{u}) - \Proj_K(\vct{y}) }{ \vct{x} }
	\leq \frac{\enorm{\vct{x}}}{f(\vct{x})} \cdot \enormsm{ \Proj_K(\vct{u}) - \Proj_K(\vct{y}) }
	\leq \frac{\enorm{\vct{x}}}{f(\vct{x})} \cdot \enormsm{ \vct{u} - \vct{y} }.
$$
We have used the fact~\eqref{eq:I-pi} that the projection onto a closed convex set is nonexpansive.  Fact~\ref{fact:gauss-lip} delivers
\begin{equation} \label{eq:e2-tau}
\big( \Var( \tau_{\vct{g}} ) \big)^{1/2}
	\leq \frac{ \enormsm{\vct{x}} }{ f(\vct{x}) }
	= \frac{1}{ f(\vct{x} / \norm{\vct{x}}) }.
\end{equation}

Finally, let us turn to the remaining variance term in~\eqref{eq:Fg-lower}.  We have already computed the Lipschitz bound we need for the analysis.  Indeed, the inequality~\eqref{eq:Fx'-lip} states that
$$
\abs{ \distsubdiff_{\vct{u}}'(\tau_{\star}) - \smash{\distsubdiff_{\vct{y}}'(\tau_{\star})} }
	\leq \big(2 \, \sup\nolimits_{\vct{s} \in S} \ \enorm{\vct{s} } \big)
	\cdot \enormsm{\vct{u}-\vct{y}}.
$$
Another invocation of Fact~\ref{fact:gauss-lip} delivers the estimate
\begin{equation} \label{eq:e1-Fg'}
\big( \Var\big( \distsubdiff_{\vct{g}}'(\tau_{\star}) \big) \big)^{1/2}
	\leq 2 \, \sup\nolimits_{\vct{s} \in S} \ \enorm{\vct{s}}.
\end{equation}

To complete the proof, we combine the inequalities~\eqref{eq:Fg-lower},~\eqref{eq:e2-tau},~\eqref{eq:e1-Fg'}, and the fact that $e_1 \geq 0$.  This is the advertised result~\eqref{eq:sharp-descent}.
\end{proof}

\section{Statistical dimension calculations}
\label{sec:stat-dim-app}

This appendix contains the details of the calculations
of the statistical dimension for several families of
convex cones: circular cones, $\ell_1$ descent cones,
and Schatten 1-norm descent cones.

\subsection{Circular cones}
\label{app:circ-cone}

First, we approximate the statistical dimension of a circular cone.
 
\begin{proof}[Proof of Proposition~\ref{prop:circ-cone}]
We begin with an exact integral expression for the statistical dimension
of the circular cone $C = \Circ_d(\alpha)$.
The spherical formulation~\eqref{eq:sdim-circ-expect}
of the statistical dimension asks us to average the squared norm
of the projection of a random unit vector $\vct{\theta}$ onto the cone.
Introduce the angle $\beta := \beta(\vct{\theta}) := \arccos(\theta_1)$
between $\vct{\theta}$ and the first standard basis vector $(1, 0, \dots, 0)$.
Elementary trigonometry shows that the squared norm of the projection of $\vct{\theta}$ onto the cone $C$ admits the expression
$$
F( \beta ) := \enormsm{ \Proj_C(\vct{\theta}) }^2 =
\begin{cases}
	1, & 0 \leq \beta < \alpha, \\
	\cos^2(\beta - \alpha), & \alpha \leq \beta < \tfrac{\pi}{2} + \alpha, \\
	0, & \tfrac{\pi}{2} + \alpha \leq \beta \leq \pi.
\end{cases}
$$
To obtain the exact statistical dimension $\delta(C)$ from~\eqref{eq:sdim-circ-expect},
we integrate $F(\phi)$ in polar coordinates in the usual way
(cf.~\cite[Lem.~6.5.1]{scwe:08}) to find
\begin{equation} \label{eq:sdim-circ}
\delta(C) = d \cdot \frac{\Gamma\big(\half d \big)}{\sqrt{\pi} \, \Gamma\big(\half (d-1)\big)}
\int_0^\pi \sin^{d-2}(\beta) \, F(\beta) \idiff{\beta}.
\end{equation}
We can approximate the integral by a routine application of Laplace's
method~\cite[Lem.~6.2.3]{AF:03}, which yields 
$$
\int_0^\pi \sin^{d-2}(\beta) \, F(\beta) \idiff{\beta}
	= \sqrt{\frac{2\pi}{d}} \, F\!\left(\frac{\pi}{2}\right) + O(d^{-3/2}).
$$
To simplify the ratio of gamma functions, recall Gautschi's inequality~\cite[Sec.~5.6.4]{Olver:2010:NMHF}:
$$
\sqrt{d-2} < \frac{\sqrt{2} \, \Gamma\big(\half d \big)}{\Gamma \big(\half (d-1) \big)} < \sqrt{d}.
$$
Combine the last three displays to reach the expression~\eqref{eq:sdim-cone-est}.

To obtain the more refined estimate $\cos(2\alpha)$ for the error term,
one may use the fact that
the intrinsic volumes of a circular cone satisfy
\begin{equation} \label{eq:circ-vols}
v_k\big( \Circ_d(\alpha) \big) =
	\frac{1}{2} { {\half (d-2)} \choose {\half (k-1)} }
	\sin^{k-1}(\alpha) \cos^{d-k-1}(\alpha)
	\quad\text{for $k = 1, \dots, d-1$.}
\end{equation}
This formula is drawn from~\cite[Ex.~4.4.8]{am:thesis}.  We are using the analytic extension to define the binomial coefficient.  The easiest way to study this sequence is to observe the close connection with the density of a binomial random variable and to apply the interlacing result, Proposition~\ref{prop:interlacing}.  In the interest of brevity, we omit the details.  
\end{proof}

\subsection{Descent cones of the \texorpdfstring{$\ell_1$}{l1} norm}
\label{app:l1-sdim}

In this section, we show how to use Recipe~\ref{fig:desc-recipe}
and Theorem~\ref{thm:sharp-descent} to compute the statistical dimension of the
descent cone of the $\ell_1$ norm at a sparse vector.  This is a warmup for the
more difficult, but entirely similar, calculation in the next section.

\begin{proof}[Proof of Proposition~\ref{prop:l1-sdim}]
Since the $\ell_1$ norm is invariant under signed permutations,
we may assume that the sparse vector $\vct{x} \in \R^d$ takes the
form $\vct{x} = (x_1, \dots, x_s, 0, \dots, 0)$,
where $x_i > 0$.  To compute $\delta\big( \Desc(\pnorm{1}{\cdot}, \vct{x}) \big)$,
we use the subdifferential bound~\eqref{eq:sdim-descent}
for the statistical dimension of a descent cone:
\begin{equation} \label{eq:sdim-descent-l1}
\delta\big( \Desc( \pnorm{1}{\cdot}, \vct{x} ) \big)
	\leq \inf_{\tau \geq 0} \ \Expect \big[ \dist^2(\vct{g}, \tau \cdot \partial \pnorm{1}{\vct{x}}) \big].
\end{equation}
Observe that the subdifferential of the $\ell_1$ norm at $\vct{x}$ has
the following structure:
\begin{equation} \label{eq:l1-subdiff}
\vct{u} \in \partial \pnorm{1}{\vct{x}}
\quad\Longleftrightarrow\quad
\begin{cases}
	u_i = 1, & i = 1, \dots, s \\
	\abs{u_i} \leq 1, & i = s+1, \dots, d.
\end{cases}
\end{equation}
We can compute the distance from a standard normal vector $\vct{g}$ to the dilated subdifferential as follows.
$$
\dist^2\big( \vct{g}, \tau \cdot \partial \pnorm{1}{\vct{x}} \big)
	= \sum_{i=1}^s (g_i - \tau)^2
	+ \sum_{i=s+1}^d \pos^2(\abs{\smash{g_i}} - \tau),
$$
where $\pos(a) := a \vee 0$ and the operator $\vee$ returns the maximum of two numbers.  Indeed, we always suffer an error in the first $s$ components, and we can always reduce the magnitude of the other components by the amount $\tau$.  Taking the expectation, we reach
\begin{equation} \label{eq:l1-sdim-pre}
\Expect \big[ \dist^2\big( \vct{g}, \tau \cdot \partial\pnorm{1}{\vct{x}} \big) \big]
	= s \, \big(1 + \tau^2\big) + (d-s) \, \sqrt{\frac{2}{\pi}} \int_\tau^\infty (u-\tau)^2 \, \econst^{-u^2/2} \idiff{u}.
\end{equation}
Introduce this expression into~\eqref{eq:sdim-descent-l1} and normalize by the ambient dimension $d$ to reach
\begin{equation} \label{eq:l1-sdim-copy}
\frac{\sdim\big( \Desc(\pnorm{1}{\cdot}, \vct{x}) \big)}{d}
	\leq \inf_{\tau \geq 0} \ \left\{ (s/d) \big(1 + \tau^2 \big)  + \left(1 -  s/d \right)
	\sqrt{\frac{2}{\pi}}  \int_\tau^\infty (u - \tau)^2 \, \econst^{-u^2/2}\idiff{u} \right\}.
\end{equation}
This expression matches the upper bound in~\eqref{eq:l1-sdim}.

Now, we need to invoke the error estimate, Theorem~\ref{thm:sharp-descent}.  An inspection of~\eqref{eq:l1-subdiff} shows that the subdifferential $\partial \pnorm{1}{\vct{x}}$
depends on the number $s$ of nonzero entries in $\vct{x}$ but
not on their magnitudes.  It follows from~\eqref{eqn:desc-polar} that, up to isometry, the descent cone $\Desc(\pnorm{1}{\cdot}, \vct{x})$ only depends on the sparsity $s$.
Therefore, we may as well assume that
$\vct{x} = (1, \dots, 1, 0, \dots, 0)$.  For this vector,
$\pnorm{1}{\vct{x} / \enorm{\vct{x}}} = \sqrt{s}$.
Second, the expression~\eqref{eq:l1-subdiff} for the subdifferential shows that $\enorm{\vct{u}} \leq \sqrt{d}$ for every subgradient $\vct{u} \in \partial \norm{\vct{x}}$.  Therefore, the error in the inequality~\eqref{eq:l1-sdim-copy} is at most $2 \sqrt{d/s}$.  We reach the lower bound in~\eqref{eq:l1-sdim}.

Finally, Lemma~\ref{slem:distance-conic} shows that the brace in~\eqref{eq:l1-sdim-copy} is a strictly convex, differentiable function of $\tau$ with a unique minimizer.  It can be verified that the minimum does not occur at $\tau = 0$.  Therefore, we determine the stationary equation~\eqref{eq:l1-stationary} by setting the derivative of the brace to zero and simplifying.
\end{proof}

\subsection{Descent cones of the Schatten 1-norm}
\label{sec:feas-cone-schatt}

Now, we present the calculation of the statistical dimension of the
descent cone of the Schatten 1-norm at a low-rank matrix.  The approach
is entirely similar with the argument in Appendix~\ref{app:l1-sdim}.

\begin{proof}[Proof of Proposition~\ref{prop:S1-sdim}]
Our aim is to identify the statistical dimension of the descent cone of the Schatten 1-norm at a fixed low-rank matrix.  The argument here parallels the proof of Proposition~\ref{prop:l1-sdim}, but we use classical results from random matrix theory to obtain the final expression.  Our asymptotic theory demonstrates that this simplification still results in a sharp estimate.

We begin with the fixed-dimension setting.  Consider an $m \times n$ real matrix $\mtx{X}$ with rank $r$.  Without loss of generality, we assume that $m \leq n$ and $0 < r \leq m$.  The Schatten 1-norm is unitarily invariant, so we can also assume that $\mtx{X}$ takes the form
$$
\mtx{X} = \begin{bmatrix} \mtx{\Sigma} & \mtx{0} \\ \mtx{0} & \mtx{0} \end{bmatrix}
\quad\text{where}\quad
\mtx{\Sigma} = \operatorname{diag}( \sigma_1, \sigma_2, \dots, \sigma_r )
\quad\text{and}\quad
\text{$\sigma_i > 0$ for $i = 1, \dots, r$.}
$$
The subdifferential bound~\eqref{eq:sdim-descent} for the statistical dimension of a descent cone states that
\begin{equation} \label{eq:sdim-descent-S1}
\delta\big( \Desc( \snorm{\cdot}, \mtx{X} ) \big)
	\leq \inf_{\tau \geq 0} \ \Expect
	\big[ \dist^2\big( \mtx{G}, \tau \cdot \partial \snorm{\mtx{X}} \big) \big],
\end{equation}
where we compute distance with respect to the Frobenius norm.  The $m \times n$ matrix
$\mtx{G}$ has independent standard normal entries, and it is partitioned conformally with $\mtx{X}$:
$$
\mtx{G} = \begin{bmatrix}
	\mtx{G}_{11} & \mtx{G}_{12} \\ \mtx{G}_{21} & \mtx{G}_{22}
\end{bmatrix}
	\quad\text{where \quad $\mtx{G}_{11}$ is $r \times r$ \quad and \quad $\mtx{G}_{22}$ is $(m-r)\times(n-r)$.}
$$
According to~\cite[Ex.~2]{Wat:92}, the subdifferential of the Schatten 1-norm at $\mtx{X}$ takes the form
\begin{equation} \label{eq:S1-subdiff}
\partial \snorm{\mtx{X}} = \left\{
	\begin{bmatrix} \Id_r & \mtx{0} \\ \mtx{0} & \mtx{W} \end{bmatrix} :
	\sigma_1(\mtx{W}) \leq 1 \right\},
\end{equation}   
where $\sigma_1( \mtx{W} )$ denotes the maximum singular value of $\mtx{W}$.
It follows that
$$
\dist\big( \mtx{G}, \tau \cdot \partial \snorm{\mtx{X}} \big)^2
	= \fnormsq{ \begin{bmatrix} \mtx{G}_{11} - \tau \, \Id_r & \mtx{G}_{12} \\
	\mtx{G}_{21} & \mtx{0} \end{bmatrix} }
	+ \inf_{\sigma_1(\mtx{W}) \leq 1} \ \fnormsq{ \mtx{G}_{22} - \tau \mtx{W} }.
$$
Using the Hoffman--Wielandt Theorem~\cite[Cor.~7.3.8]{HJ:85}, we can derive
$$
\inf_{\sigma_1(\mtx{W}) \leq 1} \ \fnormsq{ \mtx{G}_{22} - \tau \mtx{W} }
	= \inf_{\sigma_1(\mtx{W}) \leq 1} \ \sum_{i=1}^{m-r}
	\big( \sigma_i(\mtx{G}_{22}) - \tau \, \sigma_i(\mtx{W}) \big)^2
	= \sum_{i=1}^{m-r} \pos^2( \sigma_i(\mtx{G}_{22}) - \tau ),
$$
where $\sigma_i(\cdot)$ is the $i$th largest singular value.
Combining the last two displays and taking the expectation,
\begin{equation} \label{eq:S1-dist-slice}
\Expect \big[ \dist^2\big( \mtx{G}, \tau \cdot \partial \snorm{\mtx{X}} \big) \big]
	= r \big(m + n -r + \tau^2 \big) + \Expect \left[ \sum_{i=1}^{m-r}
	\pos^2(\sigma_i(\mtx{G}_{22}) - \tau) \right].
\end{equation}
Introduce this expression into~\eqref{eq:sdim-descent-S1}:
\begin{equation} \label{eq:S1-sdim-pre}
\delta\big( \Desc\big( \snorm{\cdot}, \mtx{X} \big) \big)
	\leq \inf_{\tau \geq 0} \ \left\{ r (m + n -r + \tau^2 )
	+ \Expect \left[ \sum_{i=1}^{m-r} \pos^2(\sigma_i(\mtx{G}_{22}) - \tau) \right] \right\}.
\end{equation}
We reach a nonasymptotic bound on the statistical dimension.

Next, we apply the error bound, Theorem~\ref{thm:sharp-descent}.  The expression~\eqref{eq:S1-subdiff} shows that the subdifferential $\partial \snorm{\mtx{X}}$ depends only on the rank of the matrix $\mtx{X}$ and its dimension, so the descent cone $\Desc(\snorm{\cdot}, \mtx{X})$ has the same invariance.
Therefore, we may consider the $m \times n$ rank-$r$ matrix $\mtx{X} = \Id_r \oplus \mtx{0}$, which verifies $\snorm{ \mtx{X} / \norm{\mtx{X}} } = \sqrt{r}$.  Each subgradient $\mtx{Y} \in \partial \snorm{\mtx{X}}$ satisfies the norm bound $\fnorm{ \mtx{Y} } \leq \sqrt{m}$.  We conclude that the error in~\eqref{eq:S1-sdim-pre} is no worse than $2 \sqrt{m/r}$.

It is challenging to evaluate the formula~\eqref{eq:S1-sdim-pre} exactly.  In principle, we could accomplish this task using the joint singular value density~\cite[p.~534]{And:84} of the Gaussian matrix $\mtx{G}_{22}$.  Instead, we set up a framework in which we can use classical random matrix theory to obtain a sharp asymptotic result.

Consider an infinite sequence $\{ \mtx{X}(r,m,n) \}$ of matrices, where $\mtx{X}(r,m,n)$ has rank $r$ and dimension $m \times n$ with $m \leq n$.  For simplicity, we assume that the problem parameters $r, m, n \to \infty$ with constant ratios $r/m =: \rho \in (0,1)$ and $m/n =: \nu \in (0,1]$.  The general case follows from a continuity argument.
After a change of variables $\tau \mapsto \tau \sqrt{n-r}$ and a rescaling, the expression~\eqref{eq:S1-dist-slice} leads to
\begin{multline} \label{eq:S1-dist-slice-2}
\frac{1}{mn} \Expect \big[ \dist^2\big( \mtx{G}, \tau \sqrt{n-r} \cdot \partial \snorm{\mtx{X}(r,m,n)} \big) \big] \\
	= \rho \nu +  \rho (1 - \rho \nu) \big(1 + \tau^2 \big) + (1 - \rho)(1 - \rho \nu) \cdot
	\Expect \left[ \frac{1}{m-r} \sum_{i=1}^{m-r} \pos^2( \sigma_i( \mtx{Z} ) - \tau )
	\right].
\end{multline}
Here, $\mtx{G}$ is an $m \times n$ standard normal matrix.  The matrix $\mtx{Z}$ has dimension $(m-r) \times (n-r)$, and its entries are independent $\normal(0, (n-r)^{-1})$ random variables.  

Observe that the expectation in~\eqref{eq:S1-dist-slice-2} can be viewed as a spectral function of a Gaussian matrix.  We can obtain the limiting value of this expectation from a variant of the Mar{\v c}enko--Pastur Law~\cite{MP:67}.

\begin{fact}[Spectral functions of a Gaussian matrix] \label{fact:mp}
Fix a continuous function $F : \R_+ \to \R$.
Suppose $p, q \to \infty$ and $p/q \to y \in (0, 1]$.
Let $\mtx{Z}_{pq}$ be a $p \times q$ matrix with independent $\normal(0, q^{-1})$ entries.  Then
$$
\Expect \left[ \frac{1}{p} \sum_{i=1}^p F\big(\sigma_i(\mtx{Z}_{pq}) \big) \right]
	\to \int_{a_-}^{a_+} F(u) \cdot \phi_y(u) \idiff{u}.
$$
The limits $a_{\pm} := 1 \pm \sqrt{y}$.
The kernel $\phi_y$ is a probability density supported on $[a_-, a_+]$:
$$
\phi_y(u) := \frac{1}{\pi y u} \sqrt{(u^2-a_-^2)(a_+^2-u^2)}
\quad\text{for $u \in [a_-, a_+]$.}
$$
\end{fact}

\noindent
Fact~\ref{fact:mp} is usually stated differently in the literature.  The result here follows from the almost sure weak convergence of the empirical spectral density of a sample covariance matrix to the Mar{\v c}enko--Pastur density~\cite[Thm.~3.6]{BS10:Spectral-Analysis} and the almost sure convergence of the extreme eigenvalues of a sample covariance matrix~\cite[Thm.~5.8]{BS10:Spectral-Analysis}, followed by a change of variables in the integral.  We omit the uninteresting details of this reduction.

Let us apply Fact~\ref{fact:mp} to our problem.  The limiting aspect ratio $y$ of the matrix $\mtx{Z}$ satisfies
$$
y = \frac{m-r}{n-r} = \frac{\nu(1- \rho)}{1-\rho \nu}.
$$
As $r, m, n \to \infty$, we obtain the limit, pointwise in \(\tau\ge 0\),
$$
\Expect \left[ \frac{1}{m-r} \sum_{i=1}^{m-r} \pos^2( \sigma_i( \mtx{Z} ) - \tau )
	\right]
	\to \int_{a_-}^{a_+} \pos^2( u - \tau ) \cdot \phi_y(u) \idiff{u}.
$$
Simplifying the latter integral and introducing it into~\eqref{eq:S1-dist-slice-2}, we reach
\begin{multline*}
\inf_{\tau\ge 0}\left\{\frac{1}{mn} \Expect \big[ \dist^2\big( \mtx{G}, \tau \sqrt{n-r} \cdot \partial \snorm{\mtx{X}(r,m,n)} \big) \big]\right\}
	\\ \rightarrow
	\inf_{\tau\ge 0}\left\{\rho \nu + \rho (1 - \rho \nu)\big(1 + \tau^2 \big) + (1 - \rho)(1 - \rho \nu)
	\int_{a_- \vee \tau}^{a_+} (u-\tau)^2 \cdot \phi_y(u) \idiff{u}\right\}.
\end{multline*}
By itself, pointwise convergence does not imply convergence of the infimal values.  The limit above follows from the fact that all of the functions involved are strictly convex.  For the details, see~\cite[p.~105]{McC:13}.

Rescaling the error estimate for~\eqref{eq:S1-sdim-pre}, we see that the error in the normalized statistical dimension is at most $2/(n \sqrt{mr})$, which converges to zero as the parameters grow.  We obtain the asymptotic result
\begin{equation*}
\frac{1}{mn} \delta\big( \Desc\big(\snorm{\cdot}, \mtx{X}(r,m,n) \big) \big)
	\rightarrow
	\inf_{\tau \geq 0} \ \left\{ \rho \nu + \rho  (1 - \rho \nu) \big(1 + \tau^2 \big)
	+ (1 - \rho)(1- \rho\nu)
	\int_{a_- \vee \tau}^{a_+} (u-\tau)^2 \cdot \phi_y(u) \idiff{u} \right\}.
\end{equation*}
This is the main conclusion~\eqref{eq:S1-sdim}.
To obtain the stationary equation~\eqref{eq:S1-stationary}, we differentiate the brace with respect to $\tau$ and set the derivative to zero.
\end{proof}

\subsection{Permutahedra and finite reflection groups}
\label{sec:chamb-finite-refl}
In this section, we use a deep connection between conic geometry and classical combinatorics to compute the statistical dimension of the normal cone of a (signed) permutahedron.  This computation is based on the intrinsic characterization of statistical dimension in Definition~\ref{def:sdim-int}, which is also restated in~\eqref{eq:sdim-intr}.

A \term{finite reflection group} is a finite subgroup $\coll{G}$ of the orthogonal group\footnote{The orthogonal group consists of the set of orthogonal matrices equipped with the group operation of matrix multiplication.} that is generated by reflections across hyperplanes~\cite{ST:54,Stei:59,CM:72,BB:10}.  Each finite reflection group partitions $\R^d$ into a set $\{ \mtx{U} C_{\coll{G}} : \mtx{U} \in \coll{G} \}$ of polyhedral cones called \term{chambers}. The chambers of the infinite families $A_{d-1}$ and $BC_{d}$ of irreducible finite reflection groups are isometric to the cones
\begin{gather*}
C_A := \big\{ \vct{x} \in \R^d : x_1 \leq \dots \leq x_d \big\}
\quad\text{and}\quad
C_{BC} := \big\{ \vct{x}  \in \R^d : 0 \leq x_1 \leq \dots \leq x_d \big\}.
\end{gather*}
It turns out that the chambers $C_A$ and $C_{BC}$ coincide with the normal cones of certain permutahedra.

\begin{fact}[Normal cones of permutahedra] \label{fact:perm-norm}
Suppose that the vector $\vct{x}$ has distinct entries.  Then the normal cone $\NormC(\mathcal{P}(\vct{x}), \vct{x})$ is isometric to $C_A$ and the normal cone
$\NormC(\mathcal{P}_{\pm}(\vct{x}), \vct{x})$ is isometric to $C_{BC}$.
\end{fact}

\noindent
See~\cite[Sec.~2]{HLT:11} or~\cite[Ex.~7.15]{Z:95} for a proof of Fact~\ref{fact:perm-norm}.

We claim that the statistical dimensions of the chambers $C_A$ and $C_{BC}$ can be expressed as
\begin{equation} \label{eq:frg-sdim}
\delta(C_A) = \mathrm{H}_d
\quad\text{and}\quad
\delta(C_{BC}) = \frac{1}{2} \mathrm{H}_d,
\end{equation}
where $\mathrm{H}_d := \sum_{k=1}^d i^{-1}$ is the $d$th harmonic number.
Proposition~\ref{prop:perm-cone} follows immediately when we combine this
statement with Fact~\ref{fact:perm-norm}.

Let us explain how the theory of finite reflection groups allows us to deduce the expression~\eqref{eq:frg-sdim} for the statistical dimension of the chambers.  First, it follows from~\cite{BZ:09} and the characterization~\cite[Eq.~(6.50)]{scwe:08} of intrinsic volumes in terms of polytope angles that
$$
v_k( C_{\coll{G}} ) = \frac{\abs{\coll{H}_k}}{\abs{\coll{G}}},
$$
where $\coll{H}_k := \{ \mtx{U} \in \coll{G} : \dim(\nullity(\Id - \mtx{U})) = k \}$.  Define the generating polynomial of the intrinsic volumes
$$
q_{\coll{G}}(s) := \sum_{k=0}^d v_k(C_{\coll{G}}) \, s^k
	= \frac{1}{\abs{\coll{G}}} \sum_{k=0}^d \abs{\coll{H}_k} \, s^k.
$$
This polynomial is a well-studied object in the theory of finite reflection groups, and it has many applications in conic geometry as well~\cite[Sec.~4.4]{am:thesis}.  For our purposes, we only need the relationships
$$
q_{\coll{G}}(1) = 1
\quad\text{and}\quad
\frac{\diff{q_{\coll{G}}}}{\diff{s}}(1)
	= \sum_{k=1}^d k \, v_k(C_{\coll{G}})
	= \delta(C_{\coll{G}}). 
$$
These points follow immediately from~\eqref{eq:probab-dist} and the intrinsic formulation~\eqref{eq:sdim-intr}
of the statistical dimension.

The roots $\{ \zeta_k : k = 1, \dots, d \}$ of the polynomial $q_{\coll{G}}$ are called the (negative) \term{exponents} of the reflection group~\cite[Sec.~7.9]{CM:72}.  Factoring the generating polynomial, we obtain a concise expression for the statistical dimension:
$$
\delta(C_{\coll{G}}) = \frac{\diff{q_{\coll{G}}}}{\diff{s}}(1)
	= \left( \prod_{k=1}^d \frac{1}{1-\zeta_k} \right) \cdot
	\frac{\diff{}}{\diff{s}} \prod_{k=1}^d (s - \zeta_k) \bigg\vert_{s=1}
	= \sum_{k=1}^d \frac{1}{1-\zeta_k}. 
$$
We can deduce the value of the large parenthesis because of the normalization $q_{\coll{G}}(1) = 1$.  The exponents associated with the groups $A_{d-1}$ and $BC_d$ are collected in~\cite[Tab.~10]{CM:72}, from which it follows immediately that
$$
\delta(C_A) = \sum_{k=1}^d \frac{1}{k}
\quad\text{and}\quad
\delta(C_{BC}) = \frac{1}{2} \sum_{k=1}^d \frac{1}{k}.
$$
This completes the proof of the claim~\eqref{eq:frg-sdim}. 

\begin{remark}[Related work]
The intrinsic volumes of chambers of finite reflection groups, and more generally of polyhedral cones, have appeared in many different contexts.  For example, the papers~\cite{drton2010geometric,klivans2011projection} relate the intrinsic volumes of regions of hyperplane arrangements to the characteristic polynomial of the arrangement. This result can be used to give an alternative derivation of~\eqref{eq:frg-sdim}.
\end{remark}

\section{Technical lemmas for concentration of intrinsic volumes}
\label{sec:proof-backgr-results}

This appendix contains the technical results that undergird the
proof of the result on concentration of intrinsic volumes,
Theorem~\ref{thm:main-conc}, and the approximate kinematic
bound, Theorem~\ref{thm:approx-kinem}.

\subsection{Interlacing of tail functionals}

First, we establish the interlacing inequality for the tail functionals.
This result is a straightforward consequence of the Crofton formula~\eqref{eq:Crofton}.

\begin{proof}[Proof of Proposition~\ref{prop:interlacing}]
Let $L_{d -k + 1}$ be a linear subspace of dimension $d - k + 1$, and
let $L_{d-k}$ be a linear subspace of dimension $d - k$ inside $L_{d-k+1}$.
The Crofton formula~\eqref{eq:Crofton} shows that the half-tail functionals
are weakly decreasing:
$$
2 h_{k+1}(C) = \Prob \big\{ C \cap \mtx{Q} L_{d-k} \neq \{ \vct{0} \} \big\}
	\leq \Prob\big\{ C \cap \mtx{Q} L_{d-k+1} \neq \{ \vct{0} \} \big\}
	= 2 h_k(C),
$$
where the inequality follows from the containment of the subspaces.
We can express the tail functional $t_k$ as the average of the half-tail functionals:
$$
\half t_k(C) = \half \big[ h_k(C) + h_{k+1}(C) \big]. 
$$
Therefore, $2 h_{k}(C) \geq t_k(C) \geq 2 h_{k+1}(C)$.
\end{proof}

\subsection{Bounds for tropic functions}

We continue with the proof of Lemma~\ref{lem:median-mean},
which provides a bound on the tropic functions.
This argument is based on an approximation formula from
the venerable compendium of
Abramowitz \& Stegun~\cite[Sec.~26.5.21]{AS52:Handbook-Mathematical}.

\begin{fact}[Approximation of beta distributions] \label{fact:beta}
Let $X$ be a $\textsc{beta}(a, b)$ random variable.
Assume that $a + b > 6$ and $(a + b - 1)(1 - x) \geq 0.8$.
Define the quantity $y$ via the formula
$$
y = \frac{3 \left[ w_1 \left(1 - \frac{1}{9b}\right) - w_2 \left(1 - \frac{1}{9a}\right)\right]}{\left[ \frac{w_1^2}{b} + \frac{w_2^2}{a} \right]^{1/2}}
\quad\text{where}\quad
w_1 = (bx)^{1/3}
\quad\text{and}\quad
w_2 = (a(1-x))^{1/3}.
$$
Then
$$
\Prob\{ X \leq x \} = \Phi(y) + \eps(x)
\quad\text{with}\quad
\abs{\eps(x)} \leq 5 \cdot 10^{-3}.
$$
The function $\Phi$ represents the cumulative distribution of a standard
normal random variable.
\end{fact}

\begin{proof}[Proof of Lemma~\ref{lem:median-mean}]
We need to show that the tropic function $I_k^d(k/d) \geq 0.3$ for
all integers $k$ and $d$ that satisfy $1 \leq k \leq d - 1$.
(The cases $k = 0$ and $k = d$ are trivial.)
To accomplish this goal, we represent
the tropic function in terms of a beta random variable,
and we apply Fact~\ref{fact:beta} to approximate its value.

Let $X \sim \textsc{beta}(a,b)$ with shape parameters $a = \half k$ and $b = \half (d-k)$.  In particular, $\Expect[ X ] = k/d$.  It is well known~\cite[Sec.~2]{Art:02} that the tropic function can be expressed in terms of this random variable:
$$
1 - I_k^d(k/d) = \Prob\big\{ \enormsm{\Proj_{L_k}(\vct{\theta})}^2 < k/d \big\}
	= \Prob\{ X < \Expect[ X ] \}.
$$
We must show that this probability is bounded above by $0.7$.

First, the mean--median--mode inequality~\cite{vDVW:93} implies
that the mean $k/d$ of $X$ is smaller than the median when
$k \geq \half d$, so the probability is bounded by $0.5$
in this regime.

Second, the cases where $a + b \leq 6$ correspond with the situation
where $d \leq 12$.  We can enumerate the cases where $d \leq 12$ and
$0 \leq k \leq d/2$.  In each case, we verify numerically that
the required probability is less than $0.7$.

It is easy to check that for $a=\half k$, $b=\half (d-k)$, and
$x=k/d$, the inequality $(a+b-1)(1-x) \geq 0.8$ holds when $0<k< \half d$ and $d \geq 12$.  Instantiating the formula from Fact~\ref{fact:beta}
and simplifying, we reach
$$
y = \frac{\sqrt{2}}{3} \cdot \frac{d - 2k}{\sqrt{dk(d-k)}}
	< \frac{\sqrt{2}}{3}.
$$
Indeed, for each $d$, the extremal choice is $k = 1$.
The function $\Phi$ is increasing, so we conclude that
$$
\Prob\{ X \leq k/d \} \leq \Phi\left(\frac{\sqrt{2}}{3}\right) + 5 \cdot 10^{-3}
	< 0.69.
$$
The latter bound results from numerical computation.
\end{proof}

\subsection{The projection of a spherical variable onto a cone}

In this section, we establish Lemma~\ref{lem:concentration},
which controls the probability that a spherical random variable
has an unusually large projection on a cone.
Although the lemma is framed in terms of a spherical random variable,
it is cleaner to derive the result using Gaussian methods.
We require an exponential moment inequality~\cite[Cor.~1.7.9]{Bog:98}
that ultimately depends on the Gaussian logarithmic Sobolev inequality.  

\begin{fact}[Exponential moments for a function of a Gaussian variable] \label{fact:exp-moment}
Suppose that $F : \R^d \to \R$ satisfies $\Expect F(\vct{g}) = 0$.  Assume moreover that $F$ belongs to the Gaussian Sobolev class $H^1(\gamma_d)$, i.e., the squared norm of its gradient is integrable with respect to the Gaussian measure.  Then
\begin{equation} \label{eq:gauss-conc}
\Expect \econst^{\xi F(\vct{g})}
	\leq \bigg( \Expect \econst^{ (\xi/4) \,\enormsm{\nabla F(\vct{g})}^2 } \bigg)^{2\xi/(1-2\xi)}
	\quad\text{when $0 < \xi < \half$.}
\end{equation}
\end{fact}

Using this exponential moment bound, we reach an elegant estimate for the moment generating function of the squared projection of a standard normal vector onto a cone.

\begin{sublemma}[Exponential moment bounds]
Let $K \subset \R^d$ be a closed convex cone.  Then
\begin{equation} \label{eq:exp-mom-bd}
\Expect \econst^{\xi \, (\enormsm{\Proj_K(\vct{g})}^2 - \delta(K))}
	\leq \exp\left(
	\frac{2 \xi^2 \, \delta(K)}{1 - 4 \abs{\xi}} \right)
	\quad\text{for}\quad
	-\tfrac{1}{4} < \xi < \tfrac{1}{4}.
\end{equation}
\end{sublemma}

\begin{proof}
The result is trivial when $\xi = 0$, so we may limit our attention
to the case where the parameter $\xi$ is strictly positive or strictly negative.
First, suppose that $\xi > 0$.  Consider the zero-mean function
\begin{equation} \label{eq:F-defn}
F(\vct{g}) = \enormsm{\Proj_K(\vct{g})}^2 - \delta(K)
\quad\text{with}\quad
\enormsm{\nabla F(\vct{g})}^2 = 4 \enormsm{\Proj_K(\vct{g})}^2.
\end{equation}
The gradient calculation follows from~\eqref{eq:grad-proj},
and it is easy to see that $F \in H^1(\gamma_d)$ because
the projection onto a cone is a contraction.
The exponential moment bound~\eqref{eq:gauss-conc}
delivers the estimate
$$
\Expect \econst^{\xi \, F(\vct{g})}
	\leq \left( \Expect \econst^{\xi \, \enormsm{\Proj_K(\vct{g})}^2} \right)^{2\xi/(1 - 2\xi)}
	= \exp\left( \frac{2\xi^2 \, \delta(K)}{1 - 2\xi} \right) \cdot
	\left( \Expect \econst^{\xi \, F(\vct{g})} \right)^{2\xi/(1-2\xi)}.
$$
The second relation follows when we add and subtract $\xi \, \delta(K)$ in the exponential function.  We have compared the moment generating function of $F(\vct{g})$ with itself.  Solving the relation, we obtain the inequality
$$
\Expect \econst^{\xi \, F(\vct{g})}
	\leq \exp\left( \frac{2\xi^2 \, \delta(K)}{1 - 4\xi} \right)
	\quad\text{for $0 < \xi < \tfrac{1}{4}$.}
$$
This is the bound~\eqref{eq:exp-mom-bd} for the positive range of parameters.

Now, we turn to the negative range of parameters, which requires
a more convoluted argument.  To make the analysis clearer,
we continue to assume that $\xi > 0$, and we write the negation
explicitly.
Replacing $F$ by $-F$, the exponential moment bound~\eqref{eq:gauss-conc} yields
\begin{equation} \label{eq:neg-mgf}
\Expect \econst^{- \xi \, F(\vct{g})}
	= \Expect \econst^{\xi \, (- F(\vct{g}))}
	\leq \left( \Expect \econst^{\xi \, \enormsm{\Proj_{K}(\vct{g})}^2} \right)^{2\xi/(1-2\xi)}.
\end{equation}
This time, we cannot identify a copy of the left-hand side on the right-hand side.
Instead, let us run the moment comparison argument directly on the remaining expectation:
$$
\Expect \econst^{\xi \, \enormsm{\Proj_{K}(\vct{g})}^2}
	= \econst^{\xi \, \delta(K)} \cdot
	\Expect \econst^{\xi \, F(\vct{g})}
	\leq \econst^{\xi \, \delta(K)}
	\left( \Expect \econst^{\xi \,\enormsm{\Proj_{K}(\vct{g})}^2} \right)^{2\xi/(1-2\xi)}.
$$
The last inequality follows from the exponential moment bound~\eqref{eq:gauss-conc},
just as before.  Solving this relation, we obtain
$$
\Expect \econst^{\xi \, \enormsm{\Proj_{K}(\vct{g})}^2}
	\leq \exp\left( \frac{\xi(1 - 2\xi) \, \delta(K)}{1-4\xi} \right)
	\quad\text{for $0 < \xi < \tfrac{1}{4}$.}
$$
Introduce the latter inequality into~\eqref{eq:neg-mgf} to reach
$$
\Expect \econst^{- \xi \, F(\vct{g})} \leq
	\exp\left( \frac{2\xi^2 \, \delta(K)}{1 - 4\xi} \right).
$$
This estimate addresses the remaining part of the parameter range in~\eqref{eq:exp-mom-bd}.
\end{proof}

With this result at hand, we can easily prove the tail bound
for the projection of a spherical random variable onto a cone.

\begin{proof}[Proof of Lemma~\ref{lem:concentration}]
For a parameter $\xi > 0$, the Laplace transform method~\cite[Sec.~2.1]{BLM:13} delivers
$$
\Prob\big\{ d \, \enormsm{\Proj_C(\vct{\theta})}^2 \geq \delta(C)  + \lambda \big\}
	\leq \econst^{- \xi \lambda - \xi \,\delta(C)}
	\cdot \Expect \econst^{\xi d \, \enormsm{\Proj_C(\vct{\theta})}^2 }.
$$
Let $R$ be a chi random variable with $d$ degrees of freedom, independent from $\vct{\theta}$.  Using Jensen's inequality, we can bound the expectation:
$$
\Expect \econst^{\xi d \, \enormsm{\Proj_C(\vct{\theta})}^2 }
	= \Expect \econst^{\xi \, (\Expect R^2) \, \enormsm{\Proj_C(\vct{\theta})}^2}
	\leq \Expect \econst^{\xi \, \enormsm{\Proj_C(R\vct{\theta})}^2}
	= \Expect \econst^{\xi \,\enormsm{\Proj_C(\vct{g})}^2}.
$$
Combining these results, we obtain
\begin{equation} \label{eq:gauss-tail}
\Prob\big\{ d \, \enormsm{\Proj_C(\vct{\theta})}^2 \geq \delta(C) + \lambda \big\}
	\leq \econst^{-\xi \, \lambda} \cdot 
	\Expect \econst^{\xi \, (\enormsm{\Proj_C(\vct{g})}^2 - \delta(C))}.
\end{equation}
Substitute the inequality for the moment generating function~\eqref{eq:exp-mom-bd} with $K = C$ into~\eqref{eq:gauss-tail} to reach
$$
\Prob\big\{ d \, \enormsm{\Proj_C(\vct{\theta})}^2 \geq \delta(C) + \lambda \big\}
	\leq \econst^{-\xi \, \lambda} \cdot
	\exp\left( \frac{2 \xi^2 \, \delta(C)}{1 - 4\xi} \right)
	\quad\text{for}\quad 0 < \xi < \tfrac{1}{4}.
$$
Select $\xi = \lambda / (4 \, \delta(C) + 4 \lambda)$ to determine that
\begin{equation} \label{eqn:tail-bound-part1}
\Prob\big\{ d \, \enormsm{\Proj_C(\vct{\theta})}^2 \geq \delta(C) + \lambda \big\}
	\leq \exp\left( \frac{-\lambda^2/8}{\delta(C) + \lambda} \right).
\end{equation}
This completes the first half of the argument.

The second half of the proof results in an analogous bound with $C$ replaced
by $C^\polar$.  Note that
$$
\Prob\big\{ d \, \enormsm{\Proj_C(\vct{\theta})}^2 \geq \delta(C) + \lambda \big \}
	= \Prob\big\{ d \, \big( \enormsm{\Proj_{C}(\vct{\theta})}^2 - 1\big)
		+ (d - \delta(C)) \geq \lambda \big\}
	= \Prob\big\{ \delta(C^\polar) -
	d \, \enormsm{\Proj_{C^\polar}(\vct{\theta})}^2 \geq  \lambda \big\}.
$$
The second relation follows from the Pythagorean identity~\eqref{eq:pythag} and the complementarity law~\eqref{eq:delta-dist}.  Repeating the Laplace transform argument from above, with $\xi > 0$, we obtain the inequality
\begin{equation} \label{eq:gauss-tail-2}
\Prob\big\{ d \, \enormsm{\Proj_{C}(\vct{\theta})}^2 \geq \delta(C) + \lambda \big\}
	\leq \econst^{-\xi \lambda} \cdot
	\Expect \econst^{- \xi \, (\enormsm{\Proj_{C^\polar}(\vct{g})}^2 - \delta(C^\polar))}.
\end{equation}
Introduce the bound~\eqref{eq:exp-mom-bd} with $K = C^\polar$ into~\eqref{eq:gauss-tail-2} to see that
$$
\Prob\big\{ d \, \enormsm{\Proj_C(\vct{\theta})}^2 \geq \delta(C) + \lambda \big\}
	\leq \econst^{-\xi \, \lambda} \cdot
	\exp\left( \frac{2 \xi^2 \, \delta(C^\polar)}{1 - 4\xi} \right)
	\quad\text{for}\quad 0 < \xi < \tfrac{1}{4}.
$$
Choose $\xi = \lambda/(4 \, \delta(C^\polar) + 4\lambda)$ to reach
\begin{equation} \label{eqn:tail-bound-part2}
\Prob\big\{ d \, \enormsm{\Proj_C(\vct{\theta})}^2 \geq \delta(C) + \lambda \big\}
	\leq \exp\left( \frac{-\lambda^2/8}{\delta(C^\polar) + \lambda} \right).
\end{equation}
Combine the probability bounds~\eqref{eqn:tail-bound-part1} and~\eqref{eqn:tail-bound-part2} and identify the transition width $\omega^2(C)$ to complete the proof.
\end{proof}

\subsection{Tail functionals of a product}

Finally, we argue that the tail functional of a product
cone is controlled by the tail functionals of the two
summands.

\begin{proof}[Proof of Lemma~\ref{lem:prod-tail}]
Let $C, K \subset \R^d$ be closed convex cones.
Define independent random variables $X$ and $Y$
that take values in $\{0,1, \dots,d\}$ and whose
distributions are given by the intrinsic volumes
of the cones $C$ and $K$, respectively.  That is,
$$
\Prob\{ X = k \} = v_k(C)
\quad\text{and}\quad
\Prob\{ Y = k \} = v_k(K)
\quad\text{for $k = 0, 1, 2, \dots, d$.}
$$
According to the rule~\eqref{eq:prod} for the intrinsic volumes of
a product cone,
$$
\Prob\{ X + Y = k \} = v_k(C \times K)
\quad\text{for $k = 0, 1, 2, \dots, 2d$.}
$$
By dint of this identity, we can use probabilistic reasoning to bound the tail functionals of the cone $v_k(C \times K)$.  Indeed, observe that
$$
\Prob\{ X + Y \geq \delta(C) + \delta(K) + 2\lambda \}
	\leq \Prob\{ X \geq \delta(C) + \lambda \}
	+ \Prob\{ Y \geq \delta(K) + \lambda \}.
$$
We can rewrite this inequality in terms of tail functionals:
$$
t_{\lceil \delta(C) + \delta(K) + 2 \lambda \rceil}(C \times K)
	\leq t_{\lceil \delta(C) + \lambda \rceil}(C)
	+ t_{\lceil \delta(K) + \lambda \rceil}(K).
$$
This is the advertised conclusion.
\end{proof}

\section{Statistical dimension and Gaussian width}
\label{app:sdim-width}

This appendix contains a short proof of Proposition~\ref{prop:sdim-width}, which
states that the Gaussian width of a spherical convex set is comparable
with the statistical dimension of the cone generated by the set.

\begin{proof}[Proof of Proposition~\ref{prop:sdim-width}]
Let $C$ be a convex cone in $\R^d$; we may assume $C$ is closed.  It is easy to check that the statistical dimension of $C$ dominates its squared Gaussian width:
$$
w^2(C) := \left( \Expect{} \sup\nolimits_{\smash{\vct{y}} \in C \cap \sphere{d}} \ \ip{ \smash{\vct{y}} }{ \smash{\vct{g}} } \right)^2
\leq \left( \Expect{} \sup\nolimits_{\smash{\vct{y}} \in C \cap \ball{d}} \ \ip{ \smash{\vct{y}} }{ \smash{\vct{g}} } \right)^2
\leq \Expect \bigg[ \bigl(
\sup\nolimits_{\smash{\vct{y}} \in C\cap \ball{d}}
	\ip{ \smash{\vct{y}} }{\smash{\vct{g}}} \big)^2
  \bigg]
  = \delta(C).
  $$
The first inequality holds because we have enlarged the range of the supremum.  Afterward, we invoke Jensen's inequality, and we recognize the supremum form~\eqref{eq:sdim-sup-expect} of the statistical dimension.

For the reverse inequality, define the random variable $Z := Z(\vct{g}) := \sup_{\smash{\vct{y}} \in C \cap \sphere{d-1}} \ \ip{\smash{\vct{y}}}{ \smash{\vct{g}}}$, and note that $w(C) = \Expect Z$.  The function $\vct{g} \mapsto Z(\vct{g})$ is 1-Lipschitz because the supremum is restricted to a subset of the Euclidean unit sphere.  Therefore, we can bound the fluctuation of $Z$ as follows.
\begin{equation} \label{eq:width-var}
\Expect\big[Z^2\big] - w^2(C)
	= \Expect \big[ (Z - \Expect Z)^2 \big]
	= \Var(Z)
	\leq 1.
\end{equation}
The last inequality follows from Fact~\ref{fact:gauss-lip}

As a consequence of~\eqref{eq:width-var}, we obtain the required bound $\delta(C) \leq w^2(C) + 1$ as soon as we verify that $\delta(C) \leq \Expect \big[Z^2\big]$.
Since $Z^2$ is a nonnegative random variable,
$$
\Expect \big[ Z^2 \big]
	\geq \Expect \big[ Z^2 \cdot \mathbb{1}_{\R^d \setminus C^\polar}(\vct{g}) \big]
	= \Expect \left[ \left(
	\sup\nolimits_{\smash{\vct{y}} \in C \cap \sphere{d-1}} \ \ip{ \smash{\vct{y}} }{\smash{\vct{g}}} \right)^2
	\cdot \mathbb{1}_{\R^d \setminus C^\polar}(\vct{g}) \right],
$$
where $\mathbb{1}_E$ denotes the indicator of the event $E$.  We claim that the right-hand side of this inequality equals the statistical dimension $\delta(C)$.  Indeed, for any $\vct{x} \notin C^\polar$, we must have $\sup_{\smash{\vct{y}} \in C \cap \sphere{d-1}} \ \ip{\smash{\vct{y}}}{\vct{x}} = \sup_{\smash{\vct{y}} \in C \cap \ball{d-1}} \ \ip{\smash{\vct{y}}}{\vct{x}}$ because the supremum over the ball occurs at a unit vector.  On the other hand, when $\vct{x} \in C^\polar$, we have the relation $\sup_{\smash{\vct{y}} \in C \cap \ball{d-1}} \ \ip{\smash{\vct{y}}}{\vct{x}} = 0$.  Combine these observations to reach
$$
\Expect \big[ Z^2 \big] \geq
\Expect \left[ \left(
	\sup\nolimits_{\smash{\vct{y}} \in C \cap \sphere{d-1}} \ \ip{ \smash{\vct{y}} }{\smash{\vct{g}}} \right)^2
	\cdot \mathbb{1}_{\R^d \setminus C^\polar}(\vct{g}) \right]
	= \Expect \left[ \left( \sup\nolimits_{\smash{\vct{y}} \in C \cap \ball{d-1}} \ \ip{ \smash{\vct{y}} }{\smash{\vct{g}}} \right)^2 \right].
$$
On account of~\eqref{eq:sdim-circ-expect}, we identify the right-hand side as the statistical dimension $\delta(C)$.
\end{proof}

\section*{Acknowledgments}

DA is with the School of
Mathematics, The University of Manchester.  Research supported by DFG grant AM 386/1-1 and 386/1-2.
  
ML is with the School of Mathematics, The
  University of Manchester. Research supported by Leverhulme Trust
  grant R41617 and a Seggie Brown Fellowship of the University of
  Edinburgh.
    
MBM and JAT~are with the Department of Computing
    and Mathematical Sciences, California Institute of
    Technology. Research supported by ONR awards N00014-08-1-0883 and
    N00014-11-1002, AFOSR award FA9550-09-1-0643,
    and a Sloan Research Fellowship.
    
The authors wish to thank Babak Hassibi and Samet Oymak for
helpful discussions on the connection between phase transitions
and minimax risk.  Jared Tanner provided detailed information
about contemporary research on phase transitions for random
linear inverse problems.

\bibliographystyle{myalpha}
\def\cprime{$'$} \def\cprime{$'$}

\end{document}